\documentclass{ieeeaccess}
\usepackage{graphicx}
\usepackage{amsmath, amssymb}
\usepackage{amsfonts}
\usepackage{mathtools}

\usepackage{color}
\newcommand{\hl}{\textcolor{blue}}
\renewcommand{\hl}{}

\usepackage{amsthm}

\theoremstyle{definition}
\newtheorem{theorem}{Theorem}
\newtheorem*{theorem*}{Theorem}

\newtheorem*{definition*}{Definition}

\newtheorem*{corollary*}{Corollary}

\newtheorem*{example*}{Example}
\newtheorem{proposition}[theorem]{Proposition}
\newtheorem*{proposition*}{Proposition}

\newtheorem*{lemma*}{Lemma}

\newtheorem*{conjecture*}{Conjecture}

\usepackage{algorithm}
\usepackage[noend]{algpseudocode}
\algrenewcommand\algorithmicrequire{\textbf{Input:}}
\algrenewcommand\algorithmicensure{\textbf{Output:}}

\ifCLASSOPTIONcompsoc
    \usepackage[caption=false, font=normalsize, labelfont=sf, textfont=sf]{subfig}
\else
\usepackage[caption=false, font=footnotesize]{subfig}
\fi

\usepackage{multirow}
\renewcommand{\thetable}{\Roman{table}}
\usepackage{lscape}

\usepackage{enumitem}

\usepackage{comment}

\newcommand{\rfig}[1]{Fig.~\ref{#1}}

\begin{document}

\history{Date of current version July 9, 2024.}
\doi{10.1109/ACCESS.2024.3425711}

\title{Toward Practical Benchmarks of Ising Machines: \\
A Case Study on the Quadratic Knapsack Problem}
\author{
        \uppercase{Kentaro~Ohno}\authorrefmark{1}\authorrefmark{2}, 
        \uppercase{Tatsuhiko~Shirai}\authorrefmark{3},~\IEEEmembership{Member,~IEEE},
        \uppercase{Nozomu~Togawa}\authorrefmark{2},~\IEEEmembership{Member,~IEEE}
}
\address[1]{NTT, Tokyo, Japan}
\address[2]{Department of Computer Science and Communications Engineering, Waseda University, Tokyo, Japan}
\address[3]{Waseda Institute for Advanced Study, Waseda University, Tokyo, Japan}

\corresp{Corresponding author: Kentaro~Ohno (e-mail: kentaro.ohno@togawa.cs.waseda.ac.jp).}



\begin{abstract}
    Combinatorial optimization has 
    wide applications from industry to natural science.
    Ising machines bring an emerging computing paradigm for efficiently solving a combinatorial optimization problem by searching a ground state of a given Ising model.
    Current cutting-edge Ising machines achieve fast sampling of near-optimal solutions of the max-cut problem.
    However, for problems with additional constraint conditions, their advantages have been hardly shown due to difficulties in handling the constraints.
    In this work, we focus on
    benchmarks of Ising machines on the quadratic knapsack problem (QKP).
    To bring out their practical performance, 
    we propose 
    fast two-stage post-processing for Ising machines, which makes handling the constraint easier.
    Simulation based on simulated annealing shows that the proposed method substantially improves the solving performance of Ising machines and the improvement is robust to a choice of encoding of the constraint condition.
    \hl{Through evaluation using an Ising machine called Amplify Annealing Engine,
    the proposed method is shown to dramatically improve its solving performance on the QKP.
    These results are a crucial step toward showing advantages of Ising machines on practical problems involving various constraint conditions.
    }
\end{abstract}


\titlepgskip=-21pt

\maketitle

\section{Introduction}\label{sec:introduction}
Combinatorial optimization is an important research area with applications in various fields such as artificial intelligence and operations research.
For example, the knapsack problem and its variants are famous and well-studied combinatorial optimization problems with numerous applications including production planning, resource allocation, and portfolio selection~\cite{cacchiani2022knapsack}.
Theoretically, combinatorial optimization problems are often hard to solve exactly within a reasonable amount of time due to their NP-hardness.
Therefore, various heuristics and meta-heuristics have been developed for dealing with large-scale combinatorial optimization problems.

Ising machines offer a new computing paradigm for tackling hard combinatorial optimization problems~\cite{mohseni2022ising}.
Ising machines search a ground state of a given Ising model, a model in statistical mechanics involving binary variables (called spins) and their interactions, and thus can be used for optimization over binary variables.
For problems with additional constraint conditions on binary variables, 
the \emph{penalty method} is typically used~\cite{lucas2014ising}.
A constraint on binary variables $x=(x_1,\cdots,x_n)$ is translated into a penalty term $H_\mathrm{con}(x)$ added to the objective function $H_\mathrm{obj}(x)$ with a positive coefficient $\lambda>0$ to construct an unconstrained binary optimization problem
\begin{align}\label{eq:penalty_method}
    &\operatorname{minimize} \ H_\mathrm{obj}(x) + \lambda H_\mathrm{con}(x) \\
    &\operatorname{subject \ to} \ x \in \{0,1\}^n, \notag
\end{align}
to which an Ising machine is applied.
There exist several types of Ising machines depending on the way of physical implementation: examples are quantum annealers~\cite{johnson2011quantum,houck2012chip}, coherent Ising machines~\cite{wang2013coherent,honjo2021100}, and specialized-circuit-based digital machines~\cite{yamaoka201520k,aramon2019physics,yamamoto2020statica,wang2021solving,goto2019combinatorial}. 
These machines enable fast sampling of near-optimal solutions on the max cut problem, which is naturally formulated with an Ising model. 

However, for problems with additional constraint conditions on binary variables, the superiority of Ising machines to other methods has not been observed.
For example, previous benchmark results~\cite{codognet2022quantum,huang2022benchmarking, 
parizy2021analysis,ceselli2023good} 
on the quadratic knapsack problem (QKP) and quadratic assignment problem (QAP) show that Ising machines are not competitive with existing (meta-)heuristic solvers.
A critical performance issue is that Ising machines do not necessarily output feasible solutions, i.e., solutions satisfying constraints.
The penalty coefficient $\lambda$ in (\ref{eq:penalty_method}) is required to be large for outputs to be feasible, but large $\lambda$ tends to degrade the objective value.
\hl{This trade-off also makes it difficult to fairly compare the performance of Ising machines with that of other heuristic solvers.
Therefore, it is crucial to resolve the trade-off to establish practical utility of Ising machines.
}

In this study, we focus on the benchmark of Ising machines on the QKP~\cite{gallo1980quadratic}.
The QKP is a well-studied practical problem involving one inequality constraint over binary variables.
\hl{Although it is presumably suitable for solving with Ising machines, existing Ising machine benchmarks~\cite{parizy2021analysis} on the QKP only deal with relatively easy instances that can be solved by exact methods due to the difficulty in handling the constraint for Ising machines.
Therefore, we explore a way to overcome this problem and enable effective performance comparison with existing heuristic solvers.
The application to the QKP is taken as the first attempt in this direction and would be extended to other problems in the future.
}

Several methods to encode an inequality constraint into penalties have been proposed for applying Ising machines to the QKP~\cite{tamura2021performance,jimbo2022hybrid,bontekoe2023translating,yonaga2020solving} ,
since the choice of encoding methods has impacts on controlling the trade-off between the feasibility and objective value.
Nevertheless, none of them have achieved better results than other heuristic solvers or even a simple greedy method.
Moreover, each encoding method has different advantages and disadvantages, making it difficult to select the appropriate method for a given instance.

We take another approach to enhance the performance of Ising machines by exploiting the problem structure.
We propose to incorporate efficient two-stage post-processing into the solving process using an Ising machine.
The post-processing consists of \emph{repair} and \emph{improvement} procedures (\rfig{fig:postprocess_landscape}).
First, the repair procedure converts the output of the Ising machine, if it is infeasible, into a feasible solution.
The obtained feasible solution is improved by a local improvement procedure.
Since Ising machines are suited for global search, the improvement procedure takes a complementary role to achieve further improvement via local search.
\hl{Although the post-processing consists of well-known greedy algorithms~\cite{gallo1980quadratic,chaillou1989best,billionnet1996linear}, the combination with Ising machines has not been fully explored so far.
We believe it is valuable to thoroughly examine the effectiveness of the post-processing approach as it seems to be a promising straightforward way to resolve the technical issue.
}

\begin{figure}[t]
    \centering
    \includegraphics[width=0.47\textwidth]{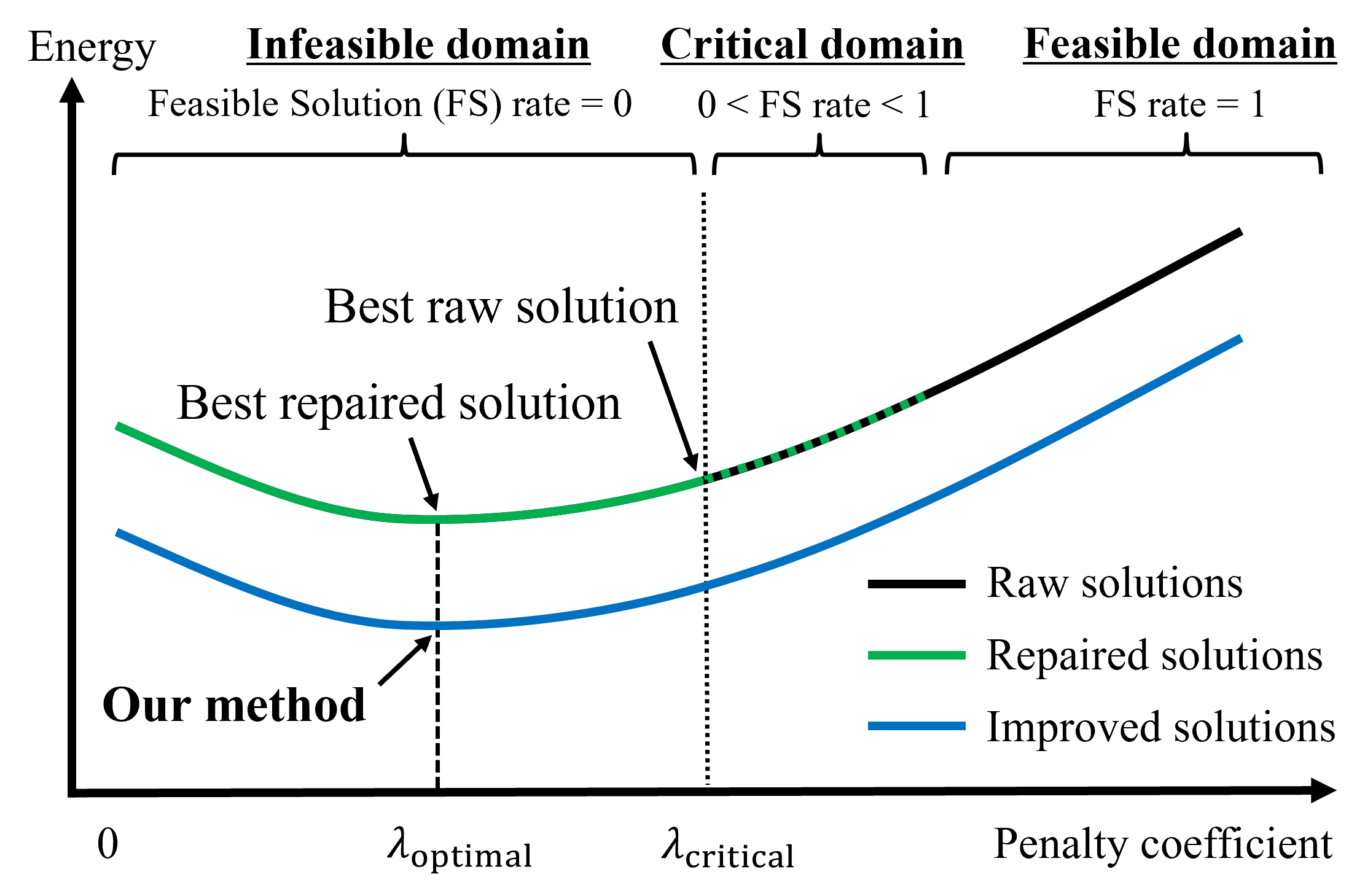}
    \caption{Conceptual figure of 
    effect of 
    two-stage post-processing. 
    ``Raw solutions'' denote outputs of Ising machines, 
    which are often infeasible when penalty coefficient $\lambda$ is small (dashed line on ``Critical domain'').
    Tuning of $\lambda$ typically involves finding $\lambda_\mathrm{critical}$ which achieves best trade-off between feasibility and objective.
    Repair procedure for infeasible solutions enables us to obtain feasible solutions even for smaller $\lambda$. 
    Improvement procedure further enhances feasible solutions with local operations.
    Optimal penalty coefficient $\lambda_\mathrm{optimal}$ is found to be much robust to choice of encoding methods for inequality constraint, in contrast to $\lambda_\mathrm{critical}$ 
    which heavily depends on encoding methods
    (see Section~\ref{sec:simulation}).
    }
    \label{fig:postprocess_landscape}
\end{figure}

We conduct simulation experiments on medium-sized QKP instances using simulated annealing.
The results show that the combined use of the repair and improvement procedures provides the synergistic effect on gaining the solving performance, achieving optimal solutions on more than 80\% of the test instances within a reasonable time.
Besides, we find that the post-processing greatly reduces the dependency of the solving performance on the choice of encoding methods of the inequality constraint into penalties, which might make practical use of Ising machines much easier.

We evaluate the performance of Amplify Annealing Engine (AE)~\cite{amplify}, one of the state-of-the-art Ising machines, with our method on a data set of large QKP instances of size ranging from 1000 to 2000.
AE combined with the post-processing achieves best known solutions on 77.5\% of test instances and a small optimality gap on the rest instances.
This result significantly exceeds the previous benchmark of Ising machines on the QKP~\cite{parizy2021analysis,ceselli2023good}. 
\hl{This is also the first result that an Ising-machine-based solver achieves a performance comparable to previous heuristics on the QKP~\cite{fomeni2014dynamic,yang2013effective,patvardhan2016parallel,chen2017iterated}.}

Our contribution is summarized as follows:
{
\begin{itemize}
    \item We propose a method to solve the QKP with Ising machines combined with the post-processing consisting of the repair and improvement procedures \hl{to overcome the difficulty in handling the constraint condition}.
    \item Through simulation experiments on medium-sized instances, we show that the post-processing is effective in obtaining optimal solutions and making the performance robust to a choice of encoding methods. 
    \item \hl{We show that the proposed method dramatically improves the solving performance of a state-of-the-art Ising machine on the QKP despite its simplicity, exceeding the previous benchmark results of Ising machines.}
\end{itemize}
}

\hl{
Since the proposed method is implemented on the basis of well-known naive algorithms, 
we expect that it can be extended and enhanced to show the advantage of Ising machines over previous heuristic methods on the QKP and other problems in the future.
Therefore, our findings are a crucial step toward showing the practical utility of Ising machines.
}

The rest of the paper is organized as follows.
Backgrounds 
are explained in Section~\ref{sec:preliminaries}.
We introduce the proposed method in Section~\ref{sec:proposed_method}.
The simulation experiment is conducted in Section~\ref{sec:simulation}.
We evaluate the proposed method using the Ising machine in Section~\ref{sec:experiment}.
Related work and future work is discussed in Section~\ref{sec:related_work}.
Section~\ref{sec:conclusion} concludes this paper.

\section{Preliminaries
}\label{sec:preliminaries}

\subsection{Ising machines}
We briefly review backgrounds on Ising machines.
An \emph{Ising model} is a model in statistical mechanics consisting of a number of binary variables $s_i \in \{\pm 1\}, i=1,\cdots,n$ called spins and their interactions.
The \emph{energy} of a state $s=(s_1,\cdots,s_n) \in \{\pm 1\}^n$ is defined as
\begin{align}
    H = \sum_{i,j} J_{ij} s_i s_j + \sum_i h_i s_i,
\end{align}
where $J_{ij} \in \mathbb R$ represents the pairwise interaction between spin $s_i$ and $s_j$ and $h_i\in \mathbb R, i=1,\cdots,n$ is called external field.
A \emph{ground state} of the Ising model is a state $s$ that minimizes the energy $H$.
\emph{Ising machines} implement fast heuristics to search a ground state of the Ising model by analog computation using quantum annealing~\cite{kadowaki1998quantum} or degenerate optical parametric oscillators~\cite{wang2013coherent}, or by digital algorithms such as simulated annealing and simulated bifurcation~\cite{goto2019combinatorial} with massive parallelization.

The problem of finding a ground state of an Ising model can be also formulated as a \emph{quadratic unconstrained binary optimization (QUBO)} problem~\cite{lucas2014ising}, 
which is a class of optimization problems over binary variables $x_i \in \{0,1\},  i=1,\cdots,n$ defined by a square matrix $Q \in \mathbb R^{n\times n}$ as follows:
\begin{align}
    &\operatorname{minimize} \ x^{\top} Q x \\
    &\operatorname{subject \ to} \ x \in \{0,1\}^n.
\end{align}
The objective value $x^{\top} Q x$ is also called the energy of $x$.

\subsection{Quadratic Knapsack Problem}

The quadratic knapsack problem (QKP)~\cite{gallo1980quadratic} is a generalization of the well-known knapsack problem and defined by data of $n$ items and the knapsack capacity $C$.
Each item $i$ is associated with a weight $w_i> 0$ and a profit $p_i\ge 0$.
In addition, for each pair $i,j \ (i<j)$ of items, a pairwise profit $p_{ij} \ge 0$ is defined and it is added to the total profit when both items are put into the knapsack.
The QKP asks to maximize the total profit maintaining the total weight within the knapsack capacity.
Namely, it is formulated as
\begin{align}\label{eq:qkp}
    \operatorname{maximize} \ &H(x) \coloneqq \sum_{i=1}^n p_i x_i + \sum_{i=1}^{n-1} \sum_{j=i+1}^n p_{ij}x_i x_j \notag \\
    \operatorname{subject\ to} \ &\sum_{i=1}^n w_i x_i \le C, \notag \\
    &x_i\in \{0,1\}, i=1,\cdots,n.
\end{align}
We define $p_{ij}\coloneqq p_{ji}$ for $i>j$ to ease notation.
We assume $w_i$ and $C$ are integers and satisfy $\min_i w_i \le C < \sum_{i=1}^n w_i$ to avoid triviality.
\hl{The QKP is an NP-hard optimization problem and the state-of-the-art exact solver can only solve some QKP instances of size up to 1500 in a reasonable time~\cite{pisinger2007solution}.
To solve large QKP instances efficiently, various heuristic approaches including the tabu search~\cite{glover2002solving,yang2013effective}, swarm optimization~\cite{xie2007mini}, dynamic programming~\cite{fomeni2014dynamic}, greedy randomized adaptive search procedure (GRASP)~\cite{yang2013effective}, and evolutionary algorithm~\cite{patvardhan2012novel,patvardhan2016parallel} have been proposed.
The current best heuristic solver is based on the iterated hyperplane exploration approach~\cite{chen2017iterated}.
}

As a particular problem structure, it is well-known that an optimum of the QKP is attained on the edge of the space of feasible solutions.
Precisely, the following holds.
For a proof, we refer to Appendix~\ref{app:proof}.
\begin{proposition}[cf.~\cite{billionnet1996linear}]\label{prop:qkp_property}
    For a QKP instance defined as (\ref{eq:qkp}),
    an optimum is attained by a solution $x \in \{0,1\}^n$ satisfying $C-\max_i w_i < \sum_{i=1}^n w_i x_i \le C$.
\end{proposition}

The QKP can be reformulated as QUBO
in the following way~\cite{glover2002solving}.
First, an integer slack variable $z\ge 0$ is introduced to represent the inequality constraint $\sum_{i=1}^n w_i x_i \le C$ as an equality constraint $\sum_{i=1}^n w_i x_i + z = C$.
By transforming the equality constraint into a penalty term in the standard way, we get a quadratic optimization problem:
\begin{align}\label{eq:qkp_qubo}
    &\operatorname{minimize} \ -H(x) + \lambda H_\mathrm{ineq} (x, z), \\
    &H_\mathrm{ineq}(x,z) = \left(\sum_{i=1}^n w_i x_i + z - C \right)^2,
\end{align}
where $\lambda > 0$ is a sufficiently large positive number.
To further translate it into a QUBO problem, the integer variable $z$ is represented by binary variables typically with binary expansion~\cite{glover2002solving}.
That is, taking sufficiently large integer $D>0$ which is an upper bound of $z$, $z$ is represented by
\begin{align}\label{eq:binary_encoding}
    k &\coloneqq \lfloor \log D \rfloor + 1, \ 
    R \coloneqq D+1 - 2^{k-1} ,\notag \\
    z &= \sum_{i=1}^{k-1} 2^{i-1}y_i + R y_k
\end{align}
using additional binary variables $y_1, \cdots,y_k \in \{0,1\}$.
Other encoding methods of the integer variable are proposed and evaluated for the use of Ising machine (without post-processing)~\cite{tamura2021performance,jimbo2022hybrid,bontekoe2023translating}.
Their performance will be compared in Section~\ref{sec:simulation} under the existence of post-processing.

We remark that a local optimum of the QUBO problem (\ref{eq:qkp_qubo}) does \emph{not} necessarily correspond to that of the QKP (\ref{eq:qkp}).
Recall that a local optimum of an optimization problem over binary variables is defined as the objective value of a feasible solution for which any flip (i.e. changing value from 0 to 1 or 1 to 0) of a variable cannot improve the objective value maintaining feasibility.
For example, we consider a trivial feasible solution $x=(0,\cdots,0)$ which clearly does not attain a local optimum of the QKP.
In the QUBO setting, a solution with $x=(0,\cdots,0)$ and $y$ which gives $z=C$ corresponds to the solution.
In fact, it attains a local minimum of the QUBO problem (\ref{eq:qkp}) for large $\lambda$ since flipping $x_i$ for any $i \in \{1,\cdots,n\}$ leads to a change of the objective value by $-p_i + \lambda w_i^2 > 0$ and similarly flipping $y_i$ for any $i \in \{1,\cdots,k\}$ increases the objective value.
In other words, a flip of $x_i$ in the QKP is realized by multiple flips involving auxiliary variables $y_i$ in the QUBO form.
Hereafter, unless otherwise noted, we use the word ``local'' in the sense of the QKP and not of QUBO.

\subsection{Challenges in Ising Machines Solving QKP}
Since the QKP can be naturally formulated with a quadratic objective function of binary variables as above, it is presumably suited for benchmarks of Ising machines.
However, in contrast to the max-cut problem on which Ising machines have achieved successful results~\cite{honjo2021100,goto2019combinatorial},
even medium-sized QKP instances that can be handled by exact methods are not adequately optimally solved by Ising machines or simulation in the previous studies~\cite{parizy2021analysis,jimbo2022hybrid,bontekoe2023translating}.
The biggest challenge is that Ising machines might output solutions violating the inequality constraint since the constraint is imposed only implicitly with the penalty term.

There is a trade-off that a large penalty is required to obtain feasible solutions with high probability whereas it also degrades the objective value.
As shown in Section~\ref{sec:simulation} below, the recently proposed encoding methods of the inequality constraint~\cite{tamura2021performance,jimbo2022hybrid,bontekoe2023translating} have a role to control this trade-off.
Nevertheless, their improvement in Ising machine performance is not satisfactory, since they are still outperformed by a simple greedy method (see simulation results in Section~\ref{sec:simulation}).
Our approach is to directly resolve the trade-off by incorporating local post-processing into Ising machines, instead of exploring the optimal encoding method.

\section{Proposed Method}\label{sec:proposed_method}

We propose to incorporate post-processing utilizing the problem structure into the solving process with Ising machines.
The post-processing consists of two steps: repair and improvement.
The repair procedure converts an infeasible solution into a feasible solution.
It is commonly used for other meta-heuristics such as evolutionary algorithms~\cite{patvardhan2016parallel,chu1998genetic}.
The improvement procedure takes a feasible solution as an input and improves the objective value by locally modifying the solution.
Both procedures are building blocks of most heuristic combinatorial optimization algorithms, often combined with randomized operations to enable global search~\cite{chen2017iterated,jovanovic2022solving}.
In our case, they are used deterministically (i.e., without randomness) following a greedy strategy, since Ising machines have a role in the global search.
We expect that Ising machines and the local post-processing work complementarily to efficiently enhance the solving performance.
One important advantage of the proposed method is that the repair procedure enables us to set the penalty coefficient $\lambda$ in (\ref{eq:qkp_qubo}) to 
small values and to tune $\lambda$ according to the objective value, not to the rate of feasible solutions, since obtained solutions are always feasible.
This effect, coupled with the local improvement, helps us to obtain the optimal solution more easily with Ising machines, as we will see in Sections~\ref{sec:simulation} and \ref{sec:experiment}.
We explain the details of the method below.

\subsection{Post-processing Algorithm on QKP}

\begin{algorithm}[t]
\small
 \caption{Post-processing on QKP}\label{alg:repair_solution}
 \begin{algorithmic}[1]
 \Require{Solution $x=(x_1,\cdots,x_n) \in \{0,1\}^n$ (possibly infeasible), Profits $(p_i)_i, (p_{ij})_{ij}$, Weights $(w_i)_i$, Capacity $C$}
 \Ensure{Feasible solution $x$}
    \For {$i=1,\cdots, n$} 
        \State $e_i \leftarrow (p_{i} + 
        \sum_{j=1}^{i-1} p_{ji} x_j +
        \sum_{j=i+1}^{n} p_{ij} x_j ) / w_i$
    \EndFor
    \While {$\sum_k w_k x_k > C$} \label{state:repair_start}
        \State Take $j \in  {\operatorname{argmin}} \{e_i \mid x_i=1\}$
        \State $x_j \leftarrow 0$ \Comment{Remove an item}
        \State Update $(e_i)_i$
    \EndWhile \label{state:repair_end}
    \For{$j$ s.t. $x_j=0$ in decreasing order of $e_j$} \label{state:add_start}
        \If {$\sum_k w_k x_k + w_j \le C$}
            \State $x_j \leftarrow 1$ \Comment{Add an item}
            \State Update $(e_i)_i$
        \EndIf
    \EndFor \label{state:add_end}
    \For{$i$ s.t. $x_i=1$ in increasing order of $e_i$} \label{state:swap_start}
            \For{$j$ s.t. $x_j=0$ in decreasing order of $e_j$}
                \If{$\sum_k w_k x_k - w_i + w_j \le C$ and $e_i w_i < e_j w_j - p_{ij}$}
                    \State $x_i \leftarrow 0$, $x_j \leftarrow 1$ \Comment{Swap items}
                    \State Update $(e_i)_i$
                \EndIf
            \EndFor
    \EndFor \label{state:swap_end}
 \State \Return $x$ 
 \end{algorithmic} 
\end{algorithm}

Both the repair and improvement procedures are built upon well-known greedy heuristics used in the previous studies~\cite{gallo1980quadratic,chaillou1989best,billionnet1996linear}.
We review the ideas of both procedures briefly to make the argument self-contained.

For the repair procedure, we note that an infeasible solution can be made into a feasible solution by removing several items from the knapsack since the weights are positive and there is a trivial feasible solution $x=(0,\cdots,0)$.
To reduce the loss of the objective value, items to be removed are selected one by one greedily.
On the simple knapsack problem with the linear objective, a greedy strategy is typically based on a metric called \emph{efficiency} defined by a ratio of the profit and weight of the item.
In the QKP, the efficiency $e_i(x)$ of item $i$ with respect to 
an incumbent solution $x$ is defined as
\begin{align}
    e_i(x) \coloneqq \frac{p_i + 
    \sum_{j=1}^{i-1} p_{ji} x_j +
    \sum_{j=i+1}^{n} p_{ij} x_j
    }{w_i}.
\end{align}
Consequently, item $i$ with $x_i=1$ achieving minimum $e_i(x)$
is removed iteratively until the constraint is satisfied.
Note that this greedy removal operation is previously used for a constructive heuristic with an input $x=(1,\cdots,1)$
~\cite{chaillou1989best,billionnet1996linear}.

The improvement procedure consists of so-called \emph{fill-up and exchange} (FE) operation~\cite{gallo1980quadratic}, which is widely used in heuristic methods on the QKP~\cite{fomeni2014dynamic,patvardhan2016parallel}.
The fill-up operation puts items into the knapsack unless it violates the capacity constraint.
Then, the exchange operation replaces an item in the knapsack with another item that is not in the knapsack, so that it improves the objective value maintaining feasibility.
In other words, the fill-up operation modifies a feasible solution to a local optimum, and the exchange operation searches neighborhood local optima.
In our method, the order of item selection for FE operation is again based on the greedy strategy with the efficiency $e_i(x)$.
An item to be included in the knapsack is chosen following the descending order of $e_i(x)$ and an item to be removed from the knapsack is chosen following the ascending order of $e_i(x)$.

The overall process is summarized in Algorithm~\ref{alg:repair_solution}.
Every time the solution $x$ is changed, the efficiency $e_i$ is updated with computational cost of order $O(n)$.
The total complexity of the algorithm is $O(n^3)$ in the worst case, but the number of the exchange operation (which is the bottleneck) is typically much less than $n^2$, and so the algorithm runs practically fast.
Indeed, a quadratic scaling of the processing time is observed in our experiments in Section~\ref{sec:experiment}.

The post-processing above is closely related to a greedy heuristic proposed by Billionnet and Calmels~\cite{billionnet1996linear}.
Their method is to first obtain a feasible solution with the greedy removal operation for $x=(1,\cdots,1)$ and then apply the FE operation.
In particular, when the penalty coefficient $\lambda$ in (\ref{eq:qkp_qubo}) is set to $0$, then the optimal solution is obviously $x=(1,\cdots,1)$.
Thus, for sufficiently small $\lambda$, an Ising machine with the post-processing outputs the same solution as the one obtained by the greedy method.

The ideas of the repair and improvement procedures are not new as mentioned above.
Besides, more elaboration on the post-processing can be made to improve the solving performance further with additional computational costs.
In this study, however, the specific implementation is not of much interest.
Rather, we aim to show that combining the simple post-processing based on the well-known ideas effectively overcomes the critical performance issue of Ising machines,
\hl{which does not seem to be understood well in the existing studies~\cite{parizy2021analysis,tamura2021performance}.
The simplicity of the proposed method is preferable in terms of extensibility:
establishing the effectiveness with the naive implementation
leads to expectation that the approach also works on other problems.
}

\begin{figure*}[t]
     \centering
     \subfloat[$n=100$\label{fig:sa_n100}]{
         \includegraphics[width=0.31\textwidth]{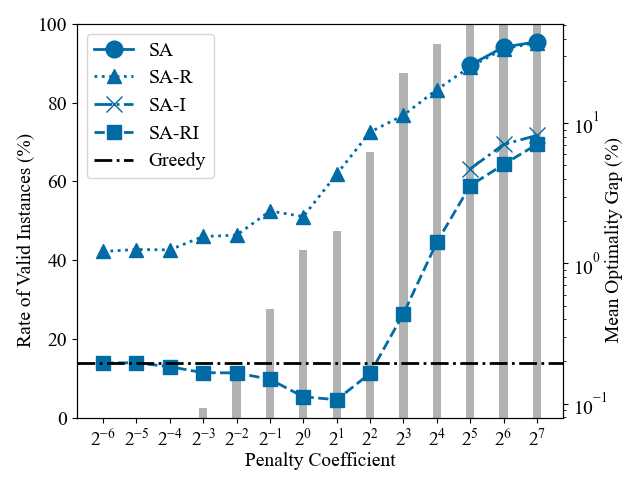}
     }
     \hfil
     \subfloat[$n=200$\label{fig:sa_n200}]{
         \includegraphics[width=0.31\textwidth]{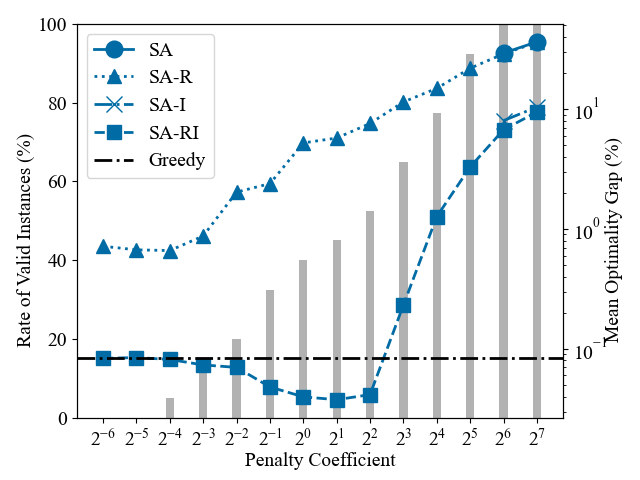}
     }
     \hfil
     \subfloat[$n=300$\label{fig:sa_3100}]{
         \includegraphics[width=0.31\textwidth]{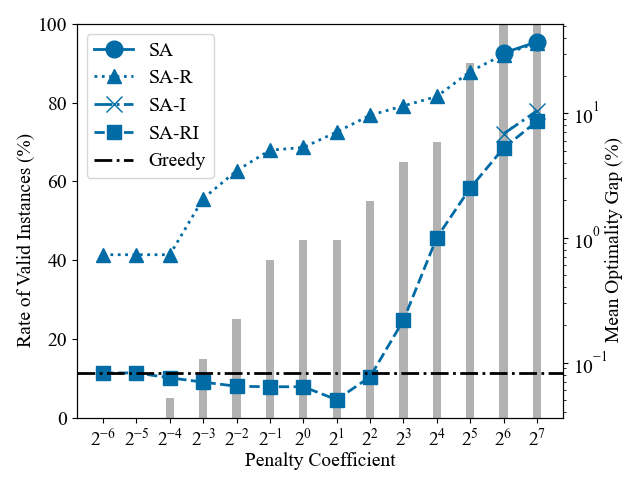}
     }
    \caption{
    Optimality gap (line graph) and number of instances on which feasible solutions are obtained with SA (bar chart) for each problem size $n$ of QKP instances.
    Optimality gap for SA and SA-I is plotted only for $\lambda$ producing feasible solutions on all instances.
    By combining repair and improvement procedures, SA-RI achieves smaller optimality gap than greedy method.
    }
    \label{fig:sa_small_plot}
\end{figure*}

\section{Simulation Experiments}\label{sec:simulation}

We validate the proposed method via simulation of Ising machines on the basis of simulated annealing (SA) that takes a QUBO problem as an input.
Note that most digital Ising machines are based on SA~\cite{yamaoka201520k,aramon2019physics}, and also SA is treated as a classical counterpart of quantum annealing~\cite{battaglia2005optimization,heim2015quantum}.
Therefore, controlled experiments with SA provide informative insights on the use of Ising machines.
For a test bed, we use a data set of 100 medium-sized QKP instances generated in the previous study~\cite{billionnet2004exact}.
There are 10 generated instances 
for each combination of the problem size $n \in \{ 100, 200, 300 \}$ and density $d\%$ of the objective function for $d \in \{25, 50, 75, 100\}$ except for $(n,d) = (300,75)$ and $(300,100)$.
Specifically, the pairwise profit $p_{ij} \ (i<j)$ is non-zero with probability $d/100$ in the generation procedure.
The exact optimal solutions of these instances are known and
the data set has been used in the existing benchmark of Ising machines~\cite{parizy2021analysis,jimbo2022hybrid,bontekoe2023translating}.
Things to be verified are as follows:
(i) better solutions (in particular, the optimal solutions) are obtained by utilizing the post-processing and (ii) the computational cost for the post-processing is sufficiently small compared to the rest of the whole process.
Furthermore, we re-evaluate various encoding methods of the inequality constraints~\cite{tamura2021performance,jimbo2022hybrid,bontekoe2023translating} under the existence of the post-processing to verify the robustness of the proposed method.

\subsection{Computational Set-up}

Each QKP instance is translated into a QUBO problem (\ref{eq:qkp_qubo}) with binary encoding (\ref{eq:binary_encoding}) of the integer variable $z$ where the upper bound $D$ of $z$ is set to the capacity $C$.
The penalty coefficient $\lambda$ is varied for $\lambda = 2^i, i=-6,-5,\cdots,6,7$.
For each $\lambda$, SA is executed 10 times to obtain 10 solutions.
The setting of SA is as follows.
We use the public implementation of SA on D-Wave Ocean SDK\footnote{https://github.com/dwavesystems/dwave-ocean-sdk} of version 6.4.1.
In the algorithm, the temperature is successively decreased from the initial value to the end value, iterating an inner loop consisting of Monte-Carlo (MC) steps for all variables.
Following the previous studies~\cite{tamura2021performance,jimbo2022hybrid}, the number of inner loops is set to $10^6$ and the initial and end temperatures are set to $n\max_{i,j} |Q_{i,j}|$ and $0.1$, respectively.
Here, $Q_{i,j}$ is the QUBO matrix for (\ref{eq:qkp_qubo}), i.e.,
\begin{align}
    \sum_{i,j:i\le j} Q_{i,j} \hat x_i \hat x_j = -H(x) + \lambda H_\mathrm{ineq} (x,z),
\end{align}
where $\hat x = (x_1, \cdots, x_n, y_1, \cdots, y_k)$ is a vector of the whole variables including $y_1,\cdots,y_k$ in (\ref{eq:binary_encoding}).
The experiment program is coded with python 3.11.4 and run on a CentOS (version 7.6.1810) server with Intel Xeon Gold 6130 chip.

We set SA without post-processing (which we simply call SA) and the greedy algorithm described in Section~\ref{sec:proposed_method} as baselines, and compare them to SA with the repair and/or improvement procedure (which we call SA-R, SA-I, and SA-RI, respectively).
We summarize the compared methods in Table~\ref{tab:baselines}.
The quality of a solution is evaluated via the optimality gap
\begin{align}\label{eq:optimality_gap}
    \mathrm{Optimality\ Gap} = {\frac{S_\mathrm{best} - S}{S_\mathrm{best}}} \times 100 \ (\%),
\end{align}
where $S_\mathrm{best}$ is the optimal value for the QKP instance and $S$ is the objective value of the solution.
For methods other than the (deterministic) greedy method, the optimality gap is taken as the minimum over all feasible solutions obtained for each $\lambda$.
For SA, we also count the number of instances on which a feasible solution is obtained, for each $\lambda$.
The optimality gap for SA and SA-I is reported only for instances on which they obtain at least one feasible solution.

\begin{table}[t]
  \caption{Description of Compared Methods.}
  \label{tab:baselines}
  \centering
  \begin{tabular}{c|l}
    \hline
    Name & Description \\
    \hline \hline
    Greedy & Equivalent to post-processing on $x= (1,\cdots,1)$ \\
    SA     & SA without post-processing (may output infeasible solutions) \\
    SA-R   & SA with repair procedure \\
    SA-I   & SA with improvement procedure (only for feasible solutions) \\
    SA-RI  & SA with both repair and improvement procedures \\
    \hline
  \end{tabular}
\end{table}

\subsection{Results}\label{subsec:sa_small_result}

\subsubsection{Observations from Tuning of Penalty Coefficients}

The optimality gap of each method aligned with the penalty coefficient $\lambda$ averaged over instances of the same size are shown in Fig.~\ref{fig:sa_small_plot}.
We also show the rate of the number of instances where SA outputs at least one feasible solution (which we call valid instances) as bar charts.
For SA and SA-I, the optimality gap is plotted only when feasible solutions are obtained on all instances for each $\lambda$ and not shown otherwise.
The first thing to observe from the results of SA is that the rate of valid instances increases for large $\lambda$, whereas large $\lambda$ degrades the optimality gap.
Therefore, SA achieves its smallest optimality gap on the minimum $\lambda_\mathrm{SA}$ among those giving feasible solutions on all instances, i.e., $\lambda_\mathrm{SA}=32$ for $n=100$ and $\lambda_\mathrm{SA}=64$ otherwise.
Note that the best optimality gap of SA is much worse than that of the greedy method.
Since the greedy method runs several orders of magnitude faster than SA, we conclude that SA without post-processing is completely inferior to the greedy method on the QKP.
When the repair method is applied, the optimality gap of SA-R roughly extrapolates that of SA, as expected.
Accordingly, the optimality gap of SA-R achieves smaller values than that of SA for $\lambda<\lambda_\mathrm{SA}$.
A similar phenomenon was observed by Fukada et al.~\cite{fukada2021three} on a variant of the QAP.
This result indicates the effectiveness of tuning $\lambda$ based on the objective value instead of the rate of feasible solutions, which is realized thanks to the repair procedure.
The optimality gap is further reduced after combining with the improvement procedure.
Although using only either of the two procedures is not sufficient to outperform the greedy method, SA-RI using both procedures achieves a smaller optimality gap than that of the greedy method.
This suggests that the two procedures improve the solving performance of Ising machines synergistically.
Note that as $\lambda$ gets closer to 0, the optimality gap of SA-RI converges to that of the greedy method.
This is expected as we argued in Section~\ref{sec:proposed_method}, that is, SA outputs the trivial solution $x=(1,\cdots,1)$ for extremely small $\lambda$.
The same argument applies to SA-R; as $\lambda\to 0$, the optimality gap converges to that of a weak version of the greedy algorithm that only repairs $x=(1,\cdots,1)$.

There are two other interesting observations from Fig.~\ref{fig:sa_small_plot} regarding the optimal penalty coefficient.
One is that penalty coefficient $\lambda$ minimizing the averaged optimality gap of SA-RI seems independent of the problem size $n$.
We discuss this phenomenon in Section~\ref{subsec:optimal_penalty}, where
the dependence of the optimal $\lambda$ on instance data including $n$ and other factors is analyzed quantitatively.
The other observation is that $\lambda$ minimizing the optimality gap of SA-R and that of SA-RI completely differ: $\lambda_{\text{SA-R}}$ for SA-R is near 0 and $\lambda_{\text{SA-RI}}$ for SA-RI is around 2 for all problem size $n$.
\hl{We analyze this in detail in Appendix~\ref{app:aux_experiments}.
}

\begin{table}[t]
  \caption{Number of Medium-sized Instances Optimally Solved.}
  \label{tab:small_optimal_count_sa}
  \centering
  \begin{tabular}{c|ccccc}
    \hline
    $n\_d$ & Greedy & SA & SA-R & SA-I & SA-RI \\ 
    \hline \hline
    100\_25   & 3 & 0 & 3 & 6 & \hl{\textbf{9}} \\
    100\_50   & 4 & 1 & 1 & 8 & \hl{\textbf{10}} \\
    100\_75   & 4 & 1 & 4 & 6 & \hl{\textbf{9}} \\
    100\_100  & 4 & 0 & 2 & 8 & \hl{\textbf{10}} \\
    200\_25   & 2 & 0 & 0 & 5 & \hl{\textbf{9}} \\
    200\_50   & 4 & 0 & 2 & 3 & \hl{\textbf{6}} \\
    200\_75   & 4 & 0 & 1 & 3 & \hl{\textbf{8}} \\
    200\_100  & 2 & 0 & 1 & \hl{\textbf{5}} & \hl{\textbf{5}} \\
    300\_25   & 4 & 0 & 2 & 3 & \hl{\textbf{8}} \\
    300\_50   & 4 & 0 & 1 & 5 & \hl{\textbf{8}} \\
    \hline
    Total     & 35 & 2 & 17 & 52 & \hl{\textbf{82}} \\
    \hline
  \end{tabular}
\end{table}

\begin{table}[t]
  \caption{Average Optimality Gap (\%).}
  \label{tab:small_optimal_gap_sa}
  \centering
  \begin{tabular}{c|rrrrr}
    \hline
    $n\_d$ & \multicolumn{1}{c}{Greedy} &
    \multicolumn{1}{c}{SA} & 
    \multicolumn{1}{c}{SA-R} & 
    \multicolumn{1}{c}{SA-I} & 
    \multicolumn{1}{c}{SA-RI} \\ 
    \hline \hline
    100\_25  & 0.370 & 6.651 & 0.797 & 0.139 & \hl{\textbf{0.047}} \\
    100\_50  & 0.101 & 6.358 & 0.272 & 0.103 & \hl{\textbf{0.000}} \\
    100\_75  & 0.115 & 5.821 & 0.208 & 0.426 & \hl{\textbf{8.4E-3}} \\
    100\_100  & 0.196 & 10.202 & 0.395 & 7.0E-3 & \hl{\textbf{0.000}} \\
    200\_25  & 0.173 & 9.404 & 0.325 & 0.318 & \hl{\textbf{5.1E-3}} \\
    200\_50  & 0.049 & 8.624 & 0.122 & 0.421 & \hl{\textbf{0.011}} \\
    200\_75  & 0.049 & 8.624 & 0.200 & 2.259 & \hl{\textbf{2.9E-3}} \\
    200\_100  & 0.062 & 10.357 & 0.206 & 0.995 & \hl{\textbf{0.034}} \\
    300\_25  & 0.127 & 10.098 & 0.230 & 0.484 & \hl{\textbf{1.4E-3}} \\
    300\_50  & 0.038 & 11.329 & 0.245 & 1.415 & \hl{\textbf{4.4E-4}} \\
    \hline
    Mean  & 0.128 & 8.747 & 0.300 & 0.657 & \hl{\textbf{0.011}} \\
    \hline
  \end{tabular}
\end{table}

\subsubsection{Results on Best Optimality Gap}

The number of instances on which each method achieved the optimal solution is reported in Table~\ref{tab:small_optimal_count_sa}.
We also summarize the optimality gap averaged over 10 instances for each pair $(n,d)$ in Table~\ref{tab:small_optimal_gap_sa}.
SA achieves the optimal solutions on only two instances among 100 instances in total.
Although SA-R achieves the optimum on several instances, the total number of such instances is less than that of the greedy method.
SA-I obtains the optimal solutions more frequently than SA-R and the greedy method, but its averaged optimality gap is worse than the others.
This means that the quality of solutions of SA-I has much variance over instances, which is often undesirable.
SA-RI, the proposed method, successfully attains the optimum on \hl{82} instances in total and achieves the smallest optimality gap for all pairs of $(n,d)$.
These results clearly demonstrate the effectiveness of combining the repair and improvement procedures as the post-processing for SA.
For full results on each instance, see Appendix~\ref{app:sa_full_results}.

\begin{table}[t]
  \caption{Average Processing Time.}
  \label{tab:sa_running_time}
  \centering
  \begin{tabular}{c|cc|cc}
    \hline
    & \multicolumn{2}{c|}{Before Post-process (s)} & \multicolumn{2}{c}{Post-process (ms)} \\
    $n\_d$  & Formulation & SA & Repair & Improve \\
    \hline \hline
    100\_25  & 0.07 & 4.4 & 0.9 & 1.0 \\
    100\_50  & 0.07 & 4.3 & 1.0 & 1.0 \\
    100\_75  & 0.07 & 4.0 & 1.0 & 0.9 \\
    100\_100  & 0.07 & 3.7 & 1.0 & 0.8 \\
    200\_25  & 0.23 & 17.3 & 3.5 & 3.7 \\
    200\_50  & 0.24 & 17.0 & 3.8 & 3.7 \\
    200\_75  & 0.24 & 14.0 & 4.1 & 2.9 \\
    200\_100  & 0.25 & 12.8 & 3.9 & 2.8 \\
    300\_25  & 0.51 & 34.4 & 8.1 & 7.0 \\
    300\_50  & 0.52 & 36.5 & 8.8 & 7.0 \\
    \hline
  \end{tabular}
\end{table}

\subsubsection{Results on Processing Time}

We evaluate the computational overhead of the post-processing.
The average processing time for each process is reported in Table~\ref{tab:sa_running_time}.
In addition to the execution time of SA and the repair and improvement procedures, we include the processing time to create the input QUBO object after reading data of the corresponding QKP instance as the ``Formulation'' column.
We see that time for each process increases roughly with an order of $n^2$.
Note that we report processing time before the post-processing in seconds and time for the post-processing in milliseconds.
The time required for the post-processing is more than 1000 times less than that of the annealing, and also much less than the formulation.
Therefore, the proposed method improves the accuracy with a negligibly small amount of additional computational cost.

\subsection{Dependency on Encoding Methods}\label{subsec:compare_encoding}

\begin{figure*}[t]
     \centering
     \subfloat[Rate of Feasible Solutions\label{fig:encoding_feasibility}]{
         \includegraphics[width=0.31\textwidth]{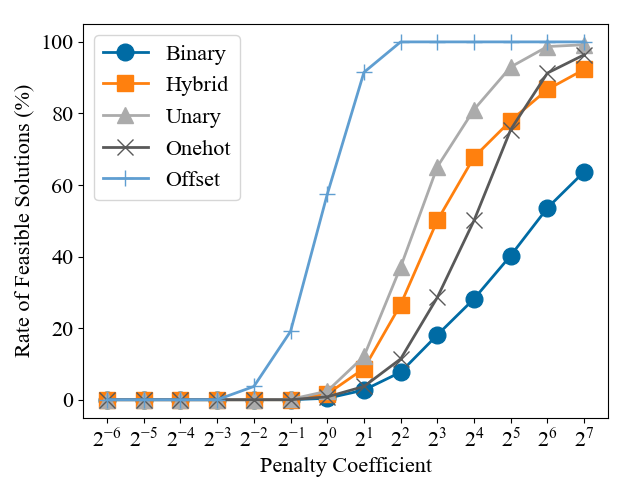}
     }
     \hfil
     \subfloat[Optimality Gap of SA\label{fig:encoding_raw}]{
         \includegraphics[width=0.31\textwidth]{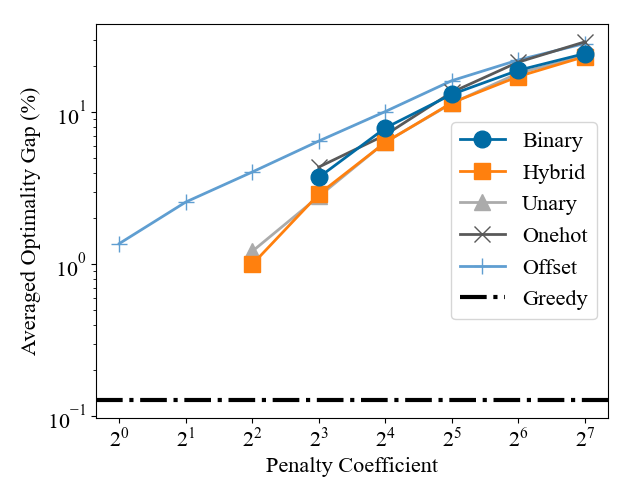}
     }
     \hfil
     \subfloat[Optimality Gap with Post-processing\label{fig:encoding_score}]{
         \includegraphics[width=0.31\textwidth]{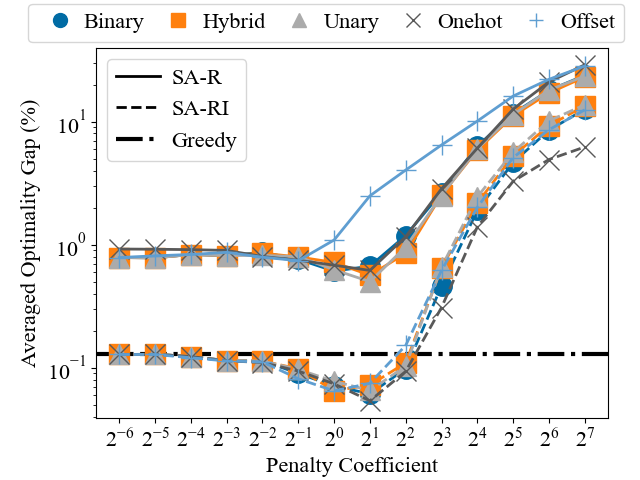}
     }
    \caption{
    Performance comparison among various encoding methods of inequality constraint on 100 medium-sized QKP instances.
    (a)(b) Choice of encoding methods controls trade-off between rates of feasible solutions and objective values.
    (c) Solving performance of proposed method is much less dependent on choice of encoding methods.
    }
    \label{fig:sa_compare_encoding}
\end{figure*}

As described in Section~\ref{sec:preliminaries}, the previous studies~\cite{tamura2021performance,jimbo2022hybrid,bontekoe2023translating} suggest that other encoding methods of the slack variable $z$ in (\ref{eq:qkp_qubo}) than the standard binary encoding (\ref{eq:binary_encoding}) might enhance the quality of solutions obtained by Ising machines.
Since their evaluation has been conducted without any post-processing, we re-evaluate various encoding methods with the proposed post-processing in this section.
\hl{We report simulation results based on SA here since it reproduces well the results of the previous studies~\cite{tamura2021performance,jimbo2022hybrid,bontekoe2023translating} as shown below.
We also conducted the same experiment with a real Ising machine and obtained mostly similar results, see Appendix~\ref{app:aux_experiments} for details.
}

The setting is as follows.
We consider the following five variations of encoding methods of $z$ in the QUBO problem (\ref{eq:qkp_qubo}) of the QKP.
The first is the binary encoding shown in (\ref{eq:binary_encoding}).
Recall that it involves $k$ auxiliary variables $y_1, \cdots, y_k$ with $k = \lfloor \log D \rfloor + 1$, where $D$ denotes the upper bound of $z$.
The second is the unary encoding defined as 
\begin{align}
    z = \sum_{i=1}^D y_i,
\end{align}
which involves $D$ auxiliary variables $y_1, \cdots, y_D$.
The third is the hybrid encoding~\cite{jimbo2022hybrid}, which hybridizes the unary and binary encoding.
As it has several degrees of freedom, we adopt the following form close to a method called $HE(1)$ in the previous experiment~\cite{jimbo2022hybrid}:
\begin{align}
    z = \sum_{i=1}^k y_i + \sum_{i=k+1}^{2k} 2 y_i, \ k \coloneqq \lceil D/3 \rceil.
\end{align}
The hybrid encoding involves $2\lceil D/3 \rceil$ auxiliary variables.
The fourth is the one-hot encoding, which uses an additional penalty term $H_\mathrm{onehot}$ and modify the objective function of the QUBO problem as 
\begin{align}
    -H(x) + \lambda \left( H_\mathrm{ineq}(x, z) + H_\mathrm{onehot}\right),
\end{align}
defining
\begin{align}
    z = \sum_{i=0}^D i y_i, \ H_\mathrm{onehot} = \left(\sum_{i=0}^D y_i -1 \right)^2.
\end{align}
The one-hot encoding involves $D+1$ auxiliary variables.
The last is the offset encoding~\cite{bontekoe2023translating},
which set $z$ to a constant 
\begin{align}
    z=W_\mathrm{offset}    
\end{align}
with some small number $W_\mathrm{offset}\ge 0$.
Since $z$ does not work as a slack variable any more,
the offset encoding does not preserve the equivalence of the optimization problems.
Nevertheless, Bontekoe et al.~\cite{bontekoe2023translating} reported that it outperformed other encoding methods.
We set $W_\mathrm{offset}=3$ following the previous result.
All methods other than the offset encoding involve the upper bound $D$ of $z$.
Note that it suffices to set $D$ to a value greater than or equal to $\max_i w_i$ to translate the QKP to the QUBO problem preserving the optimum according to Proposition~\ref{prop:qkp_property}.
On the other hand, since methods other than the binary encoding uses $O(D)$ auxiliary variables, $D$ should be sufficiently small to effectively apply Ising machines.
Therefore, we set $D$ to $\max_i w_i$ in this experiment.
Other settings are identical to those in the earlier experiment.

We remark that an output of Ising machines or SA can have a positive penalty $H_\mathrm{ineq}(x, z) > 0$ (or $H_\mathrm{onehot} > 0$ for the one-hot encoding)
even if the solution $x$ is feasible.
Such situations include a case where $z \ne \sum_i w_i x_i$, as well as a case where $\sum_{i=0}^D y_i\ne 1$ for the one-hot encoding.
It is in contrast to the previous evaluation~\cite{tamura2021performance} treating the solution as feasible only when it has zero penalty, and this difference in definition could lead to different results.
In particular, for the one-hot encoding above, it is actually not necessary to impose the one-hot constraint $\sum_{i=0}^D y_i=1$, since the inequality constraint can be satisfied even when $\sum_{i=0}^D y_i=0$ or $\sum_{i=0}^D y_i \ge 2$.
Note that this fact is also used in the previous study~\cite{bontekoe2023translating}.
Therefore, the aforementioned case where $x$ is feasible and $H_\mathrm{onehot} > 0$ can particularly often occur, and we indeed observed this phenomenon in our experiment.

Fig.~\ref{fig:sa_compare_encoding} shows the results over various $\lambda$.
Fig.~\ref{fig:encoding_feasibility} shows the rate of feasible solutions (we call FS rate) over all instances for each encoding.
On all methods, a larger penalty coefficient results in a high FS rate.
Among the tested encoding methods, the binary encoding leads to the lowest, while the offset encoding achieves the highest.
The difference might be explained by the number of flips of auxiliary variables $y_i$ required for a flip of a variable $x_i$, which is mentioned in Section~\ref{sec:preliminaries}.
More precisely, multiple MC steps in SA are required to realize a single flip on the QKP.
The offset encoding uses no auxiliary variables, and thus the penalty $H_\mathrm{ineq}$ might be easily decreased by local operations in SA, leading to the high FS rate.
In contrast, a lot of MC steps are required for changing the value of $z$ for the binary encoding, resulting in a low FS rate.
The redundancy of the representation (i.e. representing a value of $z$ by multiple combinations of values of $y_1,y_2,\cdots$) in the unary and hybrid encoding might help to make the number of required MC steps small~\cite{tamura2021performance}, and thus they give the intermediate results.
For the one-hot encoding, most solutions violate the one-hot constraint and as a result obtain a similar redundancy, which again explains the intermediate result.
The optimality gap of the feasible solutions obtained by SA is shown in Fig.~\ref{fig:encoding_raw}.
Here, we plot the optimality gap for $\lambda$ that obtains a feasible solution on more than half of all instances to exclude outlier values. 
Again, for all methods, a smaller penalty coefficient leads to better objective values.
The hybrid, unary, and offset encodings achieve a lower optimality gap than the others, due to the high FS rate at small $\lambda$.
These results on the FS rate and optimality gap agree well with the previous studies~\cite{tamura2021performance,jimbo2022hybrid,bontekoe2023translating}.

Fig.~\ref{fig:encoding_score} shows the optimality gap for each method combined with the post-processing.
Interestingly, after the post-processing, the difference among the encoding methods gets almost negligible and all methods reach a similar minimum optimality gap at the similar value of $\lambda$.
A subtle exception is the offset encoding; SA-R with the offset encoding attains the minimum optimality gap at $\lambda_{\text{SA-R}}=0.5$, unlike the others.
This is presumably because fixing the slack variable $z$ to a constant changes the effect of penalty $H_\mathrm{ineq}$ on the behavior of SA.
The overall result indicates that the proposed method is much robust to the choice of encoding methods, compared to SA without post-processing.
A fundamental reason for the somewhat surprising similarity of the post-processed outputs over the various encoding methods is unclear and might be related to the behavior of the SA algorithm.
Since a precise algorithmic analysis is beyond the scope of this paper, further investigation is left as future work.

\begin{table}[t]
  \caption{Number of Medium-sized Instances Optimally Solved.}
  \label{tab:compare_optimal_count_sa}
  \centering
  \begin{tabular}{c|ccccc}
    \hline
    $n\_d$ &  Binary  &  Hybrid  &  Unary  &  One-hot  &  Offset \\
    \hline \hline
    100\_25  & 10 & 9 & 8 & 10 & 10 \\
    100\_50  & 9 & 9 & 9 & 9 & 9 \\
    100\_75  & 9 & 8 & 8 & 8 & 8 \\
    100\_100 & 8 & 8 & 8 & 9 & 7 \\
    200\_25  & 9 & 8 & 10 & 8 & 8 \\
    200\_50  & 7 & 7 & 6 & 7 & 7 \\
    200\_75  & 8 & 9 & 9 & 8 & 8 \\
    200\_100 & 6 & 6 & 6 & 6 & 6 \\
    300\_25  & 7 & 7 & 7 & 8 & 8 \\
    300\_50  & 10 & 9 & 9 & 8 & 9 \\
    \hline
    Total    & \hl{\textbf{83}} & 80 & 80 & 81 & 80 \\
    \hline
  \end{tabular}
\end{table}

\begin{table}[t]
  \caption{Averaged Optimality Gap ($\times 0.01$ \%).}
  \label{tab:compare_optimal_gap_sa}
  \centering
  \begin{tabular}{c|rrrrr}
    \hline
    $n\_d$  &  Binary  &  Hybrid  &  Unary  &  One-hot  &  Offset \\
    \hline \hline
    100\_25  & 0.000 & 9.328 & 4.355 & 0.000 & 0.000 \\
    100\_50  & 0.384 & 0.610 & 0.610 & 0.666 & 0.610 \\
    100\_75  & 0.537 & 1.590 & 1.590 & 1.140 & 1.140 \\
    100\_100  & 0.412 & 0.412 & 14.338 & 0.205 & 14.546 \\
    200\_25  & 0.510 & 0.659 & 0.000 & 0.253 & 3.585 \\
    200\_50  & 0.343 & 0.888 & 0.761 & 0.260 & 0.888 \\
    200\_75  & 0.917 & 0.213 & 0.213 & 0.297 & 0.884 \\
    200\_100  & 0.995 & 1.079 & 1.129 & 0.624 & 0.803 \\
    300\_25  & 0.418 & 3.710 & 0.180 & 0.135 & 0.246 \\
    300\_50  & 0.000 & 0.035 & 0.241 & 0.055 & 0.184 \\
    \hline
    Mean  & 0.452 & 1.853 & 2.342 & \hl{\textbf{0.363}} & 2.288 \\
    \hline
  \end{tabular}
\end{table}

For a quantitative performance comparison, we summarize the number of instances optimally solved and the optimality gap for each encoding with the proposed method in Table~\ref{tab:compare_optimal_count_sa} and \ref{tab:compare_optimal_gap_sa}.
We see that the binary and one-hot encodings slightly outperform the other methods on average in terms of both metrics.
In particular, among the binary, unary, and one-hot encodings, the unary encoding performs the worst (by a possibly negligible margin), in contrast to the previous evaluation without the post-processing~\cite{tamura2021performance}.
In other words, whether or not a specific encoding method performs well can be easily changed by additional operations.
This leads to an insight important to practitioners that performance evaluation of Ising machines should be carefully done in a practical situation when it involves pre- or post-processing of the problem or solutions.

\subsection{Analysis of Optimal Penalty Coefficients}\label{subsec:optimal_penalty}

In the earlier experiments, we observed that the optimal penalty coefficient $\lambda_{\text{SA-RI}}$ for the proposed method varies depending on the problem instances (see Appendix~\ref{app:sa_full_results} for full results including $\lambda_{\text{SA-RI}}$ for each instance).
The optimal penalty coefficient could be estimated by some representative features of the instance data~\cite{fukada2021three}.
In this section, we analyze $\lambda_{\text{SA-RI}}$ over the tested instances to utilize the result for solving larger instances in the later section.

As representative features of the QKP, we consider the problem size $n$, density $d$ of the objective function, and tightness ratio $\alpha = C/\sum_i w_i$ of the inequality constraint.
Note that the tightness ratio $\alpha$ has not been mentioned in the QKP literature,
whereas it is recognized as an important factor in the context of the multi-dimensional knapsack problem~\cite{chu1998genetic,hill2000effects}.
We expect that $n$ has weak correlation with $\lambda_{\text{SA-RI}}$, as we see from Fig.~\ref{fig:sa_small_plot} for each $n$.
On the other hand, the density $d$ involves the scale of the increase in the objective value for putting an item into the knapsack.
Since it is typically considered that the scales of the objective function and penalty should be balanced when applying the penalty method, we expect that $\lambda_{\text{SA-RI}}$ tends to be large for large $d$.

\begin{table}[t]
  \caption{Regression Coefficients for Optimal Penalty Coefficients.}
  \label{tab:qkp_penalty_regression}
  \centering
  \begin{tabular}{cccc}
    \hline
    $A$ & $c_n$ & $c_d$ & $c_\alpha$ \\
    \hline
    $1.14$ & $0.09$ & $0.84$ & $-0.21$ \\
    \hline
  \end{tabular}
\end{table}

We model $\lambda_{\text{SA-RI}}$ as the product of the features by
\begin{align}
    \lambda_\mathrm{Estimate} = A n^{c_n} d^{c_d} \alpha^{c_\alpha},
\end{align}
where $A, c_n, c_d,$ and $c_\alpha$ are parameters to be fit.
We show the results of log-linear regression on $\lambda_{\text{SA-RI}}$ for the tested 100 instances in Table~\ref{tab:qkp_penalty_regression}.
As expected, the resulting coefficient for $n$ is close to 0 and that for $d$ is a large positive value.
The coefficient $c_\alpha$ for $\alpha$ is negative, which means that $\lambda$ should be lowered for large capacity $C$.
This might be because large $\alpha$ implies that feasible solutions occupy a large fraction of the total space $\{0,1\}^n$, and thus the penalty is not required to be much emphasized for solving the QKP.
Note that the overall analysis is on a data set created following a specific procedure of instance generation, and the result might depend on the distribution of problem instances.
Since larger instances used in the later section are based on the same generating protocol as that of the instances used above, we make use of the analysis result to solve the larger instances.

\section{Evaluation using Ising Machine}\label{sec:experiment}

In this section, we evaluate the proposed method using one of the state-of-the-art Ising machines, Amplify Annealing Engine (AE)~\cite{amplify}, on a broader set of QKP instances.
Our aim in this experiment is to verify that the proposed method works also for a high-performance Ising machine as well as for naive SA.
\hl{We use large QKP instances for which exact methods require high computational time to solve.
We compare the performance of the Ising machine with that of existing heuristic solvers.}

\begin{table}[t]
  \caption{Number of Medium-sized Instances Optimally Solved.}
  \label{tab:small_optimal_count_ae}
  \centering
  \begin{tabular}{c|ccccc}
    \hline
    $n\_d$   & Gurobi & AE & AE-R & AE-I & AE-RI \\
    \hline \hline
    100\_25  & 10 & 10 & 10 & 10 & 10 \\
    100\_50  & 10 & 10 & 10 & 10 & 10 \\
    100\_75  & 10 & 10 & 10 & 10 & 10 \\
    100\_100  & 10 & 10 & 10 & 10 & 10 \\
    200\_25  & 10 & 7 & 8 & 10 & 10 \\
    200\_50  & 9 & 8 & 8 & 10 & 10 \\
    200\_75  & 10 & 7 & 8 & 10 & 10 \\
    200\_100  & 10 & 7 & 7 & 10 & 10 \\
    300\_25  & 10 & 5 & 7 & 10 & 10 \\
    300\_50  & 10 & 6 & 8 & 10 & 10 \\
    \hline
    Total    & 99 & 80 & 86 & \hl{\textbf{100}} & \hl{\textbf{100}} \\
    \hline
  \end{tabular}
\end{table}

\begin{table*}[t]
  \caption{Number of Best Known Solutions Obtained and Average Optimality Gap ($\times 0.01$ \%).}
  \label{tab:large_optimality_gap}
  \centering
  \begin{tabular}{@{\hspace{2pt}}c|c@{\hspace{3pt}}cc@{\hspace{3pt}}cc@{\hspace{3pt}}cc@{\hspace{3pt}}cc@{\hspace{3pt}}c|c@{\hspace{3pt}}cc@{\hspace{3pt}}cc@{\hspace{3pt}}cc@{\hspace{3pt}}c@{\hspace{3pt}}}
    \hline
    & 
    \multicolumn{2}{c}{\multirow{2}{*}{Greedy}} & 
    \multicolumn{2}{c}{\multirow{2}{*}{DP+FE~\cite{fomeni2014dynamic}}} & 
    \multicolumn{2}{c}{GRASP+} & 
    \multicolumn{2}{c}{\multirow{2}{*}{IQIEA~\cite{patvardhan2016parallel}}} &
    \multicolumn{2}{c}{\multirow{2}{*}{Gurobi}} & 
    \multicolumn{2}{|c}{\multirow{2}{*}{AE}} & 
    \multicolumn{2}{c}{\multirow{2}{*}{AE-R}} & 
    \multicolumn{2}{c}{\multirow{2}{*}{AE-I}} & 
    \multicolumn{2}{c}{\multirow{2}{*}{AE-RI}}\\
    &&&&&\multicolumn{2}{c}{Tabu~\cite{yang2013effective}} &&&&&\multicolumn{2}{|c}{}\\
    $n\_d$ & \#BKS & Gap & \#BKS & Gap & \#BKS & Gap & \#BKS & Gap & \#BKS & Gap & \#BKS & Gap & \#BKS & Gap & \#BKS & Gap & \#BKS & Gap \\
    \hline \hline
    1000\_25  & 4 & 2.452 & 3 & 11.153 & \hl{\textbf{10}} & \hl{\textbf{0.000}} & \hl{\textbf{10}} & \hl{\textbf{0.000}} & 8 & 0.041 & 1 & - & 4 & 2.830 & 3&-& \hl{\textbf{10}} & \hl{\textbf{0.000}} \\
    1000\_50  & 1 & 1.928 & 1 & 0.434  & \hl{\textbf{10}} & \hl{\textbf{0.000}} & \hl{\textbf{10}} & \hl{\textbf{0.000}} & 8 & 0.010 & 0 & - & 4 & 0.428 & 2&-& 8 & 0.184 \\
    1000\_75  & 0 & 7.760 & 1 & 0.675  & 9  & 0.043 & 9  & 0.043 & 8 & 0.011 & 0 & - & 2 & 4.002 & 0&-& 8 & 0.062 \\
    1000\_100  & 0 & 3.782 & 2 & 0.495 & 9  & 0.121 & \hl{\textbf{10}} & \hl{\textbf{0.000}} & 7 & 0.259 & 0 & - & 1 & 1.833 & 0&-& 7 & 0.032 \\
    2000\_25  & 1 & 0.763 & 0 & 0.330  & \hl{\textbf{10}} & \hl{\textbf{0.000}} & \hl{\textbf{10}} & \hl{\textbf{0.000}} & 5 & 0.032 & 0 & - & 4 & 0.141 & 0&-& 8 & 0.010 \\
    2000\_50  & 3 & 1.297 & 2 & 0.337  & 9  & 0.034 & 9  & 0.037 & 7 & 0.042 & 0 & - & 4 & 0.360 & 0&-& 8 & 0.101 \\
    2000\_75  & 1 & 1.097 & 1 & 0.173  & 8  & 0.375 & 8  & 0.375 & 9 & 0.054 & 0 & - & 2 & 1.229 & 0&-& 7 & 0.393 \\
    2000\_100  & 0 & 2.577 & 1 & 0.257 & 9  & 0.533 & 9  & 0.512 & 5 & 0.152 & 0 & - & 2 & 2.873 & 0&-& 6 & 0.597 \\
    \hline
    All     & 10 & 2.707 & 11 & 1.732 & 74 & 0.138 & 75 & 0.121 & 57 & 0.075 & 1 & - & 20 & 1.712 & 5 &-& 62 & 0.172 \\
    \hline
    \multicolumn{19}{l}{\hl{
    Results are shown in boldface when the algorithm achieves the zero optimality gap on all instances, i.e., reaches the best known scores of IHEA~\cite{chen2017iterated}.}}
  \end{tabular}
\end{table*}

\begin{table}[t]
  \caption{\hl{
  Performance comparison of AE-RI and each baseline.}
  }
  \label{tab:ae_vs_baselines}
  \centering
  \begin{tabular}{l|rr|rr}
    \hline
    \multirow{2}{*}{\hl{Method}} & 
    \multirow{2}{*}{\hl{\# better}} &
    \multirow{2}{*}{\hl{\# worse}} & 
    \multicolumn{2}{c}{\hl{Wilcoxon test}} \\
    &&& \hl{statistic} & \hl{p-value} \\
    \hline
    Greedy & 70 & 0 & 2485.0 & 2E-13\\
    DP+FE~\cite{fomeni2014dynamic} & 66 & 5 & 2293.5 & 3E-9\\ 
    GRASP+Tabu~\cite{yang2013effective} & 1 & 15 & 16.0 & 0.996\\
    IQIEA~\cite{patvardhan2016parallel} & 1 & 15 & 6.0 & 0.999\\
    IHEA~\cite{chen2017iterated} & 0 & 18 & 0.0 & $>$0.999\\
    Gurobi & 16 & 13 & 217.0 & 0.504 \\
    \hline
    \multicolumn{5}{l}{\# better/worse represents the number of instances on which AE-RI} \\
    \multicolumn{5}{l}{performs better/worse.}
  \end{tabular}
\end{table}

\begin{table*}[t]
  \caption{Average Running Time (second) to Sample a Solution.}
  \label{tab:large_qkp_runtime}
  \centering
  \begin{tabular}{c|cccccc|cccc}
    \hline
    
    \multirow{2}{*}{$n$} & 
    \multirow{2}{*}{Greedy} & 
    \multirow{2}{*}{DP+FE~\cite{fomeni2014dynamic}} & 
    GRASP+ &
    \multirow{2}{*}{IQIEA~\cite{patvardhan2016parallel}} & 
    \multirow{2}{*}{IHEA~\cite{chen2017iterated}} & 
    \multirow{2}{*}{Gurobi} & 
    \multicolumn{4}{c}{AE-RI} \\
    &&& Tabu~\cite{yang2013effective} &&&& Formulation & AE & Repair & Improve  \\
    \hline \hline
    1000  & 0.27 &  2917.7 &  28.0 &  307.4 & 6.0 &  60.0 & 2.17 & 9.93 & 0.08 & 0.06 \\
    2000  & 1.08 & 51695.8 & 329.7 & 3034.0 & 22.7 & 60.4 & 10.74 & 20.50 & 0.33 & 0.32 \\
    \hline
  \end{tabular}
\end{table*}

\subsection{Setting}

In addition to the medium-sized instances in Section~\ref{sec:simulation},
we use another group of QKP instances generated in the previous study~\cite{yang2013effective}.
There are 10 instances for each combination of the problem size $n\in \{1000, 2000\}$ and density $d\in \{25,50,75,100\}$ of the objective function in the data set.
Their exact optimal solutions have not been known and the current best known objective values are reported by Chen and Hao~\cite{chen2017iterated}.
Therefore, we use their result to compute the optimality gap (\ref{eq:optimality_gap}) in which $S_\mathrm{best}$ denotes the best known objective value.

The computational environment is the same as in Section~\ref{sec:simulation}. 
We provide additional details on the use of the Ising machine.
We use AE of version 0.7.4 with A100 GPU.
The timeout for the execution of AE is set to $0.01n$ seconds for problem size $n$, which is
comparable with that of the existing solver~\cite{chen2017iterated}.
We use Amplify SDK~\cite{amplify} of version 0.11.2 to translate the QKP into a QUBO problem and input it to AE.
The slack variable $z$ is encoded into binary variables by binary expansion (\ref{eq:binary_encoding}) with $D=C$.

The penalty coefficient $\lambda$ is heuristically varied as
\begin{align}\label{eq:penalty_scale}
    \lambda = a \frac{d}{100}\sqrt{1/\alpha}, \ a=1,2, \cdots,
\end{align}
using the density $d$ and tightness ratio $\alpha = C/\sum_i w_i$, based on the result in Section~\ref{subsec:optimal_penalty}.
The upper bound of $a$ is set to 10 for medium-sized instances and 20 for large instances, which is found to be sufficient to obtain good solutions with the proposed method.
The Ising machine is executed 10 times to sample 10 solutions for each $\lambda$.
The solutions are evaluated by optimality gap with and without the post-processing.

We compare the performance of AE with and without the post-processing.
We call them AE, AE-R, AE-I, and AE-RI, respectively, following the same notation in Table~\ref{tab:baselines}.
We use the greedy method described in Section~\ref{sec:proposed_method} as a baseline.
Besides, the results of the following heuristic solvers tailored for the QKP are taken from the existing papers~\cite{patvardhan2016parallel,chen2017iterated} and included as baselines:
dynamic programming with fill-up and exchange (DP+FE)~\cite{fomeni2014dynamic}, GRASP combined with tabu search (GRASP+Tabu)~\cite{yang2013effective}, improved quantum-inspired evolutionary algorithm (IQIEA)~\cite{patvardhan2016parallel}, and iterated hyper-plane exploration approach (IHEA)~\cite{chen2017iterated}.
We also use Gurobi~\cite{gurobi}, one of the state-of-the-art commercial solvers, to provide an insight into performance comparison with a general-purpose method.
For each instance,
we run Gurobi (version 9.1.2) with a time limit of 1 minute and report the best solution found.

\subsection{Results}

We discuss the benchmark results on medium-sized and large QKP instances.
For the full results on each instance, we refer to Appendix~\ref{app:ae_full_results}.
Since the best known solutions (BKS) produced by the IHEA algorithm are used for evaluation, IHEA trivially achieves zero optimality gap for all instances and thus is omitted from the results.

The results on the medium-sized instances are summarized in Table~\ref{tab:small_optimal_count_ae}.
For the previous methods,
we omit the results since these instances are rather easy to reach optimality, and refer to the original results~\cite{chen2017iterated}.
For the instances of size $n=100$, AE successfully obtains the optimal solutions without the post-processing.
As $n$ increases, however, the number of instances solved optimally by AE decreases.
Meanwhile, the post-processing enables us to obtain the optimal solutions on all instances of sizes up to 300.
The result ensures that the proposed method can further enhance the solving performance of the state-of-the-art Ising machine.

The results on the large instances are shown in Table~
\ref{tab:large_optimality_gap}.
The averaged optimality gap for AE and AE-I are omitted in Table~\ref{tab:large_optimality_gap} since they could not obtain a feasible solution on some instances for every $(n,d)$.
AE-RI achieves the best known solutions on 62 instances out of 80 instances, whereas AE obtains the best known solution on only one instance.
Also, the greedy method performs poorly compared to other methods.
Note that the greedy method applies the same local operations as the post-processing.
Therefore, the result indicates that global search by the Ising machine and local search by the post-processing work well in a complementary manner in the proposed method.
Furthermore, AE-RI 
achieves comparable results with other heuristic solvers. 
This is the first result to achieve such high accuracy on the large-scale QKP using Ising machines, which might shed the light on the practical utility of Ising machines.
\hl{To compare the performance of AE-RI and each baseline more directly, we also provide a result of statistical testing of the AE-RI performance against each baseline in Table~\ref{tab:ae_vs_baselines}.
We conducted the one-sided Wilcoxon rank sum test with a null hypothesis that the performance of AE-RI is not better than a baseline.
More precisely, we count the number of QKP instances on which AE-RI achieved a larger/smaller objective value than each baseline method and calculate the test statistic and p-value.
The result shows that AE-RI indeed performs significantly better than the greedy and DP+FE methods, while it does not hold for other baselines.
In summary,
although our result significantly outperforms the previous
benchmarks of Ising machines, there is still a performance
gap between Ising machines and the state-of-the-art heuristic solvers such as IHEA.
Filling the gap could be an important milestone for further software and hardware development of Ising machines.
}

\begin{figure}
    \centering
    \includegraphics[width=0.48\textwidth]{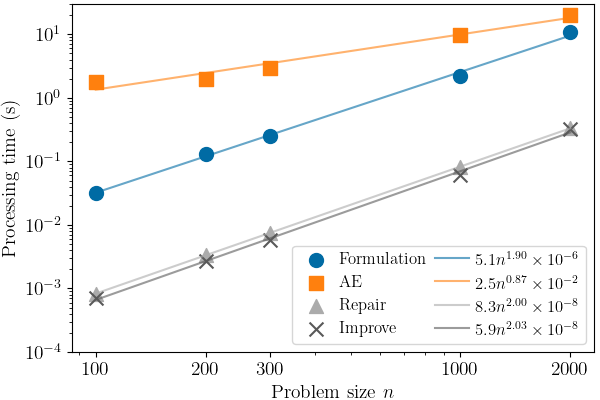}
    \caption{\hl{
    Processing time of processes in AE-RI.
    Lines are fitted with log-log regression.
    The execution time of AE is set to $O(n)$ and the other processes empirically take $O(n^2)$ runtime.
    }}
    \label{fig:runtime_scaling}
\end{figure}

We report the computational time of the proposed method in Table~\ref{tab:large_qkp_runtime}.
The processing time for formulation, repair, and improvement procedures shows an expected scaling behavior extending Table~\ref{tab:sa_running_time} to larger $n$.
The execution time of AE is around the timeout we set.
\hl{We plot the processing times against the number of variables $n$ in Fig.~\ref{fig:runtime_scaling} with log-log regression curves.
As in the case of SA-RI, we observe the quadratic scaling of processing time for formulation, repair, and improvement procedures.
}
Overall, the results imply that the post-processing causes negligible computational overhead also for the Ising machine.

Averaged running time to obtain one solution for each baseline is also listed in Table~\ref{tab:large_qkp_runtime}.
For methods other than the greedy method and Gurobi, the results are taken from the previous studies~\cite{patvardhan2016parallel,chen2017iterated}.
We do not intend a fair comparison of running time across the baselines, due to differences in the computational environments.
Moreover, since AE is a cloud service involving queue and communication time, defining a reasonable metric on computational time is itself a hard task.
Here, we just aim to get insights into the scaling behavior of running time.
The IHEA algorithm scales quite well for large $n$, and thus the comparable amount of time has been adopted for the timeout of AE in our experiments.
Further precise benchmarks including evaluation of practical solving time should be conducted in the future after establishing a method for Ising machines to achieve sufficiently high accuracy.

\section{Related Work and Discussion}\label{sec:related_work}

There are several previous studies aiming to solve the QKP using Ising machines~\cite{parizy2021analysis,tamura2021performance,jimbo2022hybrid,bontekoe2023translating}.
All of them use relatively easy QKP instances which can be handled by exact methods.
Our work is the first to solve large QKP instances ranging from 1000 to 2000 variables with Ising machines.
Parizy et al.~\cite{parizy2021analysis} propose an improvement algorithm for feasible solutions of the QKP, but their rate of instances optimally solved is only 77\% with a high-performance Ising machine while ours achieves higher rates using naive SA.
The difference might be caused by the use of the repair method.
The other studies explore a good way of encoding inequality constraints~\cite {tamura2021performance,jimbo2022hybrid,bontekoe2023translating}.
Our work takes a completely different approach and shows that
an encoding method is not an important factor for accuracy on the QKP under the existence of the post-processing
as in Section~\ref{sec:simulation}.

\hl{
The comprehensive experiments in the previous sections show that the naive post-processing leads to drastic improvement of solving performance of SA and the Ising machine.
This finding is somewhat surprising given the simplicity of the method, and seems to have been overlooked by the existing studies.
Considering that the Ising machine hardware is rapidly evolving to obtain better results and there could be room for enhancing and extending the post-processing, 
it also indicates the possibilities that Ising machines will 
be competent with other heuristic approaches 
in the future.
}

The proposed method could be extended to other problems on which a greedy heuristic is known.
Such problems may involve other types of constraints such as one-hot constraints. 
\hl{For example, the max k-cut~\cite{frieze1997improved} problem, a generalization of the max cut problem, admits an efficient implementation of a greedy local search algorithm~\cite{ma2017multiple}.
On the other hand, Ising machines do not perform well on the max k-cut problem as it involves a lot of one-hot constraints.
We expect that the post-processing approach using the greedy local search can be utilized to develop a high-performance Ising machine-based solver.
We will investigate such extensions in our future work.
}


\section{Conclusion}\label{sec:conclusion}

Toward more practical benchmarks of Ising machines,
we proposed a method to solve the QKP with Ising machines using the two-stage post-processing.
The repair and improvement procedures improve the solving performance of Ising machines synergistically.
From an empirical study using both simulation and an Ising machine, we demonstrated the effectiveness of the proposed method.
We found that the performance of the proposed method was much less dependent on a choice of the encoding methods of the inequality constraint.
Evaluation on large QKP instances showed that the Amplify Annealing Engine with the proposed post-processing achieved comparable performance with Gurobi and other heuristic methods tailored for the QKP, which is an important step toward practical utility of Ising machines.
Future work includes the extension of the proposed method to other optimization problems and establishing a reasonable benchmarking framework considering computational time required for Ising machines.

\bibliographystyle{IEEEtran}
\bibliography{ref}

\begin{IEEEbiography}[{\includegraphics[width=1in,height=1.25in,clip,keepaspectratio]{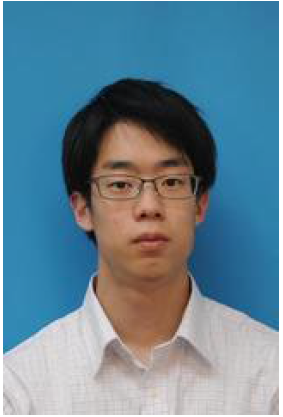}}]{Kentaro Ohno}
is a Ph.~D. student at Waseda University and works at NTT as a researcher.
He received the B.~Sci. and M.~Sci. degrees in mathematics from the University of Tokyo in 2017 and 2019, respectively. He is currently studying combinatorial optimization using Ising machines.
\end{IEEEbiography}

\begin{IEEEbiography}[{\includegraphics[width=1in,height=1.25in,clip,keepaspectratio]{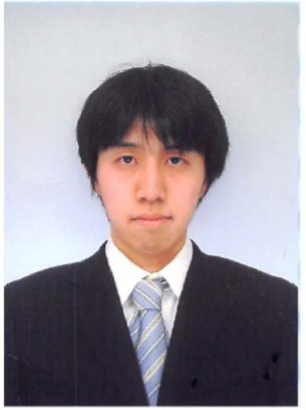}}]{Tatsuhiko Shirai}
received B.~Sci., M.~Sci., and Dr.~Sci. degrees from The University of Tokyo in 2011, 2013, and 2016, respectively. He is presently an assistant professor in the Waseda Institute for Advanced Study, Waseda University. His research interests are quantum dynamics, statistical mechanics, and computational science. He is a member of JPS.
\end{IEEEbiography}

\begin{IEEEbiography}[{\includegraphics[width=1in,height=1.25in,clip,keepaspectratio]{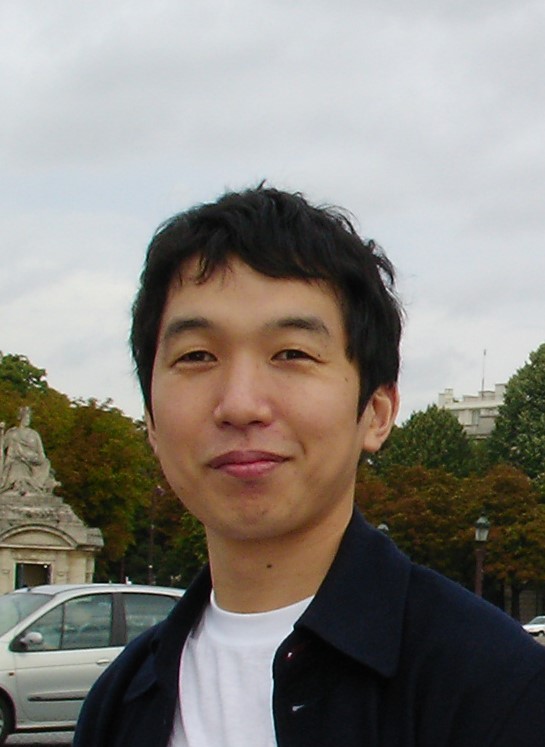}}]{Nozomu Togawa}
received the B.~Eng., M.~Eng., and Dr.~Eng. degrees from Waseda University in 1992, 1994, and 1997, respectively, all in electrical engineering. He is presently a professor in the Department of Computer Science and Communications Engineering, Waseda University. His research interests are quantum computation and integrated system design. He is a member of ACM, IEICE, and IPSJ.
\end{IEEEbiography}

\appendices

\section{Proof}\label{app:proof}

We provide a proof of Proposition~\ref{prop:qkp_property}.

\begin{proposition}
    For a QKP instance defined as Eq.~(\ref{eq:qkp}),
    an optimum is attained by a solution $x \in \{0,1\}^n$ satisfying $C-\max_i w_i < \sum_{i=1}^n w_i x_i \le C$.
\end{proposition}

\begin{proof}
Assume an optimal solution $x\in \{0,1\}^n$ satisfies $\sum_{i=1}^n w_i x_i \le C-\max_i w_i$.
We take another solution $\tilde x\in \{0,1\}^n$ obtained by 
changing the value of $x_j$ to $1$ for arbitrarily chosen $j$ such that $x_j=0$.
Note that such $j$ exists since we assume $C<\sum_i w_i$.
Note also that $\tilde x$ is a feasible solution since $\sum_{i=1}^n w_i \tilde x_i 
 = \sum_{i=1}^n w_i x_i + w_j \le  C-\max_i w_i + w_j \le C$.
Let $\phi(x)$ denote the objective value for $x$.
Since profits $p_{ij}, p_{i}$ are non-negative, we have $\phi(x) \le \phi(\tilde x)$.
Since $x$ is optimal, we get $\phi(x) = \phi(\tilde x)$ and thus $\tilde x$ is also optimal.
We replace $x$ with $\tilde x$ and repeat the same procedure, then we obtain an optimal solution satisfying $C-\max_i w_i < \sum_{i=1}^n w_i x_i \le C$.

\end{proof}

\section{Full Results of Experiments}

\begin{figure*}[t]
     \centering
     \subfloat[$n=100$\label{fig:sa_fillup_n100}]{
         \includegraphics[width=0.31\textwidth]{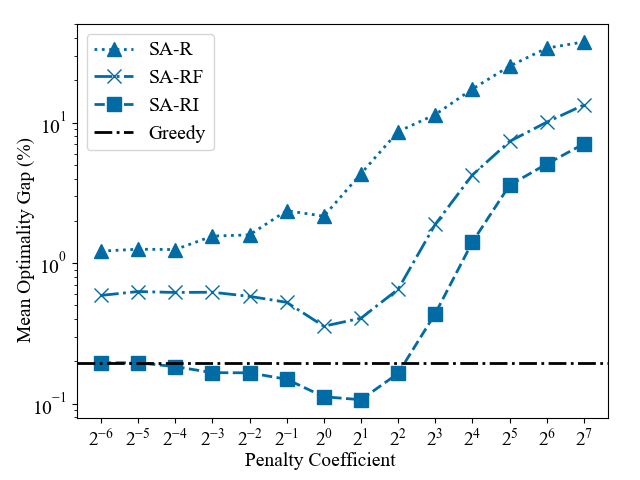}
     }
     \hfil
     \subfloat[$n=200$\label{fig:sa_fillup_n200}]{
         \includegraphics[width=0.31\textwidth]{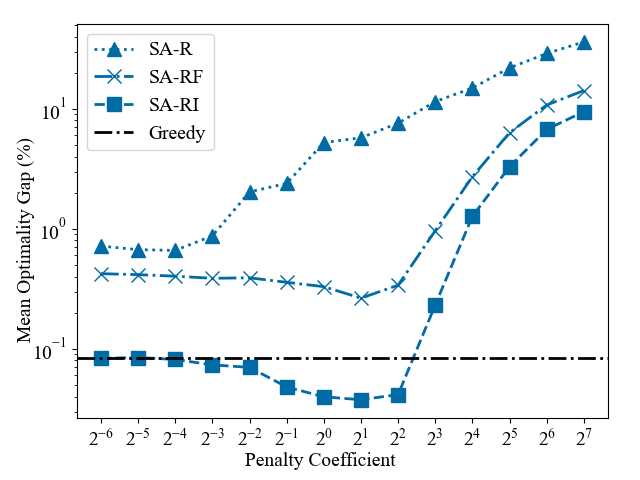}
     }
     \hfil
     \subfloat[$n=300$\label{fig:sa_fillup_3100}]{
         \includegraphics[width=0.31\textwidth]{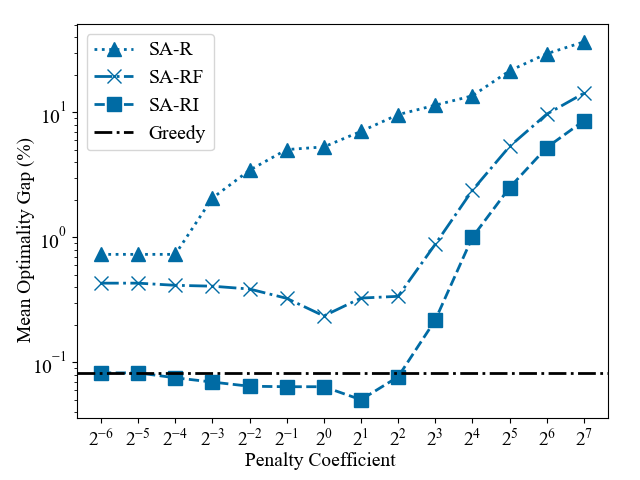}
     }
    \caption{
    Optimality gap for each problem size $n$ of QKP instances.
    Optimality gap for SA is plotted only for $\lambda$ producing feasible solutions on all instances.
    SA-RF denotes SA-R followed by fill-up operation, which produces locally optimal solutions.
    Fill-up operation improves solutions of SA-R particularly around $\lambda=2$, which is optimal penalty coefficient for SA-RI.    }
    \label{fig:sa_small_fillup_appendix}
\end{figure*}

\begin{figure*}[t]
     \centering
     \subfloat[Rate of Feasible Solutions\label{fig:ae_encoding_feasibility}]{
         \includegraphics[width=0.31\textwidth]{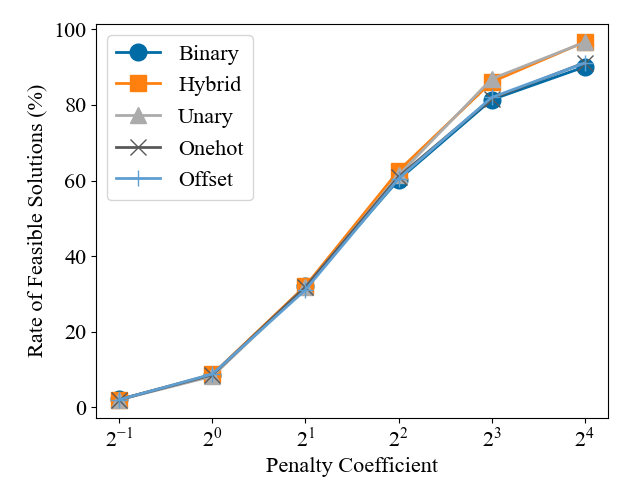}
     }
     \hfil
     \subfloat[Optimality Gap of AE\label{fig:ae_encoding_raw}]{
         \includegraphics[width=0.31\textwidth]{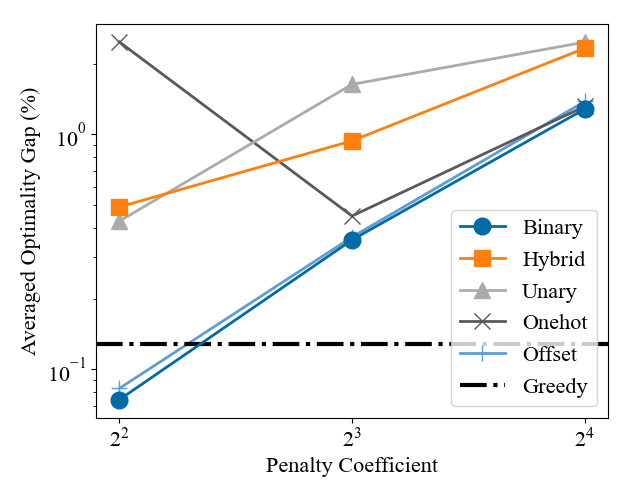}
     }
     \hfil
     \subfloat[Optimality Gap with Post-processing\label{fig:ae_encoding_score}]{
         \includegraphics[width=0.31\textwidth]{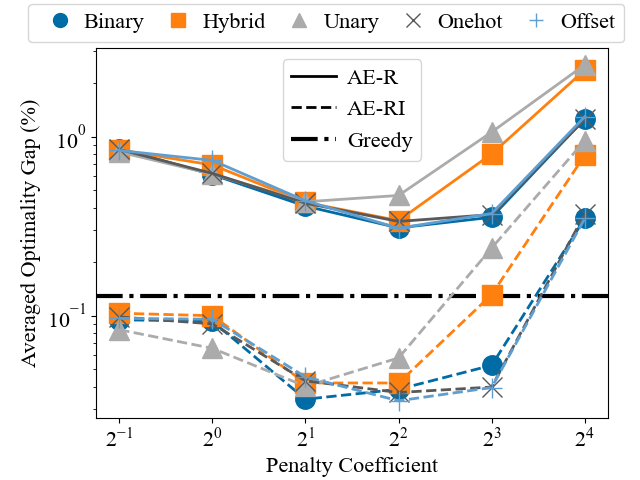}
     }
    \caption{\hl{
    Performance comparison among various encoding methods of inequality constraint on 100 medium-sized QKP instances using AE (cf. Fig.~\ref{fig:sa_compare_encoding}).
    (a)(b) Choice of encoding methods controls trade-off between rates of feasible solutions and objective values. Raw optimality gap is shown over only penalty coefficients attaining at least one feasible solution on more than 50 instances.
    (c) Solving performance of proposed method is much less dependent on choice of encoding methods.
    }}
    \label{fig:ae_compare_encoding}
\end{figure*}

\subsection{
\hl{Auxiliary Experiments}
}\label{app:aux_experiments}

\subsubsection{
\hl{Analysis of transition of best penalty coefficients}
}

We saw in Section~\ref{subsec:sa_small_result} in the main text that
$\lambda$ minimizing the optimality gap of SA-R and that of SA-RI can completely differ: $\lambda_{\text{SA-R}}$ for SA-R is near 0 and $\lambda_{\text{SA-RI}}$ for SA-RI is around 2 for all problem size $n$ in Fig.~\ref{fig:sa_small_plot}.
The result leads to apparently strange inconsistency that the outputs of SA-R around $\lambda=2$ can be improved to good solutions, but themselves are far from optimal.
We hypothesized that this is because SA-R outputs solutions distant from local optima particularly when $\lambda$ is around 2. 
To see this, we plot the optimality gap for SA-R followed by only the fill-up operation (Section~\ref{sec:proposed_method}), which we call SA-RF, in Fig.~\ref{fig:sa_small_fillup_appendix}.
The optimality gap of SA-RF attains its minimum around $\lambda=2$, which is similar to SA-RI, and the difference between SA-R and SA-RF is significantly large there.
Since the fill-up operation makes a solution locally optimal,
the result implies that the solutions obtained with SA-R are far from local optima around $\lambda=2$.
This finding contains an important suggestion on the use of Ising machines: by carefully tuning the penalty coefficient, we can obtain a solution that is itself not good but \emph{globally} (i.e., up to greedy local operations) near-optimal.
We also note that the gap between SA-R and SA-RF cannot be easily filled by emphasizing the local search phase in SA, e.g., by lowering the end temperature.
Indeed, 
we conducted additional experiments with the end temperature of SA lowered to 0.01, and got almost the same results.
This indicates that the inability to get locally optimal solutions cannot be easily resolved by tuning annealing schedules.
This is because the local operations on the QUBO problem do not correspond to those on the QKP, as described in Section~\ref{sec:preliminaries} in the main text.

\subsubsection{
\hl{Comparing encoding methods on an Ising machine}
}

\hl{
We conducted the same experiments using various encoding methods of the inequality constraints as in Section~\ref{subsec:compare_encoding} with AE.
The performance behavior over various penalty coefficients is shown in Fig.~\ref{fig:ae_compare_encoding}.
The penalty coefficient $\lambda$ is varied as $\lambda = 2^i, i\in \{-1,\cdots, 4\}$. 
As in the case of SA, we again observe that the unary and hybrid encodings achieve higher FS rate, but with the smaller gap.
We observe a slight difference in the behavior of AE-RI curves compared with those of SA-RI (Fig.~\ref{fig:sa_compare_encoding}) in that the unary and hybrid encodings obtain larger optimality gap than the other encoding methods for large penalty coefficients.
Since the specification of AE is undisclosed, we could not investigate the cause of this phenomenon further.
In any case, all encoding methods achieve similar optimality gap at the best parameter coefficient.
We further evaluate the accuracy for each method by the number of instances optimally solved and optimality gap, using the best parameter tuned based on the rescaling Eq.~(\ref{eq:penalty_scale}).
The results are shown in Table~\ref{tab:compare_optimal_count_ae} and Table~\ref{tab:compare_optimal_gap_ae}.
Again, we observe there is almost no performance gap among the encoding methods. 
Overall, the results show the robustness of the proposed method also for the real Ising machine.
}

\begin{table}[t]
  \caption{\hl{
  Number of Instances Optimally Solved Using AE-RI.
  }}
  \label{tab:compare_optimal_count_ae}
  \centering
  \begin{tabular}{c|ccccc}
    \hline
    $n\_d$ &  Binary  &  Hybrid  &  Unary  &  One-hot  &  Offset \\
    \hline \hline
    100\_25  & 10 & 10 & 10 & 10 & 10 \\
    100\_50  & 10 & 10 & 10 & 10 & 10 \\
    100\_75  & 10 & 10 & 10 & 10 & 10 \\
    100\_100  & 10 & 10 & 10 & 10 & 10 \\
    200\_25  & 10 & 9 & 10 & 9 & 9 \\
    200\_50  & 10 & 10 & 9 & 10 & 10 \\
    200\_75  & 10 & 10 & 10 & 9 & 10 \\
    200\_100  & 10 & 10 & 10 & 10 & 10 \\
    300\_25  & 10 & 10 & 10 & 10 & 10 \\
    300\_50  & 10 & 10 & 10 & 10 & 10 \\
    \hline
    Total    & \textbf{100} & 99 & 99 & 98 & 99 \\
    \hline
  \end{tabular}
\end{table}

\begin{table}[t]
  \caption{\hl{Averaged Optimality Gap ($\times 0.01$ \%) Solved Using AE-RI.}}
  \label{tab:compare_optimal_gap_ae}
  \centering
  \begin{tabular}{c|rrrrr}
    \hline
    $n\_d$  &  Binary  &  Hybrid  &  Unary  &  One-hot  &  Offset \\
    \hline \hline
    100\_25  & 0.000 & 0.000 & 0.000 & 0.000 & 0.000 \\
    100\_50  & 0.000 & 0.000 & 0.000 & 0.000 & 0.000 \\
    100\_75  & 0.000 & 0.000 & 0.000 & 0.000 & 0.000 \\
    100\_100  & 0.000 & 0.000 & 0.000 & 0.000 & 0.000 \\
    200\_25  & 0.000 & 0.338 & 0.000 & 0.338 & 3.390 \\
    200\_50  & 0.000 & 0.000 & 0.095 & 0.000 & 0.000 \\
    200\_75  & 0.000 & 0.000 & 0.000 & 0.072 & 0.000 \\
    200\_100  & 0.000 & 0.000 & 0.000 & 0.000 & 0.000 \\
    300\_25  & 0.000 & 0.000 & 0.000 & 0.000 & 0.000 \\
    300\_50  & 0.000 & 0.000 & 0.000 & 0.000 & 0.000 \\
    \hline
    Mean  & \textbf{0.000} & 0.034 & 0.009 & 0.041 & 0.339 \\
    \hline
  \end{tabular}
\end{table}

\clearpage

\subsection{Full Results on Simulated Annealing}\label{app:sa_full_results}

Full results of the simulation experiments on medium-sized instances conducted in Section~\ref{sec:simulation} are shown in Table~\ref{tab:full_sa_1}.
In addition to the best objective value over 10 solutions obtained, we show a success rate, that is, the rate of the number of times to hit the optimum out of 10 executions, and the mean objective value.
$\lambda$ denotes the optimal penalty coefficient for each method based on a lexicographic order for tuples of the best objective value, success rate, and mean objective value (e.g., if two values of $\lambda$ have the same best objective, then their success rates are compared).
If multiple values of $\lambda$ obtain the same values for all metrics, then a smaller one is reported.
`FS' stands for the rate of feasible solutions over 10 outputs of SA.
`Mean' and `Best' denote the mean and best objective values over 10 solutions for the optimal $\lambda$, respectively.
`Gap' is the optimality gap computed by the best objective value and known optimal value.
`SR' stands for a success rate.
SR takes positive values only when Gap attains zero.
Note that for instances that can be optimally solved by the greedy method, the optimal $\lambda$ for SA-RI takes the minimal value $2^{-6}$ among those tested.
This is because an output of SA-RI under $\lambda\to 0$ coincides with that of the greedy algorithm as explained in Section~\ref{sec:proposed_method}.

The full results of the performance comparison of SA-RI over various encoding methods conducted in Section~\ref{subsec:compare_encoding} are shown in Table~\ref{tab:full_sa_compare_enc_1}.
Beyond the similarity of averaged performance in the main text, we also see instance-wise similarity:
if an instance cannot be optimally solved by some encoding method, there tends to be another encoding that cannot reach optimality on that instance.
We also observe that the hybrid encoding has a relatively large optimality gap on instance 100\_25\_3, and so do the unary and offset encodings on instance 100\_100\_7.
This is because the optima on these instances are small compared to other instances.
Such a large optimality gap on one instance has a large effect on the mean values in Table~\ref{tab:compare_optimal_gap_sa}.
Due to this instability of the optimality gap as a performance metric, it might be hard to conclude the best encoding method for the proposed method.

We explain some details on the estimation of the optimal penalty coefficient $\lambda_{\text{SA-RI}}$ conducted in Section~\ref{subsec:optimal_penalty}.
We excluded instances that can be optimally solved by the greedy method,
since the optimal values of $\lambda_{\text{SA-RI}}$ on those instances are trivially $0$ and so considered as the outlier value for the analysis.
We used \texttt{LinearRegression} module in the scikit-learn library of version 1.2.2 for the regression.

To get insight into the behavior of the optimal penalty coefficient $\lambda$ for SA-RI, we plot the optimality gap for each instance on all $\lambda$ as a heat map in Fig.~\ref{fig:heatmap_small_sa}.
Here we slightly modify the performance metric as 
\begin{align*}
    &\text{Aggregated \ Optimality \ Gap} = \\
    &\frac{S_\mathrm{best} - S + (1-\mathrm{Success \ Rate})}{S_\mathrm{best}} \times 100 \ (\%),
\end{align*}
where $S$ is 
the best objective value obtained by SA-RI, Success Rate is the rate of hitting the optimum ($\le 1$), and $S_\mathrm{best}$ is the known optimum.
The aggregated optimality gap takes zero if and only if all solutions obtained are optimal.
Note also that $\lambda$ minimizing the aggregated optimality gap corresponds (if it is unique) to $\lambda$ shown in Table~\ref{tab:full_sa_1} since $S$ is an integer and Success Rate is less than 1 unless $S=S_\mathrm{best}$.
In Fig.~\ref{fig:heatmap_small_sa}, good penalty coefficients correspond to an area of dark colors.
We observe that the tested range of $\lambda$ is sufficiently wide to cover good penalty coefficients for each instance.
We see a trend that the area slightly moves to the right as $d$ gets large, which agrees with the analysis result in Section~\ref{subsec:optimal_penalty}.

\begin{landscape}

\begin{table}[t]
  \caption{Full Results of Simulated Annealing.}
    \label{tab:full_sa_1}
  \centering
\resizebox{\linewidth}{!}{
  \begin{tabular}{@{\hspace{1pt}}c@{\hspace{4pt}}r@{\hspace{4pt}}|
  r@{\hspace{6pt}}r
  @{\hspace{6pt}}r@{\hspace{6pt}} 
  r@{\hspace{6pt}}r@{\hspace{6pt}}r@{\hspace{6pt}}r@{\hspace{6pt}}r@{\hspace{6pt}}r 
  @{\hspace{6pt}}r@{\hspace{6pt}} 
  r@{\hspace{6pt}}r@{\hspace{6pt}}r@{\hspace{6pt}}r@{\hspace{6pt}}r
  @{\hspace{6pt}}r@{\hspace{6pt}} 
  r@{\hspace{6pt}}r@{\hspace{6pt}}r@{\hspace{6pt}}r@{\hspace{6pt}}r
  @{\hspace{6pt}}r@{\hspace{6pt}}
  r@{\hspace{6pt}}r@{\hspace{6pt}}r@{\hspace{6pt}}r@{\hspace{6pt}}r
  @{\hspace{3pt}}}
    \hline
    \multicolumn{2}{c|}{Instance} & 
    \multicolumn{2}{c}{Greedy} &&
    \multicolumn{6}{c}{SA} && 
    \multicolumn{5}{c}{SA-R} && 
    \multicolumn{5}{c}{SA-I} && 
    \multicolumn{5}{c}{SA-RI} \\
    \cline{6-11}
    \cline{13-17}
    \cline{19-23}
    \cline{25-29}
    $n$\_$d$\_id & Optimal & Score & Gap &  & $\lambda$ & FS & Mean & Best & Gap & SR &  & $\lambda$ & Mean & Best & Gap & SR &  & $\lambda$ & Mean & Best & Gap & SR &  & $\lambda$ & Mean & Best & Gap & SR \\
    \hline \hline
  100\_25\_1 &  18558 &  18511 & 0.253 &  &  $2^{2}$ & 0.5 &  17063.40 &  17456 &  5.938 & 0.0 &  & $2^{-1}$ &  18250.90 &  18514 &  0.237 & 0.0 &  &  $2^{2}$ &  18137.40 &  18325 &  1.256 & 0.0 &  & $2^{-3}$ &  18558.00 &  18558 &  0.000 & 1.0 \\
  100\_25\_2 &  56525 &  56525 & 0.000 &  &  $2^{0}$ & 1.0 &  41780.10 &  48803 & 13.661 & 0.0 &  & $2^{-6}$ &  55578.00 &  55578 &  1.675 & 0.0 &  & $2^{-1}$ &  56491.10 &  56525 &  0.000 & 0.7 &  & $2^{-6}$ &  56525.00 &  56525 &  0.000 & 1.0 \\
  100\_25\_3 &   3752 &   3702 & 1.333 &  &  $2^{3}$ & 0.4 &   3263.00 &   3646 &  2.825 & 0.0 &  &  $2^{3}$ &   3072.10 &   3646 &  2.825 & 0.0 &  &  $2^{3}$ &   3667.50 &   3752 &  0.000 & 0.1 &  &  $2^{3}$ &   3635.40 &   3752 &  0.000 & 0.1 \\
  100\_25\_4 &  50382 &  50382 & 0.000 &  & $2^{-1}$ & 0.7 &  39337.29 &  44631 & 11.415 & 0.0 &  & $2^{-6}$ &  50382.00 &  50382 &  0.000 & 1.0 &  & $2^{-1}$ &  50356.14 &  50382 &  0.000 & 0.2 &  & $2^{-6}$ &  50382.00 &  50382 &  0.000 & 1.0 \\
  100\_25\_5 &  61494 &  61494 & 0.000 &  & $2^{-3}$ & 1.0 &  50121.40 &  53275 & 13.366 & 0.0 &  & $2^{-6}$ &  60983.00 &  60983 &  0.831 & 0.0 &  & $2^{-2}$ &  61465.50 &  61494 &  0.000 & 0.9 &  & $2^{-6}$ &  61494.00 &  61494 &  0.000 & 1.0 \\
  100\_25\_6 &  36360 &  36189 & 0.470 &  &  $2^{2}$ & 0.7 &  32881.29 &  34982 &  3.790 & 0.0 &  & $2^{-1}$ &  36131.50 &  36360 &  0.000 & 0.1 &  &  $2^{2}$ &  35939.43 &  36198 &  0.446 & 0.0 &  &  $2^{0}$ &  36206.90 &  36360 &  0.000 & 0.1 \\
  100\_25\_7 &  14657 &  14553 & 0.710 &  &  $2^{2}$ & 0.3 &  14082.00 &  14217 &  3.002 & 0.0 &  &  $2^{0}$ &  14312.40 &  14555 &  0.696 & 0.0 &  &  $2^{2}$ &  14338.00 &  14456 &  1.371 & 0.0 &  & $2^{-3}$ &  14657.00 &  14657 &  0.000 & 1.0 \\
  100\_25\_8 &  20452 &  20307 & 0.709 &  &  $2^{2}$ & 0.3 &  18370.67 &  19544 &  4.440 & 0.0 &  &  $2^{0}$ &  19914.40 &  20106 &  1.692 & 0.0 &  &  $2^{2}$ &  20066.67 &  20337 &  0.562 & 0.0 &  &  $2^{1}$ &  20243.30 &  20355 &  0.474 & 0.0 \\
  100\_25\_9 &  35438 &  35365 & 0.206 &  &  $2^{1}$ & 0.6 &  30686.83 &  32783 &  7.492 & 0.0 &  & $2^{-2}$ &  34871.70 &  35438 &  0.000 & 0.2 &  &  $2^{1}$ &  35311.83 &  35357 &  0.229 & 0.0 &  & $2^{-2}$ &  35382.50 &  35438 &  0.000 & 0.4 \\
 100\_25\_10 &  24930 &  24926 & 0.016 &  &  $2^{1}$ & 0.3 &  24656.33 &  24784 &  0.586 & 0.0 &  & $2^{-1}$ &  24763.00 &  24926 &  0.016 & 0.0 &  &  $2^{1}$ &  24854.00 &  24926 &  0.016 & 0.0 &  & $2^{-1}$ &  24901.10 &  24930 &  0.000 & 0.1 \\
  100\_50\_1 &  83742 &  83712 & 0.036 &  &  $2^{2}$ & 0.7 &  68493.57 &  81874 &  2.231 & 0.0 &  &  $2^{0}$ &  76492.50 &  83168 &  0.685 & 0.0 &  &  $2^{2}$ &  83555.43 &  83742 &  0.000 & 0.4 &  &  $2^{2}$ &  83611.40 &  83742 &  0.000 & 0.7 \\
  100\_50\_2 & 104856 & 104770 & 0.082 &  &  $2^{0}$ & 0.8 &  83383.25 &  92819 & 11.480 & 0.0 &  & $2^{-1}$ & 102182.20 & 104766 &  0.086 & 0.0 &  &  $2^{3}$ & 104567.60 & 104856 &  0.000 & 0.2 &  &  $2^{3}$ & 104567.60 & 104856 &  0.000 & 0.2 \\
  100\_50\_3 &  34006 &  33902 & 0.306 &  &  $2^{3}$ & 0.1 &  31130.00 &  31130 &  8.457 & 0.0 &  &  $2^{1}$ &  33462.00 &  33775 &  0.679 & 0.0 &  &  $2^{3}$ &  32738.00 &  32738 &  3.729 & 0.0 &  &  $2^{1}$ &  33936.20 &  34006 &  0.000 & 0.4 \\
  100\_50\_4 & 105996 & 105996 & 0.000 &  &  $2^{0}$ & 0.9 &  82037.33 &  94211 & 11.118 & 0.0 &  & $2^{-6}$ & 105876.00 & 105876 &  0.113 & 0.0 &  &  $2^{1}$ & 105966.30 & 105996 &  0.000 & 0.8 &  & $2^{-6}$ & 105996.00 & 105996 &  0.000 & 1.0 \\
  100\_50\_5 &  56464 &  56464 & 0.000 &  &  $2^{2}$ & 0.1 &  56227.00 &  56227 &  0.420 & 0.0 &  &  $2^{1}$ &  55846.90 &  56388 &  0.135 & 0.0 &  &  $2^{2}$ &  56398.00 &  56398 &  0.117 & 0.0 &  & $2^{-6}$ &  56464.00 &  56464 &  0.000 & 1.0 \\
  100\_50\_6 &  16083 &  16083 & 0.000 &  &  $2^{3}$ & 0.1 &  15091.00 &  15091 &  6.168 & 0.0 &  & $2^{-6}$ &  16080.00 &  16080 &  0.019 & 0.0 &  &  $2^{3}$ &  16083.00 &  16083 &  0.000 & 0.1 &  & $2^{-6}$ &  16083.00 &  16083 &  0.000 & 1.0 \\
  100\_50\_7 &  52819 &  52784 & 0.066 &  &  $2^{3}$ & 0.3 &  44920.33 &  47827 &  9.451 & 0.0 &  & $2^{-6}$ &  52646.00 &  52646 &  0.328 & 0.0 &  &  $2^{3}$ &  49795.67 &  50240 &  4.883 & 0.0 &  &  $2^{1}$ &  52687.30 &  52819 &  0.000 & 0.3 \\
  100\_50\_8 &  54246 &  54030 & 0.398 &  &  $2^{2}$ & 0.1 &  52562.00 &  52562 &  3.104 & 0.0 &  &  $2^{0}$ &  53946.30 &  54176 &  0.129 & 0.0 &  &  $2^{2}$ &  53408.00 &  53408 &  1.545 & 0.0 &  &  $2^{0}$ &  54166.00 &  54246 &  0.000 & 0.3 \\
  100\_50\_9 &  68974 &  68974 & 0.000 &  &  $2^{1}$ & 0.2 &  62401.00 &  68974 &  0.000 & 0.1 &  & $2^{-6}$ &  68974.00 &  68974 &  0.000 & 1.0 &  &  $2^{1}$ &  68974.00 &  68974 &  0.000 & 0.2 &  & $2^{-6}$ &  68974.00 &  68974 &  0.000 & 1.0 \\
 100\_50\_10 &  88634 &  88527 & 0.121 &  &  $2^{3}$ & 1.0 &  63369.90 &  78751 & 11.150 & 0.0 &  &  $2^{0}$ &  68151.10 &  88146 &  0.551 & 0.0 &  &  $2^{3}$ &  88505.20 &  88634 &  0.000 & 0.1 &  &  $2^{3}$ &  88505.20 &  88634 &  0.000 & 0.1 \\
  100\_75\_1 & 189137 & 189137 & 0.000 &  & $2^{-1}$ & 1.0 & 141711.30 & 163058 & 13.788 & 0.0 &  & $2^{-6}$ & 189137.00 & 189137 &  0.000 & 1.0 &  & $2^{-1}$ & 189137.00 & 189137 &  0.000 & 1.0 &  & $2^{-6}$ & 189137.00 & 189137 &  0.000 & 1.0 \\
  100\_75\_2 &  95074 &  94980 & 0.099 &  &  $2^{3}$ & 0.6 &  77125.33 &  90450 &  4.864 & 0.0 &  &  $2^{0}$ &  94276.70 &  94822 &  0.265 & 0.0 &  &  $2^{2}$ &  94980.00 &  94980 &  0.099 & 0.0 &  &  $2^{2}$ &  94956.70 &  95074 &  0.000 & 0.1 \\
  100\_75\_3 &  62098 &  62098 & 0.000 &  &  $2^{4}$ & 0.3 &  52999.67 &  53802 & 13.360 & 0.0 &  &  $2^{1}$ &  62061.60 &  62098 &  0.000 & 0.3 &  &  $2^{4}$ &  55379.33 &  56115 &  9.635 & 0.0 &  & $2^{-6}$ &  62098.00 &  62098 &  0.000 & 1.0 \\
  100\_75\_4 &  72245 &  72167 & 0.108 &  &  $2^{3}$ & 0.3 &  68158.00 &  69522 &  3.769 & 0.0 &  &  $2^{0}$ &  71302.90 &  71995 &  0.346 & 0.0 &  &  $2^{3}$ &  70151.33 &  70843 &  1.941 & 0.0 &  &  $2^{2}$ &  72029.40 &  72245 &  0.000 & 0.3 \\
  100\_75\_5 &  27616 &  27616 & 0.000 &  &  $2^{3}$ & 0.1 &  27543.00 &  27543 &  0.264 & 0.0 &  &  $2^{1}$ &  27354.40 &  27557 &  0.214 & 0.0 &  &  $2^{3}$ &  27616.00 &  27616 &  0.000 & 0.1 &  & $2^{-6}$ &  27616.00 &  27616 &  0.000 & 1.0 \\
  100\_75\_6 & 145273 & 145224 & 0.034 &  &  $2^{2}$ & 0.9 & 116283.78 & 144022 &  0.861 & 0.0 &  &  $2^{0}$ & 131337.00 & 144746 &  0.363 & 0.0 &  &  $2^{3}$ & 143779.20 & 145273 &  0.000 & 0.1 &  &  $2^{3}$ & 143779.20 & 145273 &  0.000 & 0.1 \\
  100\_75\_7 & 110979 & 110451 & 0.476 &  &  $2^{3}$ & 0.8 &  88068.50 & 106227 &  4.282 & 0.0 &  &  $2^{0}$ & 110477.10 & 110718 &  0.235 & 0.0 &  &  $2^{3}$ & 110283.62 & 110937 &  0.038 & 0.0 &  &  $2^{1}$ & 110910.70 & 110979 &  0.000 & 0.5 \\
  100\_75\_8 &  19570 &  19570 & 0.000 &  &  $2^{2}$ & 0.2 &  19547.50 &  19570 &  0.000 & 0.1 &  &  $2^{2}$ &  18693.60 &  19570 &  0.000 & 0.1 &  &  $2^{2}$ &  19570.00 &  19570 &  0.000 & 0.2 &  & $2^{-6}$ &  19570.00 &  19570 &  0.000 & 1.0 \\
  100\_75\_9 & 104341 & 103916 & 0.407 &  &  $2^{3}$ & 0.8 &  85243.88 & 100546 &  3.637 & 0.0 &  &  $2^{0}$ & 103036.60 & 103655 &  0.657 & 0.0 &  &  $2^{3}$ & 102979.00 & 104148 &  0.185 & 0.0 &  &  $2^{2}$ & 103934.40 & 104253 &  0.084 & 0.0 \\
 100\_75\_10 & 143740 & 143695 & 0.031 &  &  $2^{1}$ & 0.9 & 108421.11 & 124507 & 13.380 & 0.0 &  &  $2^{0}$ & 119219.30 & 143740 &  0.000 & 0.1 &  &  $2^{1}$ & 143721.56 & 143740 &  0.000 & 0.7 &  &  $2^{1}$ & 143600.40 & 143740 &  0.000 & 0.7 \\
 100\_100\_1 &  81978 &  81961 & 0.021 &  &  $2^{4}$ & 0.1 &  75824.00 &  75824 &  7.507 & 0.0 &  &  $2^{2}$ &  81867.40 &  81961 &  0.021 & 0.0 &  &  $2^{5}$ &  75186.80 &  80348 &  1.988 & 0.0 &  &  $2^{3}$ &  81803.60 &  81978 &  0.000 & 0.1 \\
 100\_100\_2 & 190424 & 189998 & 0.224 &  &  $2^{1}$ & 0.7 & 148425.86 & 167736 & 11.914 & 0.0 &  &  $2^{0}$ & 167000.80 & 190338 &  0.045 & 0.0 &  &  $2^{0}$ & 190170.40 & 190424 &  0.000 & 0.1 &  &  $2^{0}$ & 190230.50 & 190424 &  0.000 & 0.1 \\
 100\_100\_3 & 225434 & 225434 & 0.000 &  &  $2^{1}$ & 1.0 & 167360.40 & 201873 & 10.451 & 0.0 &  & $2^{-1}$ & 211582.00 & 225434 &  0.000 & 0.4 &  &  $2^{1}$ & 225434.00 & 225434 &  0.000 & 1.0 &  & $2^{-6}$ & 225434.00 & 225434 &  0.000 & 1.0 \\
 100\_100\_4 &  63028 &  63028 & 0.000 &  &  $2^{4}$ & 0.1 &  54077.00 &  54077 & 14.202 & 0.0 &  & $2^{-6}$ &  63028.00 &  63028 &  0.000 & 1.0 &  &  $2^{5}$ &  53484.75 &  58212 &  7.641 & 0.0 &  & $2^{-6}$ &  63028.00 &  63028 &  0.000 & 1.0 \\
 100\_100\_5 & 230076 & 230076 & 0.000 &  &  $2^{0}$ & 0.6 & 166587.67 & 184411 & 19.848 & 0.0 &  & $2^{-6}$ & 229660.00 & 229660 &  0.181 & 0.0 &  &  $2^{2}$ & 230070.20 & 230076 &  0.000 & 0.9 &  & $2^{-6}$ & 230076.00 & 230076 &  0.000 & 1.0 \\
 100\_100\_6 &  74358 &  74358 & 0.000 &  &  $2^{4}$ & 0.1 &  68717.00 &  68717 &  7.586 & 0.0 &  &  $2^{1}$ &  74035.30 &  74103 &  0.343 & 0.0 &  &  $2^{3}$ &  74358.00 &  74358 &  0.000 & 0.1 &  & $2^{-6}$ &  74358.00 &  74358 &  0.000 & 1.0 \\
 100\_100\_7 &  10330 &  10184 & 1.413 &  &  $2^{5}$ & 0.2 &   8249.50 &   9664 &  6.447 & 0.0 &  & $2^{-6}$ &  10108.00 &  10108 &  2.149 & 0.0 &  &  $2^{6}$ &   8256.86 &   9731 &  5.799 & 0.0 &  &  $2^{5}$ &   9717.40 &  10330 &  0.000 & 0.1 \\
 100\_100\_8 &  62582 &  62422 & 0.256 &  &  $2^{3}$ & 0.2 &  62427.00 &  62457 &  0.200 & 0.0 &  &  $2^{1}$ &  62274.30 &  62484 &  0.157 & 0.0 &  &  $2^{3}$ &  62582.00 &  62582 &  0.000 & 0.2 &  &  $2^{3}$ &  62509.90 &  62582 &  0.000 & 0.6 \\
 100\_100\_9 & 232754 & 232693 & 0.026 &  &  $2^{1}$ & 1.0 & 164711.60 & 208886 & 10.255 & 0.0 &  & $2^{-6}$ & 231402.00 & 231402 &  0.581 & 0.0 &  &  $2^{6}$ & 232741.80 & 232754 &  0.000 & 0.8 &  &  $2^{6}$ & 232741.80 & 232754 &  0.000 & 0.8 \\
100\_100\_10 & 193262 & 193218 & 0.023 &  &  $2^{2}$ & 0.8 & 138825.50 & 166966 & 13.606 & 0.0 &  &  $2^{1}$ & 142628.90 & 192352 &  0.471 & 0.0 &  &  $2^{4}$ & 193174.10 & 193262 &  0.000 & 0.5 &  &  $2^{4}$ & 193174.10 & 193262 &  0.000 & 0.5 \\
  200\_25\_1 & 204441 & 204399 & 0.021 &  & $2^{-1}$ & 1.0 & 156096.80 & 180652 & 11.636 & 0.0 &  & $2^{-6}$ & 203061.00 & 203061 &  0.675 & 0.0 &  & $2^{-1}$ & 204397.00 & 204441 &  0.000 & 0.2 &  & $2^{-1}$ & 204402.20 & 204441 &  0.000 & 0.2 \\
  200\_25\_2 & 239573 & 239573 & 0.000 &  & $2^{-3}$ & 0.7 & 184058.14 & 199793 & 16.605 & 0.0 &  & $2^{-6}$ & 238734.00 & 238734 &  0.350 & 0.0 &  & $2^{-2}$ & 239564.00 & 239573 &  0.000 & 0.4 &  & $2^{-6}$ & 239573.00 & 239573 &  0.000 & 1.0 \\
  200\_25\_3 & 245463 & 244446 & 0.414 &  & $2^{-3}$ & 1.0 & 192841.00 & 223037 &  9.136 & 0.0 &  & $2^{-6}$ & 243967.00 & 243967 &  0.609 & 0.0 &  &  $2^{4}$ & 244828.60 & 245463 &  0.000 & 0.2 &  &  $2^{4}$ & 244828.60 & 245463 &  0.000 & 0.2 \\
  200\_25\_4 & 222361 & 221591 & 0.346 &  & $2^{-1}$ & 1.0 & 160720.70 & 198848 & 10.574 & 0.0 &  & $2^{-6}$ & 221054.00 & 221054 &  0.588 & 0.0 &  & $2^{-2}$ & 222314.20 & 222361 &  0.000 & 0.8 &  & $2^{-2}$ & 222314.20 & 222361 &  0.000 & 0.8 \\
  200\_25\_5 & 187324 & 187315 & 4.8E-3 &  &  $2^{1}$ & 1.0 & 145902.90 & 173607 &  7.323 & 0.0 &  & $2^{-2}$ & 168024.20 & 186861 &  0.247 & 0.0 &  &  $2^{0}$ & 187256.80 & 187324 &  0.000 & 0.4 &  &  $2^{0}$ & 187256.80 & 187324 &  0.000 & 0.4 \\
  200\_25\_6 &  80351 &  80276 & 0.093 &  &  $2^{2}$ & 0.1 &  75841.00 &  75841 &  5.613 & 0.0 &  & $2^{-5}$ &  80229.00 &  80229 &  0.152 & 0.0 &  &  $2^{2}$ &  78282.00 &  78282 &  2.575 & 0.0 &  & $2^{-2}$ &  80279.40 &  80310 &  0.051 & 0.0 \\
  200\_25\_7 &  59036 &  58858 & 0.302 &  &  $2^{3}$ & 0.2 &  52897.50 &  53754 &  8.947 & 0.0 &  &  $2^{0}$ &  58597.30 &  59002 &  0.058 & 0.0 &  &  $2^{3}$ &  55971.00 &  56795 &  3.796 & 0.0 &  &  $2^{1}$ &  58938.20 &  59036 &  0.000 & 0.5 \\
  200\_25\_8 & 149433 & 149125 & 0.206 &  &  $2^{3}$ & 1.0 & 128227.00 & 141491 &  5.315 & 0.0 &  &  $2^{0}$ & 143627.50 & 148952 &  0.322 & 0.0 &  &  $2^{1}$ & 149319.50 & 149384 &  0.033 & 0.0 &  & $2^{-1}$ & 149208.20 & 149433 &  0.000 & 0.2 \\
  200\_25\_9 &  49366 &  49366 & 0.000 &  &  $2^{3}$ & 0.3 &  44759.33 &  45303 &  8.230 & 0.0 &  & $2^{-6}$ &  49293.00 &  49293 &  0.148 & 0.0 &  &  $2^{3}$ &  47749.00 &  47984 &  2.799 & 0.0 &  & $2^{-6}$ &  49366.00 &  49366 &  0.000 & 1.0 \\
 200\_25\_10 &  48459 &  48291 & 0.347 &  &  $2^{3}$ & 0.3 &  41439.67 &  43291 & 10.665 & 0.0 &  &  $2^{0}$ &  48267.20 &  48412 &  0.097 & 0.0 &  &  $2^{3}$ &  45235.00 &  45870 &  5.343 & 0.0 &  &  $2^{0}$ &  48358.20 &  48459 &  0.000 & 0.4 \\
     \hline
  \end{tabular}
}
\end{table}
\end{landscape}
\addtocounter{table}{-1}
\begin{landscape}
\begin{table}[t]
  \caption{Full Results of Simulated Annealing (Continued).}
  \label{tab:full_sa_2}
  \centering
\resizebox{\linewidth}{!}{
  \begin{tabular}{@{\hspace{1pt}}c@{\hspace{4pt}}r@{\hspace{4pt}}|
  r@{\hspace{6pt}}r
  @{\hspace{6pt}}r@{\hspace{6pt}} 
  r@{\hspace{6pt}}r@{\hspace{6pt}}r@{\hspace{6pt}}r@{\hspace{6pt}}r@{\hspace{6pt}}r 
  @{\hspace{6pt}}r@{\hspace{6pt}} 
  r@{\hspace{6pt}}r@{\hspace{6pt}}r@{\hspace{6pt}}r@{\hspace{6pt}}r
  @{\hspace{6pt}}r@{\hspace{6pt}} 
  r@{\hspace{6pt}}r@{\hspace{6pt}}r@{\hspace{6pt}}r@{\hspace{6pt}}r
  @{\hspace{6pt}}r@{\hspace{6pt}}
  r@{\hspace{6pt}}r@{\hspace{6pt}}r@{\hspace{6pt}}r@{\hspace{6pt}}r
  @{\hspace{3pt}}}
    \hline
    \multicolumn{2}{c|}{Instance} & 
    \multicolumn{2}{c}{Greedy} &&
    \multicolumn{6}{c}{SA} && 
    \multicolumn{5}{c}{SA-R} && 
    \multicolumn{5}{c}{SA-I} && 
    \multicolumn{5}{c}{SA-RI} \\
    \cline{6-11}
    \cline{13-17}
    \cline{19-23}
    \cline{25-29}
    $n$\_$d$\_id & Optimal & Score & Gap &  & $\lambda$ & FS & Mean & Best & Gap & SR &  & $\lambda$ & Mean & Best & Gap & SR &  & $\lambda$ & Mean & Best & Gap & SR &  & $\lambda$ & Mean & Best & Gap & SR \\
    \hline \hline
  200\_50\_1 & 372097 & 372097 & 0.000 &  &  $2^{0}$ & 0.9 & 294501.33 & 326397 & 12.282 & 0.0 &  & $2^{-6}$ & 372097.00 & 372097 &  0.000 & 1.0 &  &  $2^{1}$ & 372097.00 & 372097 &  0.000 & 0.9 &  & $2^{-6}$ & 372097.00 & 372097 &  0.000 & 1.0 \\
  200\_50\_2 & 211130 & 210485 & 0.305 &  &  $2^{4}$ & 0.3 & 176418.67 & 193428 &  8.384 & 0.0 &  &  $2^{1}$ & 209220.60 & 210624 &  0.240 & 0.0 &  &  $2^{4}$ & 202804.00 & 205228 &  2.795 & 0.0 &  &  $2^{1}$ & 210838.30 & 211090 &  0.019 & 0.0 \\
  200\_50\_3 & 227185 & 227185 & 0.000 &  &  $2^{4}$ & 0.3 & 205635.67 & 214039 &  5.786 & 0.0 &  &  $2^{1}$ & 226135.00 & 226428 &  0.333 & 0.0 &  &  $2^{3}$ & 225172.00 & 225172 &  0.886 & 0.0 &  & $2^{-6}$ & 227185.00 & 227185 &  0.000 & 1.0 \\
  200\_50\_4 & 228572 & 228572 & 0.000 &  &  $2^{3}$ & 0.1 & 219575.00 & 219575 &  3.936 & 0.0 &  & $2^{-6}$ & 228572.00 & 228572 &  0.000 & 1.0 &  &  $2^{3}$ & 223641.00 & 223641 &  2.157 & 0.0 &  & $2^{-6}$ & 228572.00 & 228572 &  0.000 & 1.0 \\
  200\_50\_5 & 479651 & 479451 & 0.042 &  & $2^{-2}$ & 0.9 & 370026.67 & 426263 & 11.131 & 0.0 &  & $2^{-6}$ & 479004.00 & 479004 &  0.135 & 0.0 &  &  $2^{7}$ & 479436.40 & 479451 &  0.042 & 0.0 &  & $2^{-6}$ & 479451.00 & 479451 &  0.042 & 0.0 \\
  200\_50\_6 & 426777 & 426436 & 0.080 &  & $2^{-2}$ & 0.5 & 347722.20 & 348287 & 18.391 & 0.0 &  & $2^{-6}$ & 425835.00 & 425835 &  0.221 & 0.0 &  & $2^{-2}$ & 426657.00 & 426657 &  0.028 & 0.0 &  & $2^{-1}$ & 426601.20 & 426657 &  0.028 & 0.0 \\
  200\_50\_7 & 220890 & 220806 & 0.038 &  &  $2^{3}$ & 0.1 & 201893.00 & 201893 &  8.600 & 0.0 &  & $2^{-1}$ & 218989.00 & 220456 &  0.196 & 0.0 &  &  $2^{2}$ & 220842.00 & 220842 &  0.022 & 0.0 &  &  $2^{0}$ & 220683.20 & 220842 &  0.022 & 0.0 \\
  200\_50\_8 & 317952 & 317880 & 0.023 &  &  $2^{3}$ & 1.0 & 256469.50 & 288786 &  9.173 & 0.0 &  & $2^{-1}$ & 317150.10 & 317742 &  0.066 & 0.0 &  &  $2^{0}$ & 317952.00 & 317952 &  0.000 & 0.1 &  & $2^{-3}$ & 317952.00 & 317952 &  0.000 & 1.0 \\
  200\_50\_9 & 104936 & 104936 & 0.000 &  &  $2^{4}$ & 0.3 &  95284.33 &  98644 &  5.996 & 0.0 &  &  $2^{1}$ & 103506.60 & 104913 &  0.022 & 0.0 &  &  $2^{4}$ &  99015.33 & 101660 &  3.122 & 0.0 &  & $2^{-6}$ & 104936.00 & 104936 &  0.000 & 1.0 \\
 200\_50\_10 & 284751 & 284741 & 3.5E-3 &  &  $2^{3}$ & 0.7 & 259466.86 & 277449 &  2.564 & 0.0 &  &  $2^{0}$ & 281643.60 & 284741 & 3.5E-3 & 0.0 &  &  $2^{2}$ & 283984.00 & 283984 &  0.269 & 0.0 &  &  $2^{2}$ & 284054.30 & 284751 &  0.000 & 0.1 \\
  200\_75\_1 & 442894 & 442423 & 0.106 &  &  $2^{4}$ & 0.7 & 383841.14 & 426329 &  3.740 & 0.0 &  &  $2^{0}$ & 439884.10 & 441263 &  0.368 & 0.0 &  &  $2^{4}$ & 437349.86 & 442602 &  0.066 & 0.0 &  &  $2^{2}$ & 441433.60 & 442862 & 7.2E-3 & 0.0 \\
  200\_75\_2 & 286643 & 286632 & 3.8E-3 &  &  $2^{4}$ & 0.2 & 268944.00 & 274508 &  4.233 & 0.0 &  &  $2^{1}$ & 284667.30 & 286615 & 9.8E-3 & 0.0 &  &  $2^{4}$ & 279890.50 & 281342 &  1.849 & 0.0 &  &  $2^{2}$ & 286183.10 & 286643 &  0.000 & 0.2 \\
  200\_75\_3 &  61924 &  61924 & 0.000 &  &  $2^{5}$ & 0.1 &  59806.00 &  59806 &  3.420 & 0.0 &  &  $2^{3}$ &  61171.00 &  61475 &  0.725 & 0.0 &  &  $2^{5}$ &  60365.00 &  60365 &  2.518 & 0.0 &  & $2^{-6}$ &  61924.00 &  61924 &  0.000 & 1.0 \\
  200\_75\_4 & 128351 & 128351 & 0.000 &  &  $2^{5}$ & 0.3 &  94229.67 &  96894 & 24.509 & 0.0 &  &  $2^{1}$ & 128196.50 & 128351 &  0.000 & 0.6 &  &  $2^{5}$ & 102443.67 & 104265 & 18.766 & 0.0 &  & $2^{-6}$ & 128351.00 & 128351 &  0.000 & 1.0 \\
  200\_75\_5 & 137885 & 137764 & 0.088 &  &  $2^{5}$ & 0.3 & 120425.00 & 123930 & 10.121 & 0.0 &  &  $2^{2}$ & 137000.80 & 137690 &  0.141 & 0.0 &  &  $2^{5}$ & 125643.00 & 127511 &  7.524 & 0.0 &  &  $2^{3}$ & 137751.30 & 137885 &  0.000 & 0.2 \\
  200\_75\_6 & 229631 & 229250 & 0.166 &  &  $2^{5}$ & 0.1 & 204220.00 & 204220 & 11.066 & 0.0 &  &  $2^{2}$ & 227701.70 & 229186 &  0.194 & 0.0 &  &  $2^{5}$ & 215268.00 & 215268 &  6.255 & 0.0 &  &  $2^{2}$ & 229316.40 & 229631 &  0.000 & 0.4 \\
  200\_75\_7 & 269887 & 269887 & 0.000 &  &  $2^{4}$ & 0.1 & 255693.00 & 255693 &  5.259 & 0.0 &  &  $2^{1}$ & 267588.90 & 269558 &  0.122 & 0.0 &  &  $2^{4}$ & 261943.00 & 261943 &  2.943 & 0.0 &  & $2^{-6}$ & 269887.00 & 269887 &  0.000 & 1.0 \\
  200\_75\_8 & 600858 & 600806 & 8.7E-3 &  &  $2^{0}$ & 1.0 & 451760.70 & 530210 & 11.758 & 0.0 &  & $2^{-6}$ & 600659.00 & 600659 &  0.033 & 0.0 &  &  $2^{0}$ & 599597.40 & 600858 &  0.000 & 0.2 &  &  $2^{0}$ & 599597.40 & 600858 &  0.000 & 0.2 \\
  200\_75\_9 & 516771 & 516151 & 0.120 &  &  $2^{2}$ & 0.9 & 420729.67 & 482853 &  6.563 & 0.0 &  &  $2^{0}$ & 494623.20 & 515058 &  0.331 & 0.0 &  & $2^{-1}$ & 516494.00 & 516494 &  0.054 & 0.0 &  &  $2^{1}$ & 516395.70 & 516661 &  0.021 & 0.0 \\
 200\_75\_10 & 142694 & 142694 & 0.000 &  &  $2^{4}$ & 0.1 & 134749.00 & 134749 &  5.568 & 0.0 &  &  $2^{2}$ & 141502.70 & 142585 &  0.076 & 0.0 &  &  $2^{4}$ & 138983.00 & 138983 &  2.601 & 0.0 &  & $2^{-6}$ & 142694.00 & 142694 &  0.000 & 1.0 \\
 200\_100\_1 & 937149 & 937123 & 2.8E-3 &  & $2^{-1}$ & 0.7 & 717350.86 & 778851 & 16.891 & 0.0 &  & $2^{-6}$ & 935700.00 & 935700 &  0.155 & 0.0 &  &  $2^{4}$ & 937093.10 & 937149 &  0.000 & 0.3 &  &  $2^{4}$ & 937093.10 & 937149 &  0.000 & 0.3 \\
 200\_100\_2 & 303058 & 302690 & 0.121 &  &  $2^{6}$ & 0.5 & 237593.60 & 250511 & 17.339 & 0.0 &  &  $2^{2}$ & 301369.40 & 302035 &  0.338 & 0.0 &  &  $2^{6}$ & 246849.60 & 257617 & 14.994 & 0.0 &  &  $2^{4}$ & 301856.00 & 303050 & 2.6E-3 & 0.0 \\
 200\_100\_3 &  29367 &  29296 & 0.242 &  &  $2^{6}$ & 0.2 &  27624.00 &  27834 &  5.220 & 0.0 &  &  $2^{4}$ &  28956.90 &  29176 &  0.650 & 0.0 &  &  $2^{6}$ &  28285.00 &  28285 &  3.684 & 0.0 &  & $2^{-6}$ &  29296.00 &  29296 &  0.242 & 0.0 \\
 200\_100\_4 & 100838 & 100838 & 0.000 &  &  $2^{5}$ & 0.1 &  93818.00 &  93818 &  6.962 & 0.0 &  &  $2^{3}$ & 100508.50 & 100837 & 9.9E-4 & 0.0 &  &  $2^{5}$ &  95446.00 &  95446 &  5.347 & 0.0 &  & $2^{-6}$ & 100838.00 & 100838 &  0.000 & 1.0 \\
 200\_100\_5 & 786635 & 786482 & 0.019 &  &  $2^{1}$ & 0.9 & 586694.89 & 689990 & 12.286 & 0.0 &  &  $2^{1}$ & 606642.30 & 786169 &  0.059 & 0.0 &  &  $2^{4}$ & 786458.40 & 786490 &  0.018 & 0.0 &  &  $2^{4}$ & 786458.40 & 786490 &  0.018 & 0.0 \\
 200\_100\_6 &  41171 &  41171 & 0.000 &  &  $2^{6}$ & 0.1 &  39199.00 &  39199 &  4.790 & 0.0 &  &  $2^{4}$ &  39688.60 &  41171 &  0.000 & 0.2 &  &  $2^{6}$ &  39980.00 &  39980 &  2.893 & 0.0 &  & $2^{-6}$ &  41171.00 &  41171 &  0.000 & 1.0 \\
 200\_100\_7 & 701094 & 700965 & 0.018 &  &  $2^{1}$ & 0.6 & 540193.33 & 609513 & 13.063 & 0.0 &  &  $2^{1}$ & 604114.80 & 700570 &  0.075 & 0.0 &  &  $2^{1}$ & 700970.50 & 700998 &  0.014 & 0.0 &  &  $2^{1}$ & 700993.00 & 701094 &  0.000 & 0.2 \\
 200\_100\_8 & 782443 & 781455 & 0.126 &  &  $2^{1}$ & 1.0 & 597174.90 & 693221 & 11.403 & 0.0 &  & $2^{-6}$ & 779797.00 & 779797 &  0.338 & 0.0 &  &  $2^{3}$ & 781926.60 & 782408 & 4.5E-3 & 0.0 &  &  $2^{3}$ & 781926.60 & 782408 & 4.5E-3 & 0.0 \\
 200\_100\_9 & 628992 & 628893 & 0.016 &  &  $2^{3}$ & 0.8 & 546637.50 & 582449 &  7.400 & 0.0 &  &  $2^{0}$ & 626696.10 & 627799 &  0.190 & 0.0 &  &  $2^{3}$ & 627774.00 & 628992 &  0.000 & 0.1 &  &  $2^{2}$ & 627240.90 & 628992 &  0.000 & 0.1 \\
200\_100\_10 & 378442 & 378169 & 0.072 &  &  $2^{5}$ & 0.4 & 340558.25 & 347358 &  8.214 & 0.0 &  &  $2^{2}$ & 376514.00 & 377460 &  0.259 & 0.0 &  &  $2^{5}$ & 355961.50 & 363794 &  3.871 & 0.0 &  & $2^{-6}$ & 378169.00 & 378169 &  0.072 & 0.0 \\
  300\_25\_1 &  29140 &  29140 & 0.000 &  &  $2^{5}$ & 0.4 &  24376.50 &  25337 & 13.051 & 0.0 &  &  $2^{2}$ &  28985.20 &  29104 &  0.124 & 0.0 &  &  $2^{5}$ &  26907.75 &  27951 &  4.080 & 0.0 &  & $2^{-6}$ &  29140.00 &  29140 &  0.000 & 1.0 \\
  300\_25\_2 & 281990 & 281268 & 0.256 &  &  $2^{2}$ & 0.3 & 265680.00 & 279136 &  1.012 & 0.0 &  & $2^{-1}$ & 280505.20 & 281931 &  0.021 & 0.0 &  &  $2^{2}$ & 280874.33 & 281077 &  0.324 & 0.0 &  &  $2^{0}$ & 281882.10 & 281973 & 6.0E-3 & 0.0 \\
  300\_25\_3 & 231075 & 231075 & 0.000 &  &  $2^{3}$ & 0.1 & 205225.00 & 205225 & 11.187 & 0.0 &  & $2^{-1}$ & 230662.90 & 231075 &  0.000 & 0.3 &  &  $2^{3}$ & 226661.00 & 226661 &  1.910 & 0.0 &  & $2^{-6}$ & 231075.00 & 231075 &  0.000 & 1.0 \\
  300\_25\_4 & 444759 & 444712 & 0.011 &  & $2^{-2}$ & 1.0 & 333689.10 & 377515 & 15.119 & 0.0 &  & $2^{-6}$ & 443909.00 & 443909 &  0.191 & 0.0 &  & $2^{-2}$ & 444559.20 & 444725 & 7.6E-3 & 0.0 &  & $2^{-2}$ & 444559.20 & 444725 & 7.6E-3 & 0.0 \\
  300\_25\_5 &  14988 &  14879 & 0.727 &  &  $2^{5}$ & 0.2 &  12814.50 &  13524 &  9.768 & 0.0 &  &  $2^{3}$ &  14729.90 &  14958 &  0.200 & 0.0 &  &  $2^{5}$ &  13494.00 &  14245 &  4.957 & 0.0 &  &  $2^{3}$ &  14894.40 &  14988 &  0.000 & 0.1 \\
  300\_25\_6 & 269782 & 269671 & 0.041 &  &  $2^{3}$ & 0.3 & 256370.00 & 257802 &  4.441 & 0.0 &  & $2^{-1}$ & 267575.50 & 269671 &  0.041 & 0.0 &  &  $2^{3}$ & 263143.00 & 263753 &  2.235 & 0.0 &  &  $2^{1}$ & 269338.00 & 269782 &  0.000 & 0.3 \\
  300\_25\_7 & 485263 & 484539 & 0.149 &  & $2^{-4}$ & 0.6 & 373124.17 & 378224 & 22.058 & 0.0 &  & $2^{-6}$ & 483078.00 & 483078 &  0.450 & 0.0 &  &  $2^{1}$ & 484619.00 & 485263 &  0.000 & 0.1 &  &  $2^{1}$ & 484619.00 & 485263 &  0.000 & 0.1 \\
  300\_25\_8 &   9343 &   9343 & 0.000 &  &  $2^{6}$ & 0.6 &   7148.50 &   8223 & 11.988 & 0.0 &  & $2^{-6}$ &   9224.00 &   9224 &  1.274 & 0.0 &  &  $2^{6}$ &   8061.83 &   8863 &  5.138 & 0.0 &  & $2^{-6}$ &   9343.00 &   9343 &  0.000 & 1.0 \\
  300\_25\_9 & 250761 & 250551 & 0.084 &  &  $2^{3}$ & 0.3 & 238705.67 & 243136 &  3.041 & 0.0 &  & $2^{-1}$ & 250406.50 & 250751 & 4.0E-3 & 0.0 &  &  $2^{3}$ & 248275.33 & 248900 &  0.742 & 0.0 &  &  $2^{1}$ & 250398.50 & 250761 &  0.000 & 0.1 \\
 300\_25\_10 & 383377 & 383377 & 0.000 &  &  $2^{0}$ & 0.9 & 309068.11 & 347659 &  9.317 & 0.0 &  & $2^{-1}$ & 331814.50 & 383377 &  0.000 & 0.1 &  &  $2^{0}$ & 382790.67 & 383377 &  0.000 & 0.2 &  & $2^{-6}$ & 383377.00 & 383377 &  0.000 & 1.0 \\
  300\_50\_1 & 513379 & 513361 & 3.5E-3 &  &  $2^{4}$ & 0.4 & 470796.75 & 483335 &  5.852 & 0.0 &  &  $2^{1}$ & 511127.00 & 512984 &  0.077 & 0.0 &  &  $2^{2}$ & 513084.00 & 513084 &  0.057 & 0.0 &  &  $2^{1}$ & 513057.50 & 513379 &  0.000 & 0.1 \\
  300\_50\_2 & 105543 & 105543 & 0.000 &  &  $2^{5}$ & 0.2 &  90216.00 &  95103 &  9.892 & 0.0 &  &  $2^{2}$ & 103926.40 & 105543 &  0.000 & 0.3 &  &  $2^{5}$ &  95996.50 &  99371 &  5.848 & 0.0 &  & $2^{-6}$ & 105543.00 & 105543 &  0.000 & 1.0 \\
  300\_50\_3 & 875788 & 874561 & 0.140 &  & $2^{-1}$ & 1.0 & 665029.60 & 745587 & 14.867 & 0.0 &  & $2^{-6}$ & 871417.00 & 871417 &  0.499 & 0.0 &  &  $2^{0}$ & 875387.40 & 875769 & 2.2E-3 & 0.0 &  &  $2^{0}$ & 875387.40 & 875769 & 2.2E-3 & 0.0 \\
  300\_50\_4 & 307124 & 307124 & 0.000 &  &  $2^{5}$ & 0.3 & 265725.00 & 270113 & 12.051 & 0.0 &  & $2^{-6}$ & 306937.00 & 306937 &  0.061 & 0.0 &  &  $2^{5}$ & 281834.00 & 285461 &  7.054 & 0.0 &  & $2^{-6}$ & 307124.00 & 307124 &  0.000 & 1.0 \\
  300\_50\_5 & 727820 & 727486 & 0.046 &  &  $2^{1}$ & 0.9 & 616767.22 & 655825 &  9.892 & 0.0 &  & $2^{-1}$ & 726349.50 & 727463 &  0.049 & 0.0 &  &  $2^{0}$ & 727650.00 & 727684 &  0.019 & 0.0 &  &  $2^{1}$ & 727634.60 & 727820 &  0.000 & 0.1 \\
  300\_50\_6 & 734053 & 733855 & 0.027 &  &  $2^{1}$ & 0.9 & 619285.56 & 662004 &  9.815 & 0.0 &  &  $2^{0}$ & 662379.00 & 733923 &  0.018 & 0.0 &  &  $2^{2}$ & 733833.10 & 734053 &  0.000 & 0.1 &  &  $2^{0}$ & 733963.40 & 734053 &  0.000 & 0.1 \\
  300\_50\_7 &  43595 &  43524 & 0.163 &  &  $2^{6}$ & 0.3 &  31501.33 &  33915 & 22.204 & 0.0 &  &  $2^{3}$ &  42913.00 &  43203 &  0.899 & 0.0 &  &  $2^{6}$ &  35597.00 &  37920 & 13.018 & 0.0 &  &  $2^{4}$ &  43479.40 &  43595 &  0.000 & 0.2 \\
  300\_50\_8 & 767977 & 767959 & 2.3E-3 &  &  $2^{0}$ & 0.7 & 640103.00 & 703023 &  8.458 & 0.0 &  &  $2^{0}$ & 678100.80 & 767311 &  0.087 & 0.0 &  &  $2^{3}$ & 766792.70 & 767960 & 2.2E-3 & 0.0 &  &  $2^{3}$ & 766792.70 & 767960 & 2.2E-3 & 0.0 \\
  300\_50\_9 & 761351 & 761351 & 0.000 &  &  $2^{1}$ & 0.8 & 642154.75 & 689821 &  9.395 & 0.0 &  & $2^{-6}$ & 760565.00 & 760565 &  0.103 & 0.0 &  &  $2^{0}$ & 761154.12 & 761351 &  0.000 & 0.2 &  & $2^{-6}$ & 761351.00 & 761351 &  0.000 & 1.0 \\
 300\_50\_10 & 996070 & 996070 & 0.000 &  & $2^{-2}$ & 1.0 & 710165.70 & 887845 & 10.865 & 0.0 &  & $2^{-6}$ & 989559.00 & 989559 &  0.654 & 0.0 &  &  $2^{5}$ & 994305.30 & 996070 &  0.000 & 0.1 &  & $2^{-6}$ & 996070.00 & 996070 &  0.000 & 1.0 \\
    \hline
  \end{tabular}
}
\end{table}
\end{landscape}

\begin{landscape}

\begin{table}[t]
  \caption{Full Results of SA-RI with Various Encoding Methods.}
  \label{tab:full_sa_compare_enc_1}
  \centering
\resizebox{\linewidth}{!}{
  \begin{tabular}{@{\hspace{1pt}}c@{\hspace{4pt}}r@{\hspace{4pt}}|r@{\hspace{3pt}}
  r@{\hspace{5pt}}r@{\hspace{5pt}}r@{\hspace{5pt}}r@{\hspace{5pt}}r
  @{\hspace{5pt}}r@{\hspace{5pt}}
  r@{\hspace{5pt}}r@{\hspace{5pt}}r@{\hspace{5pt}}r@{\hspace{5pt}}r
  @{\hspace{5pt}}r@{\hspace{5pt}}
  r@{\hspace{5pt}}r@{\hspace{5pt}}r@{\hspace{5pt}}r@{\hspace{5pt}}r
  @{\hspace{5pt}}r@{\hspace{5pt}}
  r@{\hspace{5pt}}r@{\hspace{5pt}}r@{\hspace{5pt}}r@{\hspace{5pt}}r
  @{\hspace{5pt}}r@{\hspace{5pt}}
  r@{\hspace{5pt}}r@{\hspace{5pt}}r@{\hspace{5pt}}r@{\hspace{5pt}}r
  @{\hspace{3pt}}}
    \hline
    \multicolumn{2}{c|}{Instance} &&
    \multicolumn{5}{c}{Binary} && 
    \multicolumn{5}{c}{Hybrid} && 
    \multicolumn{5}{c}{Unary} && 
    \multicolumn{5}{c}{One-hot} && 
    \multicolumn{5}{c}{Offset} \\
    \cline{4-8}
    \cline{10-14}
    \cline{16-20}
    \cline{22-26}
    \cline{28-32}
    $n$\_$d$\_id & Optimal &
    & $\lambda$ & Mean & Best & Gap & SR &  
    & $\lambda$ & Mean & Best & Gap & SR &  
    & $\lambda$ & Mean & Best & Gap & SR &  
    & $\lambda$ & Mean & Best & Gap & SR &  
    & $\lambda$ & Mean & Best & Gap & SR \\
    \hline \hline
  100\_25\_1 & 18558 & &$2^{-3}$ & 18558.0 & 18558 & 0.000 &1.0 & &$2^{-3}$ & 18558.0 & 18558 & 0.000 &1.0 & &$2^{-3}$ & 18558.0 & 18558 & 0.000 &1.0 & &$2^{-4}$ & 18558.0 & 18558 & 0.000 &1.0 & &$2^{-4}$ & 18558.0 & 18558 & 0.000 &1.0 \\
  100\_25\_2 & 56525 & &$2^{-6}$ & 56525.0 & 56525 & 0.000 &1.0 & &$2^{-6}$ & 56525.0 & 56525 & 0.000 &1.0 & &$2^{-6}$ & 56525.0 & 56525 & 0.000 &1.0 & &$2^{-6}$ & 56525.0 & 56525 & 0.000 &1.0 & &$2^{-6}$ & 56525.0 & 56525 & 0.000 &1.0 \\
  100\_25\_3 &  3752 & & $2^{3}$ &  3664.4 &  3752 & 0.000 &0.2 & & $2^{1}$ &  3705.0 &  3717 & 0.933 &0.0 & & $2^{3}$ &  3532.2 &  3752 & 0.000 &0.1 & & $2^{3}$ &  3562.3 &  3752 & 0.000 &0.1 & & $2^{3}$ &  3582.6 &  3752 & 0.000 &0.1 \\
  100\_25\_4 & 50382 & &$2^{-6}$ & 50382.0 & 50382 & 0.000 &1.0 & &$2^{-6}$ & 50382.0 & 50382 & 0.000 &1.0 & &$2^{-6}$ & 50382.0 & 50382 & 0.000 &1.0 & &$2^{-6}$ & 50382.0 & 50382 & 0.000 &1.0 & &$2^{-6}$ & 50382.0 & 50382 & 0.000 &1.0 \\
  100\_25\_5 & 61494 & &$2^{-6}$ & 61494.0 & 61494 & 0.000 &1.0 & &$2^{-6}$ & 61494.0 & 61494 & 0.000 &1.0 & &$2^{-6}$ & 61494.0 & 61494 & 0.000 &1.0 & &$2^{-6}$ & 61494.0 & 61494 & 0.000 &1.0 & &$2^{-6}$ & 61494.0 & 61494 & 0.000 &1.0 \\
  100\_25\_6 & 36360 & &$2^{-1}$ & 36206.1 & 36360 & 0.000 &0.1 & &$2^{-1}$ & 36223.2 & 36360 & 0.000 &0.2 & & $2^{1}$ & 36169.3 & 36303 & 0.157 &0.0 & &$2^{-1}$ & 36225.2 & 36360 & 0.000 &0.2 & & $2^{0}$ & 36200.2 & 36360 & 0.000 &0.1 \\
  100\_25\_7 & 14657 & &$2^{-3}$ & 14657.0 & 14657 & 0.000 &1.0 & &$2^{-3}$ & 14657.0 & 14657 & 0.000 &1.0 & &$2^{-3}$ & 14657.0 & 14657 & 0.000 &1.0 & &$2^{-3}$ & 14657.0 & 14657 & 0.000 &1.0 & &$2^{-3}$ & 14657.0 & 14657 & 0.000 &1.0 \\
  100\_25\_8 & 20452 & & $2^{0}$ & 20290.9 & 20452 & 0.000 &0.2 & & $2^{0}$ & 20286.2 & 20452 & 0.000 &0.1 & & $2^{1}$ & 20319.4 & 20395 & 0.279 &0.0 & & $2^{1}$ & 20330.2 & 20452 & 0.000 &0.3 & & $2^{1}$ & 20304.9 & 20452 & 0.000 &0.2 \\
  100\_25\_9 & 35438 & &$2^{-2}$ & 35401.6 & 35438 & 0.000 &0.6 & &$2^{-2}$ & 35401.6 & 35438 & 0.000 &0.6 & &$2^{-1}$ & 35370.3 & 35438 & 0.000 &0.4 & &$2^{-3}$ & 35385.6 & 35438 & 0.000 &0.5 & &$2^{-4}$ & 35438.0 & 35438 & 0.000 &1.0 \\
 100\_25\_10 & 24930 & &$2^{-1}$ & 24909.8 & 24930 & 0.000 &0.1 & & $2^{0}$ & 24904.7 & 24930 & 0.000 &0.2 & & $2^{0}$ & 24901.4 & 24930 & 0.000 &0.2 & & $2^{1}$ & 24864.0 & 24930 & 0.000 &0.3 & &$2^{-2}$ & 24926.4 & 24930 & 0.000 &0.1 \\
  100\_50\_1 & 83742 & & $2^{1}$ & 83695.7 & 83742 & 0.000 &0.5 & & $2^{2}$ & 83361.3 & 83742 & 0.000 &0.4 & & $2^{1}$ & 83711.9 & 83742 & 0.000 &0.5 & & $2^{2}$ & 83190.4 & 83742 & 0.000 &0.4 & & $2^{1}$ & 83736.0 & 83742 & 0.000 &0.8 \\
  100\_50\_2 &104856 & & $2^{1}$ &104775.1 &104856 & 0.000 &0.1 & & $2^{1}$ &104776.5 &104792 & 0.061 &0.0 & & $2^{1}$ &104774.7 &104792 & 0.061 &0.0 & & $2^{1}$ &104749.7 &104856 & 0.000 &0.1 & &$2^{-1}$ &104751.2 &104792 & 0.061 &0.0 \\
  100\_50\_3 & 34006 & & $2^{1}$ & 33930.0 & 34006 & 0.000 &0.5 & & $2^{0}$ & 33933.2 & 34006 & 0.000 &0.3 & & $2^{0}$ & 33933.2 & 34006 & 0.000 &0.3 & & $2^{2}$ & 33969.6 & 34006 & 0.000 &0.5 & & $2^{2}$ & 33875.8 & 34006 & 0.000 &0.3 \\
  100\_50\_4 &105996 & &$2^{-6}$ &105996.0 &105996 & 0.000 &1.0 & &$2^{-6}$ &105996.0 &105996 & 0.000 &1.0 & &$2^{-6}$ &105996.0 &105996 & 0.000 &1.0 & &$2^{-6}$ &105996.0 &105996 & 0.000 &1.0 & &$2^{-6}$ &105996.0 &105996 & 0.000 &1.0 \\
  100\_50\_5 & 56464 & &$2^{-6}$ & 56464.0 & 56464 & 0.000 &1.0 & &$2^{-6}$ & 56464.0 & 56464 & 0.000 &1.0 & &$2^{-6}$ & 56464.0 & 56464 & 0.000 &1.0 & &$2^{-6}$ & 56464.0 & 56464 & 0.000 &1.0 & &$2^{-6}$ & 56464.0 & 56464 & 0.000 &1.0 \\
  100\_50\_6 & 16083 & &$2^{-6}$ & 16083.0 & 16083 & 0.000 &1.0 & &$2^{-6}$ & 16083.0 & 16083 & 0.000 &1.0 & &$2^{-6}$ & 16083.0 & 16083 & 0.000 &1.0 & &$2^{-6}$ & 16083.0 & 16083 & 0.000 &1.0 & &$2^{-6}$ & 16083.0 & 16083 & 0.000 &1.0 \\
  100\_50\_7 & 52819 & & $2^{2}$ & 52347.7 & 52819 & 0.000 &0.2 & & $2^{1}$ & 52735.7 & 52819 & 0.000 &0.4 & & $2^{0}$ & 52787.5 & 52819 & 0.000 &0.1 & & $2^{1}$ & 52760.1 & 52819 & 0.000 &0.2 & & $2^{0}$ & 52697.5 & 52819 & 0.000 &0.3 \\
  100\_50\_8 & 54246 & & $2^{0}$ & 54089.4 & 54246 & 0.000 &0.2 & & $2^{2}$ & 53755.0 & 54246 & 0.000 &0.1 & & $2^{2}$ & 53480.8 & 54246 & 0.000 &0.1 & & $2^{2}$ & 53889.0 & 54246 & 0.000 &0.3 & & $2^{1}$ & 54144.4 & 54246 & 0.000 &0.2 \\
  100\_50\_9 & 68974 & &$2^{-6}$ & 68974.0 & 68974 & 0.000 &1.0 & &$2^{-6}$ & 68974.0 & 68974 & 0.000 &1.0 & &$2^{-6}$ & 68974.0 & 68974 & 0.000 &1.0 & &$2^{-6}$ & 68974.0 & 68974 & 0.000 &1.0 & &$2^{-6}$ & 68974.0 & 68974 & 0.000 &1.0 \\
 100\_50\_10 & 88634 & & $2^{0}$ & 88470.4 & 88600 & 0.038 &0.0 & & $2^{0}$ & 88507.4 & 88634 & 0.000 &0.2 & & $2^{0}$ & 88528.5 & 88634 & 0.000 &0.1 & & $2^{0}$ & 88492.2 & 88575 & 0.067 &0.0 & & $2^{2}$ & 88346.8 & 88634 & 0.000 &0.2 \\
  100\_75\_1 &189137 & &$2^{-6}$ &189137.0 &189137 & 0.000 &1.0 & &$2^{-6}$ &189137.0 &189137 & 0.000 &1.0 & &$2^{-6}$ &189137.0 &189137 & 0.000 &1.0 & &$2^{-6}$ &189137.0 &189137 & 0.000 &1.0 & &$2^{-6}$ &189137.0 &189137 & 0.000 &1.0 \\
  100\_75\_2 & 95074 & & $2^{2}$ & 94970.1 & 95074 & 0.000 &0.1 & & $2^{0}$ & 94968.1 & 95003 & 0.075 &0.0 & &$2^{-1}$ & 94978.9 & 95003 & 0.075 &0.0 & &$2^{-1}$ & 94982.3 & 95003 & 0.075 &0.0 & &$2^{-1}$ & 94988.6 & 95003 & 0.075 &0.0 \\
  100\_75\_3 & 62098 & &$2^{-6}$ & 62098.0 & 62098 & 0.000 &1.0 & &$2^{-6}$ & 62098.0 & 62098 & 0.000 &1.0 & &$2^{-6}$ & 62098.0 & 62098 & 0.000 &1.0 & &$2^{-6}$ & 62098.0 & 62098 & 0.000 &1.0 & &$2^{-6}$ & 62098.0 & 62098 & 0.000 &1.0 \\
  100\_75\_4 & 72245 & & $2^{1}$ & 72113.1 & 72245 & 0.000 &0.2 & & $2^{2}$ & 72097.2 & 72245 & 0.000 &0.1 & & $2^{2}$ & 72134.7 & 72245 & 0.000 &0.4 & & $2^{1}$ & 72053.0 & 72245 & 0.000 &0.2 & & $2^{0}$ & 72136.8 & 72245 & 0.000 &0.3 \\
  100\_75\_5 & 27616 & &$2^{-6}$ & 27616.0 & 27616 & 0.000 &1.0 & &$2^{-6}$ & 27616.0 & 27616 & 0.000 &1.0 & &$2^{-6}$ & 27616.0 & 27616 & 0.000 &1.0 & &$2^{-6}$ & 27616.0 & 27616 & 0.000 &1.0 & &$2^{-6}$ & 27616.0 & 27616 & 0.000 &1.0 \\
  100\_75\_6 &145273 & & $2^{1}$ &144452.1 &145273 & 0.000 &0.1 & & $2^{1}$ &144995.9 &145273 & 0.000 &0.1 & & $2^{1}$ &145049.0 &145273 & 0.000 &0.3 & & $2^{1}$ &145104.7 &145273 & 0.000 &0.1 & & $2^{1}$ &144143.2 &145273 & 0.000 &0.1 \\
  100\_75\_7 &110979 & & $2^{1}$ &110824.5 &110979 & 0.000 &0.3 & & $2^{2}$ &110412.9 &110979 & 0.000 &0.2 & & $2^{1}$ &110802.9 &110979 & 0.000 &0.2 & & $2^{1}$ &110910.8 &110979 & 0.000 &0.5 & & $2^{1}$ &110873.5 &110979 & 0.000 &0.4 \\
  100\_75\_8 & 19570 & &$2^{-6}$ & 19570.0 & 19570 & 0.000 &1.0 & &$2^{-6}$ & 19570.0 & 19570 & 0.000 &1.0 & &$2^{-6}$ & 19570.0 & 19570 & 0.000 &1.0 & &$2^{-6}$ & 19570.0 & 19570 & 0.000 &1.0 & &$2^{-6}$ & 19570.0 & 19570 & 0.000 &1.0 \\
  100\_75\_9 &104341 & & $2^{2}$ &104028.4 &104285 & 0.054 &0.0 & & $2^{2}$ &104078.8 &104253 & 0.084 &0.0 & & $2^{2}$ &103668.9 &104253 & 0.084 &0.0 & & $2^{2}$ &104138.8 &104300 & 0.039 &0.0 & & $2^{2}$ &103659.2 &104300 & 0.039 &0.0 \\
 100\_75\_10 &143740 & & $2^{2}$ &143457.8 &143740 & 0.000 &0.6 & & $2^{2}$ &143318.1 &143740 & 0.000 &0.6 & & $2^{2}$ &143575.6 &143740 & 0.000 &0.7 & & $2^{2}$ &143472.4 &143740 & 0.000 &0.7 & & $2^{2}$ &143149.9 &143740 & 0.000 &0.4 \\
 100\_100\_1 & 81978 & &$2^{-6}$ & 81961.0 & 81961 & 0.021 &0.0 & &$2^{-6}$ & 81961.0 & 81961 & 0.021 &0.0 & & $2^{3}$ & 81962.7 & 81978 & 0.000 &0.1 & & $2^{3}$ & 81801.4 & 81978 & 0.000 &0.1 & &$2^{-6}$ & 81961.0 & 81961 & 0.021 &0.0 \\
 100\_100\_2 &190424 & & $2^{0}$ &190301.4 &190385 & 0.020 &0.0 & & $2^{2}$ &190221.7 &190385 & 0.020 &0.0 & & $2^{0}$ &190262.7 &190385 & 0.020 &0.0 & & $2^{2}$ &189529.8 &190385 & 0.020 &0.0 & & $2^{0}$ &190241.0 &190385 & 0.020 &0.0 \\
 100\_100\_3 &225434 & &$2^{-6}$ &225434.0 &225434 & 0.000 &1.0 & &$2^{-6}$ &225434.0 &225434 & 0.000 &1.0 & &$2^{-6}$ &225434.0 &225434 & 0.000 &1.0 & &$2^{-6}$ &225434.0 &225434 & 0.000 &1.0 & &$2^{-6}$ &225434.0 &225434 & 0.000 &1.0 \\
 100\_100\_4 & 63028 & &$2^{-6}$ & 63028.0 & 63028 & 0.000 &1.0 & &$2^{-6}$ & 63028.0 & 63028 & 0.000 &1.0 & &$2^{-6}$ & 63028.0 & 63028 & 0.000 &1.0 & &$2^{-6}$ & 63028.0 & 63028 & 0.000 &1.0 & &$2^{-6}$ & 63028.0 & 63028 & 0.000 &1.0 \\
 100\_100\_5 &230076 & &$2^{-6}$ &230076.0 &230076 & 0.000 &1.0 & &$2^{-6}$ &230076.0 &230076 & 0.000 &1.0 & &$2^{-6}$ &230076.0 &230076 & 0.000 &1.0 & &$2^{-6}$ &230076.0 &230076 & 0.000 &1.0 & &$2^{-6}$ &230076.0 &230076 & 0.000 &1.0 \\
 100\_100\_6 & 74358 & &$2^{-6}$ & 74358.0 & 74358 & 0.000 &1.0 & &$2^{-6}$ & 74358.0 & 74358 & 0.000 &1.0 & &$2^{-6}$ & 74358.0 & 74358 & 0.000 &1.0 & &$2^{-6}$ & 74358.0 & 74358 & 0.000 &1.0 & &$2^{-6}$ & 74358.0 & 74358 & 0.000 &1.0 \\
 100\_100\_7 & 10330 & & $2^{5}$ & 10037.3 & 10330 & 0.000 &0.2 & & $2^{5}$ &  9658.8 & 10330 & 0.000 &0.2 & &$2^{-6}$ & 10184.0 & 10184 & 1.413 &0.0 & & $2^{5}$ & 10083.3 & 10330 & 0.000 &0.3 & &$2^{-6}$ & 10184.0 & 10184 & 1.413 &0.0 \\
 100\_100\_8 & 62582 & & $2^{3}$ & 62481.0 & 62582 & 0.000 &0.4 & & $2^{3}$ & 62499.5 & 62582 & 0.000 &0.5 & & $2^{3}$ & 62084.8 & 62582 & 0.000 &0.4 & & $2^{3}$ & 62131.2 & 62582 & 0.000 &0.3 & & $2^{3}$ & 61924.8 & 62582 & 0.000 &0.3 \\
 100\_100\_9 &232754 & & $2^{2}$ &232735.7 &232754 & 0.000 &0.7 & & $2^{1}$ &232741.8 &232754 & 0.000 &0.8 & & $2^{1}$ &232735.7 &232754 & 0.000 &0.7 & & $2^{2}$ &232741.8 &232754 & 0.000 &0.8 & & $2^{1}$ &232746.9 &232754 & 0.000 &0.8 \\
100\_100\_10 &193262 & & $2^{3}$ &193217.4 &193262 & 0.000 &0.6 & & $2^{3}$ &193093.8 &193262 & 0.000 &0.4 & & $2^{3}$ &192524.2 &193262 & 0.000 &0.5 & & $2^{3}$ &193178.5 &193262 & 0.000 &0.5 & & $2^{1}$ &193189.0 &193262 & 0.000 &0.5 \\
  200\_25\_1 &204441 & &$2^{-1}$ &204373.3 &204441 & 0.000 &0.1 & & $2^{1}$ &204252.8 &204441 & 0.000 &0.2 & & $2^{0}$ &204389.2 &204441 & 0.000 &0.2 & &$2^{-6}$ &204399.0 &204399 & 0.021 &0.0 & & $2^{1}$ &204267.8 &204401 & 0.020 &0.0 \\
  200\_25\_2 &239573 & &$2^{-6}$ &239573.0 &239573 & 0.000 &1.0 & &$2^{-6}$ &239573.0 &239573 & 0.000 &1.0 & &$2^{-6}$ &239573.0 &239573 & 0.000 &1.0 & &$2^{-6}$ &239573.0 &239573 & 0.000 &1.0 & &$2^{-6}$ &239573.0 &239573 & 0.000 &1.0 \\
  200\_25\_3 &245463 & &$2^{-1}$ &244649.4 &245463 & 0.000 &0.2 & & $2^{1}$ &244815.7 &245380 & 0.034 &0.0 & & $2^{1}$ &244885.1 &245463 & 0.000 &0.4 & & $2^{2}$ &244737.3 &245463 & 0.000 &0.1 & & $2^{2}$ &244590.1 &244631 & 0.339 &0.0 \\
  200\_25\_4 &222361 & &$2^{-1}$ &221822.0 &222361 & 0.000 &0.3 & & $2^{1}$ &221928.1 &222361 & 0.000 &0.2 & & $2^{0}$ &222110.5 &222361 & 0.000 &0.6 & &$2^{-1}$ &221928.3 &222361 & 0.000 &0.3 & &$2^{-1}$ &221991.0 &222361 & 0.000 &0.3 \\
  200\_25\_5 &187324 & & $2^{0}$ &187126.5 &187324 & 0.000 &0.1 & &$2^{-1}$ &187272.9 &187324 & 0.000 &0.1 & & $2^{1}$ &186945.8 &187324 & 0.000 &0.1 & &$2^{-6}$ &187315.0 &187315 &4.8E-3 &0.0 & & $2^{0}$ &187102.9 &187324 & 0.000 &0.1 \\
  200\_25\_6 & 80351 & &$2^{-2}$ & 80293.0 & 80310 & 0.051 &0.0 & & $2^{0}$ & 80244.5 & 80351 & 0.000 &0.1 & & $2^{1}$ & 80223.4 & 80351 & 0.000 &0.1 & &$2^{-1}$ & 80276.7 & 80351 & 0.000 &0.1 & &$2^{-1}$ & 80311.6 & 80351 & 0.000 &0.3 \\
  200\_25\_7 & 59036 & & $2^{1}$ & 58863.9 & 59036 & 0.000 &0.4 & & $2^{1}$ & 58907.7 & 59036 & 0.000 &0.3 & & $2^{1}$ & 58943.4 & 59036 & 0.000 &0.7 & & $2^{1}$ & 58921.4 & 59036 & 0.000 &0.2 & & $2^{0}$ & 58973.4 & 59036 & 0.000 &0.6 \\
  200\_25\_8 &149433 & & $2^{0}$ &149256.1 &149433 & 0.000 &0.1 & & $2^{0}$ &149231.6 &149385 & 0.032 &0.0 & &$2^{-1}$ &149187.6 &149433 & 0.000 &0.2 & & $2^{0}$ &149268.8 &149433 & 0.000 &0.1 & &$2^{-1}$ &149216.6 &149433 & 0.000 &0.2 \\
  200\_25\_9 & 49366 & &$2^{-6}$ & 49366.0 & 49366 & 0.000 &1.0 & &$2^{-6}$ & 49366.0 & 49366 & 0.000 &1.0 & &$2^{-6}$ & 49366.0 & 49366 & 0.000 &1.0 & &$2^{-6}$ & 49366.0 & 49366 & 0.000 &1.0 & &$2^{-6}$ & 49366.0 & 49366 & 0.000 &1.0 \\
 200\_25\_10 & 48459 & & $2^{0}$ & 48341.4 & 48459 & 0.000 &0.3 & & $2^{2}$ & 48241.5 & 48459 & 0.000 &0.5 & & $2^{2}$ & 48179.4 & 48459 & 0.000 &0.2 & & $2^{2}$ & 48223.9 & 48459 & 0.000 &0.3 & & $2^{2}$ & 48213.7 & 48459 & 0.000 &0.3 \\
    \hline
  \end{tabular}
}
\end{table}
\end{landscape}
\addtocounter{table}{-1}
\begin{landscape}
\begin{table}[t]
  \caption{Full Results of SA-RI with Various Encoding Methods (continued).}
  \label{tab:full_sa_compare_enc_2}
  \centering
\resizebox{\linewidth}{!}{
  \begin{tabular}{@{\hspace{1pt}}c@{\hspace{4pt}}r@{\hspace{4pt}}|r@{\hspace{3pt}}
  r@{\hspace{5pt}}r@{\hspace{5pt}}r@{\hspace{5pt}}r@{\hspace{5pt}}r
  @{\hspace{5pt}}r@{\hspace{5pt}}
  r@{\hspace{5pt}}r@{\hspace{5pt}}r@{\hspace{5pt}}r@{\hspace{5pt}}r
  @{\hspace{5pt}}r@{\hspace{5pt}}
  r@{\hspace{5pt}}r@{\hspace{5pt}}r@{\hspace{5pt}}r@{\hspace{5pt}}r
  @{\hspace{5pt}}r@{\hspace{5pt}}
  r@{\hspace{5pt}}r@{\hspace{5pt}}r@{\hspace{5pt}}r@{\hspace{5pt}}r
  @{\hspace{5pt}}r@{\hspace{5pt}}
  r@{\hspace{5pt}}r@{\hspace{5pt}}r@{\hspace{5pt}}r@{\hspace{5pt}}r
  @{\hspace{3pt}}}
    \hline
    \multicolumn{2}{c|}{Instance} &&
    \multicolumn{5}{c}{Binary} && 
    \multicolumn{5}{c}{Hybrid} && 
    \multicolumn{5}{c}{Unary} && 
    \multicolumn{5}{c}{One-hot} && 
    \multicolumn{5}{c}{Offset} \\
    \cline{4-8}
    \cline{10-14}
    \cline{16-20}
    \cline{22-26}
    \cline{28-32}
    $n$\_$d$\_id & Optimal &
    & $\lambda$ & Mean & Best & Gap & SR &  
    & $\lambda$ & Mean & Best & Gap & SR &  
    & $\lambda$ & Mean & Best & Gap & SR &  
    & $\lambda$ & Mean & Best & Gap & SR &  
    & $\lambda$ & Mean & Best & Gap & SR \\
    \hline \hline
  200\_50\_1 &372097 & &$2^{-6}$ &372097.0 &372097 & 0.000 &1.0 & &$2^{-6}$ &372097.0 &372097 & 0.000 &1.0 & &$2^{-6}$ &372097.0 &372097 & 0.000 &1.0 & &$2^{-6}$ &372097.0 &372097 & 0.000 &1.0 & &$2^{-6}$ &372097.0 &372097 & 0.000 &1.0 \\
  200\_50\_2 &211130 & & $2^{1}$ &210742.6 &211110 &9.5E-3 &0.0 & & $2^{1}$ &210684.5 &211090 & 0.019 &0.0 & & $2^{1}$ &210780.1 &211090 & 0.019 &0.0 & & $2^{2}$ &210615.9 &211087 & 0.020 &0.0 & & $2^{0}$ &210762.7 &211090 & 0.019 &0.0 \\
  200\_50\_3 &227185 & &$2^{-6}$ &227185.0 &227185 & 0.000 &1.0 & &$2^{-6}$ &227185.0 &227185 & 0.000 &1.0 & &$2^{-6}$ &227185.0 &227185 & 0.000 &1.0 & &$2^{-6}$ &227185.0 &227185 & 0.000 &1.0 & &$2^{-6}$ &227185.0 &227185 & 0.000 &1.0 \\
  200\_50\_4 &228572 & &$2^{-6}$ &228572.0 &228572 & 0.000 &1.0 & &$2^{-6}$ &228572.0 &228572 & 0.000 &1.0 & &$2^{-6}$ &228572.0 &228572 & 0.000 &1.0 & &$2^{-6}$ &228572.0 &228572 & 0.000 &1.0 & &$2^{-6}$ &228572.0 &228572 & 0.000 &1.0 \\
  200\_50\_5 &479651 & & $2^{2}$ &479110.6 &479651 & 0.000 &0.1 & &$2^{-6}$ &479451.0 &479451 & 0.042 &0.0 & &$2^{-6}$ &479451.0 &479451 & 0.042 &0.0 & & $2^{2}$ &479388.1 &479651 & 0.000 &0.1 & &$2^{-6}$ &479451.0 &479451 & 0.042 &0.0 \\
  200\_50\_6 &426777 & & $2^{2}$ &426554.7 &426686 & 0.021 &0.0 & & $2^{2}$ &426607.3 &426657 & 0.028 &0.0 & & $2^{3}$ &425383.2 &426720 & 0.013 &0.0 & & $2^{2}$ &426329.3 &426762 &3.5E-3 &0.0 & &$2^{-1}$ &426578.3 &426657 & 0.028 &0.0 \\
  200\_50\_7 &220890 & & $2^{0}$ &220757.5 &220890 & 0.000 &0.2 & &$2^{-1}$ &220798.4 &220890 & 0.000 &0.1 & & $2^{0}$ &220798.0 &220890 & 0.000 &0.2 & &$2^{-1}$ &220818.9 &220890 & 0.000 &0.2 & & $2^{0}$ &220730.0 &220890 & 0.000 &0.2 \\
  200\_50\_8 &317952 & &$2^{-3}$ &317952.0 &317952 & 0.000 &1.0 & &$2^{-3}$ &317952.0 &317952 & 0.000 &1.0 & &$2^{-3}$ &317952.0 &317952 & 0.000 &1.0 & &$2^{-3}$ &317952.0 &317952 & 0.000 &1.0 & &$2^{-3}$ &317952.0 &317952 & 0.000 &1.0 \\
  200\_50\_9 &104936 & &$2^{-6}$ &104936.0 &104936 & 0.000 &1.0 & &$2^{-6}$ &104936.0 &104936 & 0.000 &1.0 & &$2^{-6}$ &104936.0 &104936 & 0.000 &1.0 & &$2^{-6}$ &104936.0 &104936 & 0.000 &1.0 & &$2^{-6}$ &104936.0 &104936 & 0.000 &1.0 \\
 200\_50\_10 &284751 & &$2^{-6}$ &284741.0 &284741 &3.5E-3 &0.0 & & $2^{1}$ &284455.5 &284751 & 0.000 &0.1 & & $2^{0}$ &284726.6 &284745 &2.1E-3 &0.0 & & $2^{0}$ &284568.4 &284745 &2.1E-3 &0.0 & & $2^{0}$ &284469.4 &284751 & 0.000 &0.1 \\
  200\_75\_1 &442894 & & $2^{2}$ &442036.9 &442582 & 0.070 &0.0 & & $2^{1}$ &442446.4 &442894 & 0.000 &0.1 & & $2^{0}$ &442542.4 &442894 & 0.000 &0.2 & & $2^{2}$ &442049.5 &442862 &7.2E-3 &0.0 & & $2^{0}$ &442320.8 &442602 & 0.066 &0.0 \\
  200\_75\_2 &286643 & & $2^{2}$ &286070.4 &286643 & 0.000 &0.2 & & $2^{1}$ &286522.5 &286643 & 0.000 &0.1 & & $2^{1}$ &286600.9 &286643 & 0.000 &0.1 & & $2^{1}$ &286497.3 &286643 & 0.000 &0.1 & & $2^{0}$ &286596.7 &286643 & 0.000 &0.2 \\
  200\_75\_3 & 61924 & &$2^{-6}$ & 61924.0 & 61924 & 0.000 &1.0 & &$2^{-6}$ & 61924.0 & 61924 & 0.000 &1.0 & &$2^{-6}$ & 61924.0 & 61924 & 0.000 &1.0 & &$2^{-6}$ & 61924.0 & 61924 & 0.000 &1.0 & &$2^{-6}$ & 61924.0 & 61924 & 0.000 &1.0 \\
  200\_75\_4 &128351 & &$2^{-6}$ &128351.0 &128351 & 0.000 &1.0 & &$2^{-6}$ &128351.0 &128351 & 0.000 &1.0 & &$2^{-6}$ &128351.0 &128351 & 0.000 &1.0 & &$2^{-6}$ &128351.0 &128351 & 0.000 &1.0 & &$2^{-6}$ &128351.0 &128351 & 0.000 &1.0 \\
  200\_75\_5 &137885 & & $2^{3}$ &137836.2 &137885 & 0.000 &0.7 & & $2^{3}$ &137818.3 &137885 & 0.000 &0.5 & & $2^{2}$ &137797.3 &137885 & 0.000 &0.3 & & $2^{2}$ &137799.7 &137885 & 0.000 &0.4 & & $2^{1}$ &137804.5 &137885 & 0.000 &0.4 \\
  200\_75\_6 &229631 & & $2^{2}$ &229109.2 &229631 & 0.000 &0.4 & & $2^{2}$ &229158.3 &229631 & 0.000 &0.3 & & $2^{3}$ &229025.1 &229631 & 0.000 &0.3 & & $2^{0}$ &229219.2 &229631 & 0.000 &0.2 & & $2^{0}$ &229471.5 &229631 & 0.000 &0.8 \\
  200\_75\_7 &269887 & &$2^{-6}$ &269887.0 &269887 & 0.000 &1.0 & &$2^{-6}$ &269887.0 &269887 & 0.000 &1.0 & &$2^{-6}$ &269887.0 &269887 & 0.000 &1.0 & &$2^{-6}$ &269887.0 &269887 & 0.000 &1.0 & &$2^{-6}$ &269887.0 &269887 & 0.000 &1.0 \\
  200\_75\_8 &600858 & & $2^{1}$ &600777.1 &600858 & 0.000 &0.3 & & $2^{0}$ &600822.9 &600858 & 0.000 &0.3 & & $2^{1}$ &600832.0 &600858 & 0.000 &0.5 & & $2^{1}$ &600537.0 &600858 & 0.000 &0.4 & & $2^{1}$ &600805.1 &600858 & 0.000 &0.3 \\
  200\_75\_9 &516771 & & $2^{1}$ &516286.9 &516661 & 0.021 &0.0 & & $2^{2}$ &515933.9 &516661 & 0.021 &0.0 & & $2^{2}$ &516320.0 &516661 & 0.021 &0.0 & & $2^{1}$ &516247.3 &516655 & 0.022 &0.0 & & $2^{1}$ &516346.8 &516655 & 0.022 &0.0 \\
 200\_75\_10 &142694 & &$2^{-6}$ &142694.0 &142694 & 0.000 &1.0 & &$2^{-6}$ &142694.0 &142694 & 0.000 &1.0 & &$2^{-6}$ &142694.0 &142694 & 0.000 &1.0 & &$2^{-6}$ &142694.0 &142694 & 0.000 &1.0 & &$2^{-6}$ &142694.0 &142694 & 0.000 &1.0 \\
 200\_100\_1 &937149 & & $2^{2}$ &937076.0 &937149 & 0.000 &0.2 & & $2^{4}$ &935463.6 &937149 & 0.000 &0.3 & & $2^{4}$ &937066.9 &937149 & 0.000 &0.3 & & $2^{4}$ &937109.8 &937149 & 0.000 &0.5 & & $2^{4}$ &937069.0 &937149 & 0.000 &0.3 \\
 200\_100\_2 &303058 & & $2^{2}$ &302633.0 &302992 & 0.022 &0.0 & & $2^{2}$ &302537.4 &302992 & 0.022 &0.0 & & $2^{3}$ &302609.7 &303004 & 0.018 &0.0 & & $2^{4}$ &301944.2 &302992 & 0.022 &0.0 & & $2^{2}$ &302699.8 &303050 &2.6E-3 &0.0 \\
 200\_100\_3 & 29367 & & $2^{5}$ & 29286.0 & 29367 & 0.000 &0.1 & & $2^{5}$ & 29202.0 & 29367 & 0.000 &0.1 & & $2^{3}$ & 29303.1 & 29367 & 0.000 &0.1 & & $2^{6}$ & 27884.0 & 29367 & 0.000 &0.2 & & $2^{5}$ & 29303.1 & 29367 & 0.000 &0.1 \\
 200\_100\_4 &100838 & &$2^{-6}$ &100838.0 &100838 & 0.000 &1.0 & &$2^{-6}$ &100838.0 &100838 & 0.000 &1.0 & &$2^{-6}$ &100838.0 &100838 & 0.000 &1.0 & &$2^{-6}$ &100838.0 &100838 & 0.000 &1.0 & &$2^{-6}$ &100838.0 &100838 & 0.000 &1.0 \\
 200\_100\_5 &786635 & & $2^{1}$ &786455.3 &786490 & 0.018 &0.0 & & $2^{3}$ &784084.2 &786490 & 0.018 &0.0 & & $2^{1}$ &786483.6 &786490 & 0.018 &0.0 & & $2^{1}$ &786459.9 &786490 & 0.018 &0.0 & & $2^{1}$ &785903.0 &786627 &1.0E-3 &0.0 \\
 200\_100\_6 & 41171 & &$2^{-6}$ & 41171.0 & 41171 & 0.000 &1.0 & &$2^{-6}$ & 41171.0 & 41171 & 0.000 &1.0 & &$2^{-6}$ & 41171.0 & 41171 & 0.000 &1.0 & &$2^{-6}$ & 41171.0 & 41171 & 0.000 &1.0 & &$2^{-6}$ & 41171.0 & 41171 & 0.000 &1.0 \\
 200\_100\_7 &701094 & & $2^{2}$ &699205.4 &701094 & 0.000 &0.1 & & $2^{1}$ &700956.9 &701094 & 0.000 &0.1 & & $2^{2}$ &700449.5 &701094 & 0.000 &0.2 & & $2^{1}$ &701029.7 &701094 & 0.000 &0.5 & & $2^{1}$ &701005.9 &701094 & 0.000 &0.3 \\
 200\_100\_8 &782443 & & $2^{2}$ &781916.1 &782397 &5.9E-3 &0.0 & & $2^{2}$ &781864.4 &782397 &5.9E-3 &0.0 & & $2^{2}$ &781869.3 &782408 &4.5E-3 &0.0 & & $2^{2}$ &781829.2 &782408 &4.5E-3 &0.0 & & $2^{2}$ &781773.2 &782408 &4.5E-3 &0.0 \\
 200\_100\_9 &628992 & & $2^{2}$ &626603.8 &628992 & 0.000 &0.1 & & $2^{2}$ &628045.4 &628992 & 0.000 &0.1 & & $2^{1}$ &628873.4 &628992 & 0.000 &0.1 & & $2^{1}$ &628854.7 &628992 & 0.000 &0.2 & & $2^{1}$ &628427.9 &628992 & 0.000 &0.1 \\
200\_100\_10 &378442 & & $2^{2}$ &378056.4 &378240 & 0.053 &0.0 & & $2^{3}$ &377947.6 &378208 & 0.062 &0.0 & &$2^{-6}$ &378169.0 &378169 & 0.072 &0.0 & & $2^{2}$ &377887.9 &378375 & 0.018 &0.0 & &$2^{-6}$ &378169.0 &378169 & 0.072 &0.0 \\
  300\_25\_1 & 29140 & &$2^{-6}$ & 29140.0 & 29140 & 0.000 &1.0 & &$2^{-6}$ & 29140.0 & 29140 & 0.000 &1.0 & &$2^{-6}$ & 29140.0 & 29140 & 0.000 &1.0 & &$2^{-6}$ & 29140.0 & 29140 & 0.000 &1.0 & &$2^{-6}$ & 29140.0 & 29140 & 0.000 &1.0 \\
  300\_25\_2 &281990 & & $2^{0}$ &281884.5 &281990 & 0.000 &0.1 & & $2^{0}$ &281903.0 &281959 & 0.011 &0.0 & & $2^{0}$ &281839.1 &281990 & 0.000 &0.1 & & $2^{0}$ &281893.4 &281970 &7.1E-3 &0.0 & &$2^{-1}$ &281933.8 &281959 & 0.011 &0.0 \\
  300\_25\_3 &231075 & &$2^{-6}$ &231075.0 &231075 & 0.000 &1.0 & &$2^{-6}$ &231075.0 &231075 & 0.000 &1.0 & &$2^{-6}$ &231075.0 &231075 & 0.000 &1.0 & &$2^{-6}$ &231075.0 &231075 & 0.000 &1.0 & &$2^{-6}$ &231075.0 &231075 & 0.000 &1.0 \\
  300\_25\_4 &444759 & &$2^{-6}$ &444712.0 &444712 & 0.011 &0.0 & & $2^{0}$ &444670.7 &444759 & 0.000 &0.1 & &$2^{-1}$ &444615.7 &444725 &7.6E-3 &0.0 & & $2^{2}$ &444320.2 &444759 & 0.000 &0.1 & & $2^{0}$ &444586.0 &444759 & 0.000 &0.1 \\
  300\_25\_5 & 14988 & & $2^{3}$ & 14891.3 & 14988 & 0.000 &0.1 & & $2^{4}$ & 14883.3 & 14935 & 0.354 &0.0 & & $2^{4}$ & 14893.7 & 14988 & 0.000 &0.1 & & $2^{4}$ & 14885.8 & 14988 & 0.000 &0.2 & & $2^{0}$ & 14892.9 & 14988 & 0.000 &0.1 \\
  300\_25\_6 &269782 & & $2^{0}$ &269477.8 &269715 & 0.025 &0.0 & & $2^{0}$ &269676.7 &269782 & 0.000 &0.1 & & $2^{0}$ &269693.2 &269782 & 0.000 &0.2 & & $2^{0}$ &269663.4 &269782 & 0.000 &0.4 & & $2^{0}$ &269697.6 &269782 & 0.000 &0.2 \\
  300\_25\_7 &485263 & & $2^{1}$ &484762.6 &485232 &6.4E-3 &0.0 & & $2^{1}$ &484681.2 &485232 &6.4E-3 &0.0 & & $2^{1}$ &484848.1 &485232 &6.4E-3 &0.0 & & $2^{1}$ &484686.2 &485232 &6.4E-3 &0.0 & & $2^{0}$ &484645.6 &485197 & 0.014 &0.0 \\
  300\_25\_8 &  9343 & &$2^{-6}$ &  9343.0 &  9343 & 0.000 &1.0 & &$2^{-6}$ &  9343.0 &  9343 & 0.000 &1.0 & &$2^{-6}$ &  9343.0 &  9343 & 0.000 &1.0 & &$2^{-6}$ &  9343.0 &  9343 & 0.000 &1.0 & &$2^{-6}$ &  9343.0 &  9343 & 0.000 &1.0 \\
  300\_25\_9 &250761 & & $2^{2}$ &250247.4 &250761 & 0.000 &0.1 & &$2^{-1}$ &250709.6 &250761 & 0.000 &0.1 & &$2^{-1}$ &250712.6 &250751 &4.0E-3 &0.0 & & $2^{0}$ &250682.0 &250761 & 0.000 &0.2 & & $2^{0}$ &250648.1 &250761 & 0.000 &0.2 \\
 300\_25\_10 &383377 & &$2^{-6}$ &383377.0 &383377 & 0.000 &1.0 & &$2^{-6}$ &383377.0 &383377 & 0.000 &1.0 & &$2^{-6}$ &383377.0 &383377 & 0.000 &1.0 & &$2^{-6}$ &383377.0 &383377 & 0.000 &1.0 & &$2^{-6}$ &383377.0 &383377 & 0.000 &1.0 \\
  300\_50\_1 &513379 & & $2^{1}$ &513154.1 &513379 & 0.000 &0.3 & &$2^{-6}$ &513361.0 &513361 &3.5E-3 &0.0 & & $2^{3}$ &511778.4 &513379 & 0.000 &0.1 & & $2^{1}$ &513167.2 &513379 & 0.000 &0.2 & & $2^{1}$ &513172.5 &513379 & 0.000 &0.2 \\
  300\_50\_2 &105543 & &$2^{-6}$ &105543.0 &105543 & 0.000 &1.0 & &$2^{-6}$ &105543.0 &105543 & 0.000 &1.0 & &$2^{-6}$ &105543.0 &105543 & 0.000 &1.0 & &$2^{-6}$ &105543.0 &105543 & 0.000 &1.0 & &$2^{-6}$ &105543.0 &105543 & 0.000 &1.0 \\
  300\_50\_3 &875788 & & $2^{1}$ &875017.9 &875788 & 0.000 &0.1 & & $2^{2}$ &874770.7 &875788 & 0.000 &0.1 & & $2^{0}$ &874618.8 &875577 & 0.024 &0.0 & & $2^{1}$ &874958.2 &875788 & 0.000 &0.1 & & $2^{1}$ &874887.9 &875627 & 0.018 &0.0 \\
  300\_50\_4 &307124 & &$2^{-6}$ &307124.0 &307124 & 0.000 &1.0 & &$2^{-6}$ &307124.0 &307124 & 0.000 &1.0 & &$2^{-6}$ &307124.0 &307124 & 0.000 &1.0 & &$2^{-6}$ &307124.0 &307124 & 0.000 &1.0 & &$2^{-6}$ &307124.0 &307124 & 0.000 &1.0 \\
  300\_50\_5 &727820 & & $2^{3}$ &725663.1 &727820 & 0.000 &0.1 & & $2^{1}$ &727586.4 &727820 & 0.000 &0.2 & & $2^{0}$ &727594.3 &727820 & 0.000 &0.1 & & $2^{0}$ &727614.8 &727820 & 0.000 &0.1 & & $2^{1}$ &727431.3 &727820 & 0.000 &0.2 \\
  300\_50\_6 &734053 & & $2^{1}$ &733858.5 &734053 & 0.000 &0.2 & & $2^{1}$ &733956.0 &734053 & 0.000 &0.1 & & $2^{0}$ &733917.8 &734053 & 0.000 &0.1 & & $2^{0}$ &733972.8 &734029 &3.3E-3 &0.0 & & $2^{0}$ &734000.3 &734053 & 0.000 &0.2 \\
  300\_50\_7 & 43595 & & $2^{2}$ & 43523.3 & 43595 & 0.000 &0.1 & & $2^{4}$ & 43464.8 & 43595 & 0.000 &0.3 & & $2^{4}$ & 43502.8 & 43595 & 0.000 &0.2 & & $2^{2}$ & 43523.3 & 43595 & 0.000 &0.1 & & $2^{4}$ & 43358.5 & 43595 & 0.000 &0.1 \\
  300\_50\_8 &767977 & & $2^{0}$ &767772.4 &767977 & 0.000 &0.1 & & $2^{2}$ &767522.0 &767977 & 0.000 &0.1 & & $2^{0}$ &767759.1 &767977 & 0.000 &0.1 & & $2^{1}$ &767688.1 &767960 &2.2E-3 &0.0 & & $2^{1}$ &767671.8 &767977 & 0.000 &0.2 \\
  300\_50\_9 &761351 & &$2^{-6}$ &761351.0 &761351 & 0.000 &1.0 & &$2^{-6}$ &761351.0 &761351 & 0.000 &1.0 & &$2^{-6}$ &761351.0 &761351 & 0.000 &1.0 & &$2^{-6}$ &761351.0 &761351 & 0.000 &1.0 & &$2^{-6}$ &761351.0 &761351 & 0.000 &1.0 \\
 300\_50\_10 &996070 & &$2^{-6}$ &996070.0 &996070 & 0.000 &1.0 & &$2^{-6}$ &996070.0 &996070 & 0.000 &1.0 & &$2^{-6}$ &996070.0 &996070 & 0.000 &1.0 & &$2^{-6}$ &996070.0 &996070 & 0.000 &1.0 & &$2^{-6}$ &996070.0 &996070 & 0.000 &1.0 \\  
    \hline
  \end{tabular}
}
\end{table}
\end{landscape}

\begin{figure*}[t]
     \centering
     \subfloat[$n=100, d=25$]{
         \includegraphics[width=0.31\textwidth]{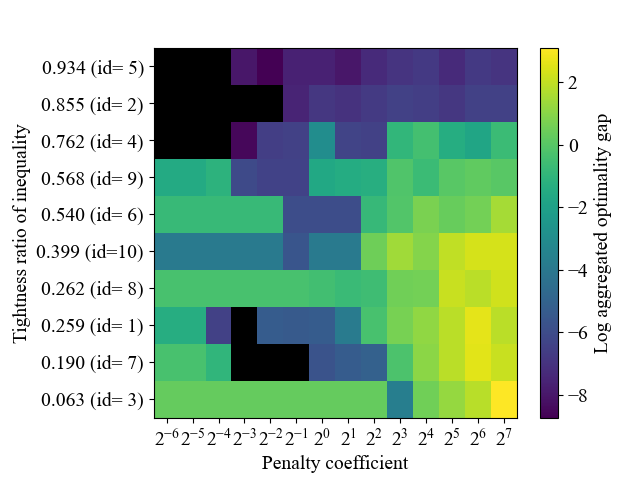}
     }
     \hfil
     \subfloat[$n=100, d=50$]{
         \includegraphics[width=0.31\textwidth]{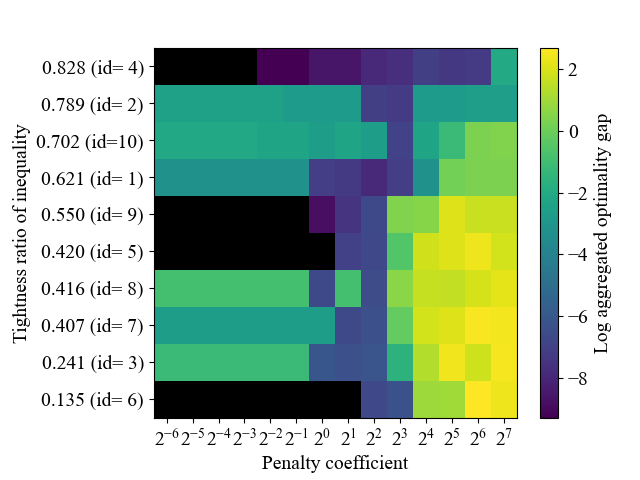}
     }
     \hfil
     \subfloat[$n=100, d=75$]{
         \includegraphics[width=0.31\textwidth]{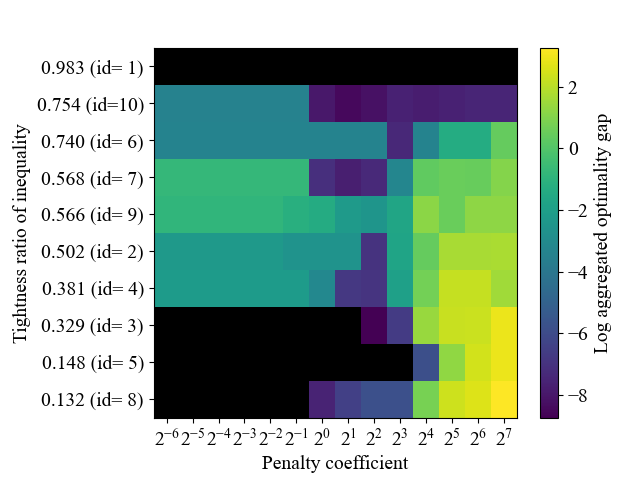}
     }
     \hfil
     \subfloat[$n=100, d=100$]{
         \includegraphics[width=0.31\textwidth]{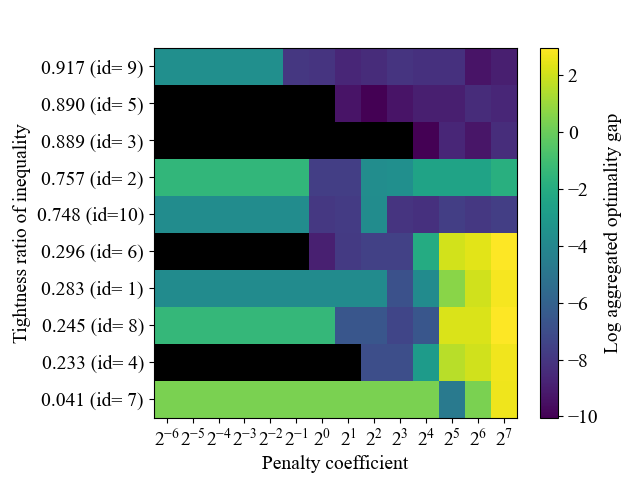}
     }
     \hfil
     \subfloat[$n=200, d=25$]{
         \includegraphics[width=0.31\textwidth]{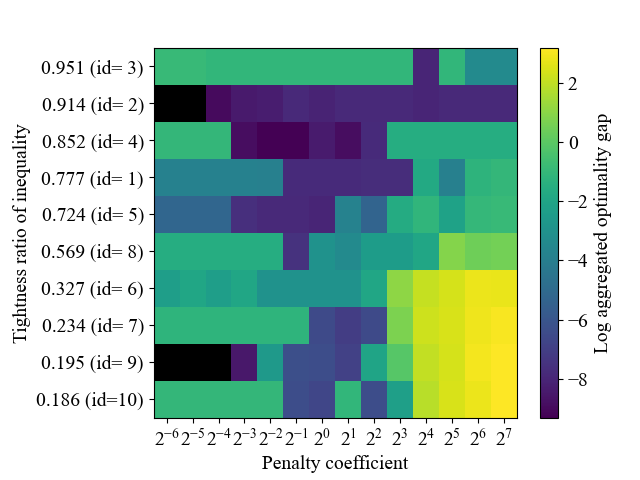}
     }
     \hfil
     \subfloat[$n=200, d=50$]{
         \includegraphics[width=0.31\textwidth]{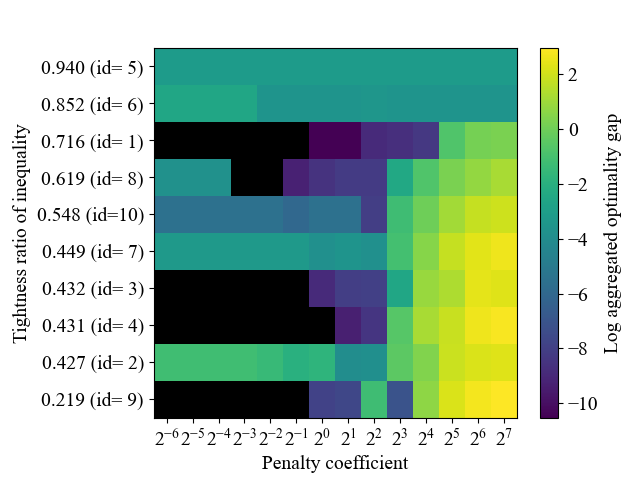}
     }
     \hfil
     \subfloat[$n=200, d=75$]{
         \includegraphics[width=0.31\textwidth]{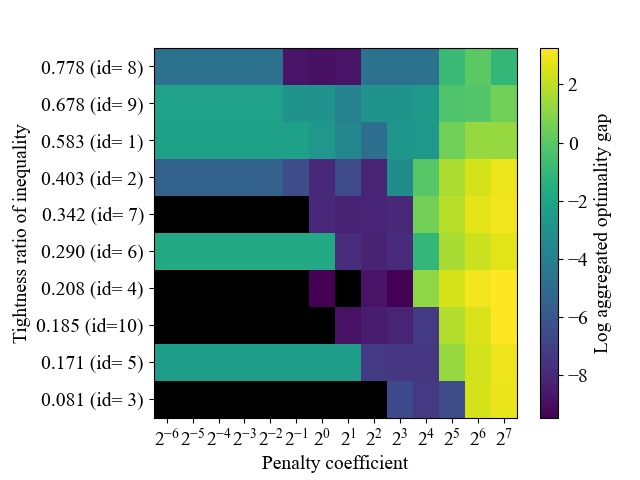}
     }
     \hfil
     \subfloat[$n=200, d=100$]{
         \includegraphics[width=0.31\textwidth]{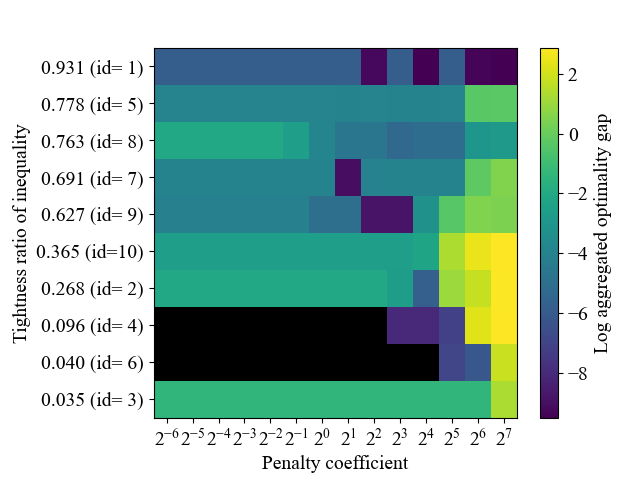}
     }
     \hfil
     \subfloat[$n=300, d=25$]{
         \includegraphics[width=0.31\textwidth]{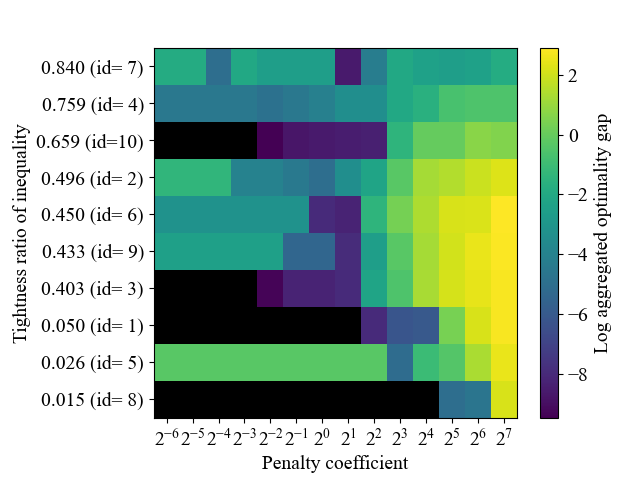}
     }
     \hfil
     \subfloat[$n=300, d=50$]{
         \includegraphics[width=0.31\textwidth]{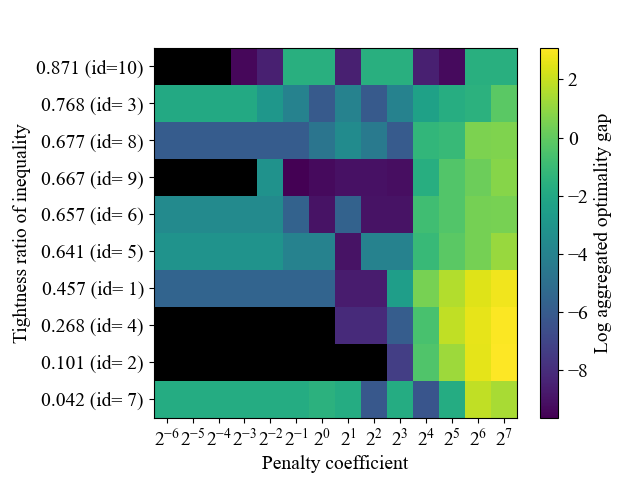}
     }
    \caption{
        Aggregated optimality gap for SA-RI on all 100 medium-sized QKP instances.
        Color bar for gap is shown in log scale. Zero gap (i.e. 100\% success rate) is shown in black. Instances are sorted with tightness ratio $\alpha = C/\sum_i w_i$ for each combination of problem size $n$ and density $d$.
    }
    \label{fig:heatmap_small_sa}
\end{figure*}


\begin{landscape}

\begin{table}[t]
  \caption{Full Results of Ising Machine on Medium-sized Instances.}
  \label{tab:full_ae_1}
  \centering
\resizebox{\linewidth}{!}{
  \begin{tabular}{@{\hspace{1pt}}c@{\hspace{4pt}}r@{\hspace{4pt}}|
  r@{\hspace{6pt}}r
  @{\hspace{6pt}}r@{\hspace{6pt}} 
  r@{\hspace{6pt}}r@{\hspace{6pt}}r@{\hspace{6pt}}r@{\hspace{6pt}}r@{\hspace{6pt}}r 
  @{\hspace{6pt}}r@{\hspace{6pt}} 
  r@{\hspace{6pt}}r@{\hspace{6pt}}r@{\hspace{6pt}}r@{\hspace{6pt}}r
  @{\hspace{6pt}}r@{\hspace{6pt}} 
  r@{\hspace{6pt}}r@{\hspace{6pt}}r@{\hspace{6pt}}r@{\hspace{6pt}}r
  @{\hspace{6pt}}r@{\hspace{6pt}}
  r@{\hspace{6pt}}r@{\hspace{6pt}}r@{\hspace{6pt}}r@{\hspace{6pt}}r
  @{\hspace{3pt}}}
    \hline
    \multicolumn{2}{c|}{Instance} & 
    \multicolumn{2}{c}{Gurobi} &&
    \multicolumn{6}{c}{AE} && 
    \multicolumn{5}{c}{AE-R} && 
    \multicolumn{5}{c}{AE-I} && 
    \multicolumn{5}{c}{AE-RI} \\
    \cline{6-11}
    \cline{13-17}
    \cline{19-23}
    \cline{25-29}
    $n$\_$d$\_id & Optimal & Score & Gap &  & $a$ & FS & Mean & Best & Gap & SR &  & $a$ & Mean & Best & Gap & SR &  & $a$ & Mean & Best & Gap & SR &  & $a$ & Mean & Best & Gap & SR \\
    \hline \hline
  100\_25\_1 &  18558 &  18558 & 0.000 &  &  5 & 1.0 &  18557.70 &  18558 &  0.000 & 0.9 &  &  5 &  18557.70 &  18558 &  0.000 & 0.9 &  &  3 &  18558.00 &  18558 &  0.000 & 1.0 &  &  3 &  18558.00 &  18558 &  0.000 & 1.0 \\
  100\_25\_2 &  56525 &  56525 & 0.000 &  &  2 & 1.0 &  56525.00 &  56525 &  0.000 & 1.0 &  &  2 &  56525.00 &  56525 &  0.000 & 1.0 &  &  2 &  56525.00 &  56525 &  0.000 & 1.0 &  &  2 &  56525.00 &  56525 &  0.000 & 1.0 \\
  100\_25\_3 &   3752 &   3752 & 0.000 &  & 10 & 1.0 &   3752.00 &   3752 &  0.000 & 1.0 &  &  2 &   3752.00 &   3752 &  0.000 & 1.0 &  &  8 &   3752.00 &   3752 &  0.000 & 1.0 &  &  2 &   3752.00 &   3752 &  0.000 & 1.0 \\
  100\_25\_4 &  50382 &  50382 & 0.000 &  &  3 & 1.0 &  50382.00 &  50382 &  0.000 & 1.0 &  &  3 &  50382.00 &  50382 &  0.000 & 1.0 &  &  3 &  50382.00 &  50382 &  0.000 & 1.0 &  &  3 &  50382.00 &  50382 &  0.000 & 1.0 \\
  100\_25\_5 &  61494 &  61494 & 0.000 &  &  4 & 1.0 &  61494.00 &  61494 &  0.000 & 1.0 &  &  4 &  61494.00 &  61494 &  0.000 & 1.0 &  &  2 &  61494.00 &  61494 &  0.000 & 1.0 &  &  1 &  61494.00 &  61494 &  0.000 & 1.0 \\
  100\_25\_6 &  36360 &  36360 & 0.000 &  &  6 & 1.0 &  36360.00 &  36360 &  0.000 & 1.0 &  &  2 &  36360.00 &  36360 &  0.000 & 1.0 &  &  6 &  36360.00 &  36360 &  0.000 & 1.0 &  &  2 &  36360.00 &  36360 &  0.000 & 1.0 \\
  100\_25\_7 &  14657 &  14657 & 0.000 &  &  3 & 1.0 &  14657.00 &  14657 &  0.000 & 1.0 &  &  3 &  14657.00 &  14657 &  0.000 & 1.0 &  &  2 &  14657.00 &  14657 &  0.000 & 1.0 &  &  1 &  14657.00 &  14657 &  0.000 & 1.0 \\
  100\_25\_8 &  20452 &  20452 & 0.000 &  &  3 & 1.0 &  20452.00 &  20452 &  0.000 & 1.0 &  &  3 &  20452.00 &  20452 &  0.000 & 1.0 &  &  3 &  20452.00 &  20452 &  0.000 & 1.0 &  &  3 &  20452.00 &  20452 &  0.000 & 1.0 \\
  100\_25\_9 &  35438 &  35438 & 0.000 &  &  5 & 1.0 &  35438.00 &  35438 &  0.000 & 1.0 &  &  5 &  35438.00 &  35438 &  0.000 & 1.0 &  &  5 &  35438.00 &  35438 &  0.000 & 1.0 &  &  5 &  35438.00 &  35438 &  0.000 & 1.0 \\
 100\_25\_10 &  24930 &  24930 & 0.000 &  &  3 & 1.0 &  24930.00 &  24930 &  0.000 & 1.0 &  &  3 &  24930.00 &  24930 &  0.000 & 1.0 &  &  3 &  24930.00 &  24930 &  0.000 & 1.0 &  &  3 &  24930.00 &  24930 &  0.000 & 1.0 \\
  100\_50\_1 &  83742 &  83742 & 0.000 &  &  3 & 1.0 &  83742.00 &  83742 &  0.000 & 1.0 &  &  3 &  83742.00 &  83742 &  0.000 & 1.0 &  &  3 &  83742.00 &  83742 &  0.000 & 1.0 &  &  2 &  83742.00 &  83742 &  0.000 & 1.0 \\
  100\_50\_2 & 104856 & 104856 & 0.000 &  &  3 & 1.0 & 104856.00 & 104856 &  0.000 & 1.0 &  &  3 & 104856.00 & 104856 &  0.000 & 1.0 &  &  3 & 104856.00 & 104856 &  0.000 & 1.0 &  &  3 & 104856.00 & 104856 &  0.000 & 1.0 \\
  100\_50\_3 &  34006 &  34006 & 0.000 &  &  4 & 1.0 &  34006.00 &  34006 &  0.000 & 1.0 &  &  4 &  34006.00 &  34006 &  0.000 & 1.0 &  &  4 &  34006.00 &  34006 &  0.000 & 1.0 &  &  3 &  34006.00 &  34006 &  0.000 & 1.0 \\
  100\_50\_4 & 105996 & 105996 & 0.000 &  &  2 & 1.0 & 105996.00 & 105996 &  0.000 & 1.0 &  &  2 & 105996.00 & 105996 &  0.000 & 1.0 &  &  2 & 105996.00 & 105996 &  0.000 & 1.0 &  &  1 & 105996.00 & 105996 &  0.000 & 1.0 \\
  100\_50\_5 &  56464 &  56464 & 0.000 &  &  4 & 1.0 &  56464.00 &  56464 &  0.000 & 1.0 &  &  4 &  56464.00 &  56464 &  0.000 & 1.0 &  &  3 &  56464.00 &  56464 &  0.000 & 1.0 &  &  3 &  56464.00 &  56464 &  0.000 & 1.0 \\
  100\_50\_6 &  16083 &  16083 & 0.000 &  &  2 & 1.0 &  16083.00 &  16083 &  0.000 & 1.0 &  &  2 &  16083.00 &  16083 &  0.000 & 1.0 &  &  2 &  16083.00 &  16083 &  0.000 & 1.0 &  &  1 &  16083.00 &  16083 &  0.000 & 1.0 \\
  100\_50\_7 &  52819 &  52819 & 0.000 &  &  2 & 1.0 &  52819.00 &  52819 &  0.000 & 1.0 &  &  2 &  52819.00 &  52819 &  0.000 & 1.0 &  &  2 &  52819.00 &  52819 &  0.000 & 1.0 &  &  2 &  52819.00 &  52819 &  0.000 & 1.0 \\
  100\_50\_8 &  54246 &  54246 & 0.000 &  &  4 & 1.0 &  54246.00 &  54246 &  0.000 & 1.0 &  &  4 &  54246.00 &  54246 &  0.000 & 1.0 &  &  4 &  54246.00 &  54246 &  0.000 & 1.0 &  &  4 &  54246.00 &  54246 &  0.000 & 1.0 \\
  100\_50\_9 &  68974 &  68974 & 0.000 &  &  3 & 1.0 &  68974.00 &  68974 &  0.000 & 1.0 &  &  3 &  68974.00 &  68974 &  0.000 & 1.0 &  &  3 &  68974.00 &  68974 &  0.000 & 1.0 &  &  2 &  68974.00 &  68974 &  0.000 & 1.0 \\
 100\_50\_10 &  88634 &  88634 & 0.000 &  &  9 & 1.0 &  88607.00 &  88634 &  0.000 & 0.5 &  &  9 &  88607.00 &  88634 &  0.000 & 0.5 &  &  8 &  88613.70 &  88634 &  0.000 & 0.7 &  &  8 &  88613.70 &  88634 &  0.000 & 0.7 \\
  100\_75\_1 & 189137 & 189137 & 0.000 &  &  3 & 1.0 & 189089.00 & 189137 &  0.000 & 0.9 &  &  1 & 189137.00 & 189137 &  0.000 & 1.0 &  &  3 & 189137.00 & 189137 &  0.000 & 1.0 &  &  1 & 189137.00 & 189137 &  0.000 & 1.0 \\
  100\_75\_2 &  95074 &  95074 & 0.000 &  &  4 & 1.0 &  95032.90 &  95074 &  0.000 & 0.5 &  &  4 &  95032.90 &  95074 &  0.000 & 0.5 &  &  5 &  95064.80 &  95074 &  0.000 & 0.8 &  &  5 &  95064.80 &  95074 &  0.000 & 0.8 \\
  100\_75\_3 &  62098 &  62098 & 0.000 &  &  2 & 1.0 &  62098.00 &  62098 &  0.000 & 1.0 &  &  2 &  62098.00 &  62098 &  0.000 & 1.0 &  &  2 &  62098.00 &  62098 &  0.000 & 1.0 &  &  1 &  62098.00 &  62098 &  0.000 & 1.0 \\
  100\_75\_4 &  72245 &  72245 & 0.000 &  &  7 & 1.0 &  72232.20 &  72245 &  0.000 & 0.5 &  &  7 &  72232.20 &  72245 &  0.000 & 0.5 &  &  9 &  72244.50 &  72245 &  0.000 & 0.9 &  &  9 &  72244.50 &  72245 &  0.000 & 0.9 \\
  100\_75\_5 &  27616 &  27616 & 0.000 &  &  3 & 1.0 &  27616.00 &  27616 &  0.000 & 1.0 &  &  1 &  27616.00 &  27616 &  0.000 & 1.0 &  &  2 &  27616.00 &  27616 &  0.000 & 1.0 &  &  1 &  27616.00 &  27616 &  0.000 & 1.0 \\
  100\_75\_6 & 145273 & 145273 & 0.000 &  &  2 & 1.0 & 145273.00 & 145273 &  0.000 & 1.0 &  &  2 & 145273.00 & 145273 &  0.000 & 1.0 &  &  2 & 145273.00 & 145273 &  0.000 & 1.0 &  &  2 & 145273.00 & 145273 &  0.000 & 1.0 \\
  100\_75\_7 & 110979 & 110979 & 0.000 &  &  6 & 1.0 & 110960.70 & 110979 &  0.000 & 0.5 &  &  6 & 110960.70 & 110979 &  0.000 & 0.5 &  &  7 & 110977.20 & 110979 &  0.000 & 0.8 &  &  7 & 110977.20 & 110979 &  0.000 & 0.8 \\
  100\_75\_8 &  19570 &  19570 & 0.000 &  &  2 & 1.0 &  19570.00 &  19570 &  0.000 & 1.0 &  &  2 &  19570.00 &  19570 &  0.000 & 1.0 &  &  2 &  19570.00 &  19570 &  0.000 & 1.0 &  &  1 &  19570.00 &  19570 &  0.000 & 1.0 \\
  100\_75\_9 & 104341 & 104341 & 0.000 &  &  8 & 1.0 & 104262.50 & 104341 &  0.000 & 0.3 &  &  8 & 104262.50 & 104341 &  0.000 & 0.3 &  &  6 & 104328.10 & 104341 &  0.000 & 0.8 &  &  6 & 104328.10 & 104341 &  0.000 & 0.8 \\
 100\_75\_10 & 143740 & 143740 & 0.000 &  &  8 & 1.0 & 143702.40 & 143740 &  0.000 & 0.7 &  &  8 & 143702.40 & 143740 &  0.000 & 0.7 &  &  6 & 143740.00 & 143740 &  0.000 & 1.0 &  &  1 & 143740.00 & 143740 &  0.000 & 1.0 \\
 100\_100\_1 &  81978 &  81978 & 0.000 &  &  5 & 1.0 &  81978.00 &  81978 &  0.000 & 1.0 &  &  5 &  81978.00 &  81978 &  0.000 & 1.0 &  &  5 &  81978.00 &  81978 &  0.000 & 1.0 &  &  5 &  81978.00 &  81978 &  0.000 & 1.0 \\
 100\_100\_2 & 190424 & 190424 & 0.000 &  &  4 & 1.0 & 190410.40 & 190424 &  0.000 & 0.8 &  &  4 & 190410.40 & 190424 &  0.000 & 0.8 &  &  4 & 190416.20 & 190424 &  0.000 & 0.8 &  &  4 & 190416.20 & 190424 &  0.000 & 0.8 \\
 100\_100\_3 & 225434 & 225434 & 0.000 &  &  3 & 1.0 & 225412.40 & 225434 &  0.000 & 0.7 &  &  3 & 225412.40 & 225434 &  0.000 & 0.7 &  &  3 & 225434.00 & 225434 &  0.000 & 1.0 &  &  2 & 225434.00 & 225434 &  0.000 & 1.0 \\
 100\_100\_4 &  63028 &  63028 & 0.000 &  &  3 & 1.0 &  63028.00 &  63028 &  0.000 & 1.0 &  &  3 &  63028.00 &  63028 &  0.000 & 1.0 &  &  3 &  63028.00 &  63028 &  0.000 & 1.0 &  &  1 &  63028.00 &  63028 &  0.000 & 1.0 \\
 100\_100\_5 & 230076 & 230076 & 0.000 &  &  3 & 1.0 & 229861.00 & 230076 &  0.000 & 0.7 &  &  3 & 229861.00 & 230076 &  0.000 & 0.7 &  &  3 & 230076.00 & 230076 &  0.000 & 1.0 &  &  1 & 230076.00 & 230076 &  0.000 & 1.0 \\
 100\_100\_6 &  74358 &  74358 & 0.000 &  &  3 & 1.0 &  74358.00 &  74358 &  0.000 & 1.0 &  &  3 &  74358.00 &  74358 &  0.000 & 1.0 &  &  3 &  74358.00 &  74358 &  0.000 & 1.0 &  &  3 &  74358.00 &  74358 &  0.000 & 1.0 \\
 100\_100\_7 &  10330 &  10330 & 0.000 &  &  4 & 1.0 &  10330.00 &  10330 &  0.000 & 1.0 &  &  4 &  10330.00 &  10330 &  0.000 & 1.0 &  &  4 &  10330.00 &  10330 &  0.000 & 1.0 &  &  4 &  10330.00 &  10330 &  0.000 & 1.0 \\
 100\_100\_8 &  62582 &  62582 & 0.000 &  &  3 & 1.0 &  62582.00 &  62582 &  0.000 & 1.0 &  &  3 &  62582.00 &  62582 &  0.000 & 1.0 &  &  3 &  62582.00 &  62582 &  0.000 & 1.0 &  &  1 &  62582.00 &  62582 &  0.000 & 1.0 \\
 100\_100\_9 & 232754 & 232754 & 0.000 &  &  7 & 1.0 & 232270.00 & 232754 &  0.000 & 0.2 &  &  7 & 232270.00 & 232754 &  0.000 & 0.2 &  &  4 & 232754.00 & 232754 &  0.000 & 1.0 &  &  2 & 232754.00 & 232754 &  0.000 & 1.0 \\
100\_100\_10 & 193262 & 193262 & 0.000 &  &  7 & 1.0 & 193246.10 & 193262 &  0.000 & 0.7 &  &  7 & 193246.10 & 193262 &  0.000 & 0.7 &  &  4 & 193262.00 & 193262 &  0.000 & 1.0 &  &  2 & 193262.00 & 193262 &  0.000 & 1.0 \\
  200\_25\_1 & 204441 & 204441 & 0.000 &  &  5 & 1.0 & 200560.60 & 204401 &  0.020 & 0.0 &  &  5 & 204006.70 & 204401 &  0.020 & 0.0 &  &  7 & 204351.60 & 204441 &  0.000 & 0.4 &  &  7 & 204397.30 & 204441 &  0.000 & 0.7 \\
  200\_25\_2 & 239573 & 239573 & 0.000 &  & 10 & 1.0 & 237686.30 & 239573 &  0.000 & 0.1 &  &  3 & 239511.70 & 239573 &  0.000 & 0.2 &  & 10 & 239568.50 & 239573 &  0.000 & 0.7 &  &  3 & 239570.00 & 239573 &  0.000 & 0.8 \\
  200\_25\_3 & 245463 & 245463 & 0.000 &  &  8 & 1.0 & 242063.30 & 245463 &  0.000 & 0.3 &  &  8 & 242063.30 & 245463 &  0.000 & 0.3 &  &  9 & 245338.30 & 245463 &  0.000 & 0.4 &  &  6 & 245119.50 & 245463 &  0.000 & 0.5 \\
  200\_25\_4 & 222361 & 222361 & 0.000 &  &  6 & 1.0 & 220185.30 & 222361 &  0.000 & 0.3 &  &  4 & 222134.60 & 222361 &  0.000 & 0.8 &  &  6 & 222323.90 & 222361 &  0.000 & 0.8 &  &  6 & 222361.00 & 222361 &  0.000 & 1.0 \\
  200\_25\_5 & 187324 & 187324 & 0.000 &  &  8 & 1.0 & 186980.30 & 187316 & 4.3E-3 & 0.0 &  &  4 & 187130.00 & 187324 &  0.000 & 0.1 &  &  7 & 187310.80 & 187324 &  0.000 & 0.3 &  &  5 & 187318.10 & 187324 &  0.000 & 0.3 \\
  200\_25\_6 &  80351 &  80351 & 0.000 &  &  5 & 1.0 &  80227.90 &  80351 &  0.000 & 0.5 &  &  5 &  80271.00 &  80351 &  0.000 & 0.5 &  &  5 &  80309.10 &  80351 &  0.000 & 0.5 &  &  5 &  80312.90 &  80351 &  0.000 & 0.5 \\
  200\_25\_7 &  59036 &  59036 & 0.000 &  &  7 & 1.0 &  59029.20 &  59036 &  0.000 & 0.8 &  &  7 &  59029.20 &  59036 &  0.000 & 0.8 &  &  9 &  59036.00 &  59036 &  0.000 & 1.0 &  &  9 &  59036.00 &  59036 &  0.000 & 1.0 \\
  200\_25\_8 & 149433 & 149433 & 0.000 &  &  9 & 1.0 & 149152.30 & 149407 &  0.017 & 0.0 &  &  9 & 149152.30 & 149407 &  0.017 & 0.0 &  &  8 & 149394.10 & 149433 &  0.000 & 0.4 &  &  8 & 149394.10 & 149433 &  0.000 & 0.4 \\
  200\_25\_9 &  49366 &  49366 & 0.000 &  &  7 & 1.0 &  49363.20 &  49366 &  0.000 & 0.8 &  &  7 &  49363.20 &  49366 &  0.000 & 0.8 &  &  7 &  49366.00 &  49366 &  0.000 & 1.0 &  &  7 &  49366.00 &  49366 &  0.000 & 1.0 \\
 200\_25\_10 &  48459 &  48459 & 0.000 &  &  4 & 1.0 &  48459.00 &  48459 &  0.000 & 1.0 &  &  4 &  48459.00 &  48459 &  0.000 & 1.0 &  &  4 &  48459.00 &  48459 &  0.000 & 1.0 &  &  4 &  48459.00 &  48459 &  0.000 & 1.0 \\
 \hline
  \end{tabular}
}
\end{table}
\end{landscape}

\addtocounter{table}{-1}

\begin{landscape}
\begin{table}[t]
  \caption{Full Results of Ising Machine on Medium-sized Instances (continued).}
  \label{tab:full_ae_2}
  \centering
\resizebox{\linewidth}{!}{
  \begin{tabular}{@{\hspace{1pt}}c@{\hspace{4pt}}r@{\hspace{4pt}}|
  r@{\hspace{6pt}}r
  @{\hspace{6pt}}r@{\hspace{6pt}} 
  r@{\hspace{6pt}}r@{\hspace{6pt}}r@{\hspace{6pt}}r@{\hspace{6pt}}r@{\hspace{6pt}}r 
  @{\hspace{6pt}}r@{\hspace{6pt}} 
  r@{\hspace{6pt}}r@{\hspace{6pt}}r@{\hspace{6pt}}r@{\hspace{6pt}}r
  @{\hspace{6pt}}r@{\hspace{6pt}} 
  r@{\hspace{6pt}}r@{\hspace{6pt}}r@{\hspace{6pt}}r@{\hspace{6pt}}r
  @{\hspace{6pt}}r@{\hspace{6pt}}
  r@{\hspace{6pt}}r@{\hspace{6pt}}r@{\hspace{6pt}}r@{\hspace{6pt}}r
  @{\hspace{3pt}}}
    \hline
    \multicolumn{2}{c|}{Instance} & 
    \multicolumn{2}{c}{Gurobi} &&
    \multicolumn{6}{c}{AE} && 
    \multicolumn{5}{c}{AE-R} && 
    \multicolumn{5}{c}{AE-I} && 
    \multicolumn{5}{c}{AE-RI} \\
    \cline{6-11}
    \cline{13-17}
    \cline{19-23}
    \cline{25-29}
    $n$\_$d$\_id & Optimal & Score & Gap &  & $a$ & FS & Mean & Best & Gap & SR &  & $a$ & Mean & Best & Gap & SR &  & $a$ & Mean & Best & Gap & SR &  & $a$ & Mean & Best & Gap & SR \\
    \hline \hline
  200\_50\_1 & 372097 & 372097 & 0.000 &  &  4 & 1.0 & 372097.00 & 372097 &  0.000 & 1.0 &  &  3 & 372097.00 & 372097 &  0.000 & 1.0 &  &  4 & 372097.00 & 372097 &  0.000 & 1.0 &  &  2 & 372097.00 & 372097 &  0.000 & 1.0 \\
  200\_50\_2 & 211130 & 211130 & 0.000 &  &  4 & 0.6 & 203100.17 & 211122 & 3.8E-3 & 0.0 &  &  4 & 210236.80 & 211122 & 3.8E-3 & 0.0 &  & 10 & 210992.40 & 211130 &  0.000 & 0.1 &  &  3 & 211052.50 & 211130 &  0.000 & 0.1 \\
  200\_50\_3 & 227185 & 227185 & 0.000 &  &  9 & 1.0 & 226925.00 & 227185 &  0.000 & 0.1 &  &  9 & 226925.00 & 227185 &  0.000 & 0.1 &  &  6 & 227167.10 & 227185 &  0.000 & 0.9 &  &  1 & 227185.00 & 227185 &  0.000 & 1.0 \\
  200\_50\_4 & 228572 & 228572 & 0.000 &  &  4 & 1.0 & 228572.00 & 228572 &  0.000 & 1.0 &  &  4 & 228572.00 & 228572 &  0.000 & 1.0 &  &  4 & 228572.00 & 228572 &  0.000 & 1.0 &  &  3 & 228572.00 & 228572 &  0.000 & 1.0 \\
  200\_50\_5 & 479651 & 479651 & 0.000 &  &  9 & 1.0 & 477806.20 & 479651 &  0.000 & 0.4 &  &  9 & 477806.20 & 479651 &  0.000 & 0.4 &  &  9 & 479548.30 & 479651 &  0.000 & 0.5 &  &  3 & 479551.00 & 479651 &  0.000 & 0.5 \\
  200\_50\_6 & 426777 & 426672 & 0.025 &  &  8 & 1.0 & 424845.00 & 426777 &  0.000 & 0.1 &  &  8 & 424845.00 & 426777 &  0.000 & 0.1 &  &  7 & 426750.20 & 426777 &  0.000 & 0.7 &  &  7 & 426757.10 & 426777 &  0.000 & 0.7 \\
  200\_50\_7 & 220890 & 220890 & 0.000 &  &  6 & 1.0 & 220720.90 & 220890 &  0.000 & 0.2 &  &  6 & 220720.90 & 220890 &  0.000 & 0.2 &  &  6 & 220841.30 & 220890 &  0.000 & 0.3 &  &  1 & 220868.40 & 220890 &  0.000 & 0.7 \\
  200\_50\_8 & 317952 & 317952 & 0.000 &  &  5 & 1.0 & 317803.90 & 317952 &  0.000 & 0.2 &  &  5 & 317803.90 & 317952 &  0.000 & 0.2 &  &  5 & 317913.90 & 317952 &  0.000 & 0.7 &  &  2 & 317952.00 & 317952 &  0.000 & 1.0 \\
  200\_50\_9 & 104936 & 104936 & 0.000 &  &  5 & 1.0 & 104742.10 & 104936 &  0.000 & 0.1 &  &  5 & 104742.10 & 104936 &  0.000 & 0.1 &  &  6 & 104884.90 & 104936 &  0.000 & 0.9 &  &  3 & 104904.60 & 104936 &  0.000 & 0.9 \\
 200\_50\_10 & 284751 & 284751 & 0.000 &  &  5 & 1.0 & 284299.40 & 284745 & 2.1E-3 & 0.0 &  &  5 & 284299.40 & 284745 & 2.1E-3 & 0.0 &  &  5 & 284632.00 & 284751 &  0.000 & 0.2 &  &  5 & 284635.70 & 284751 &  0.000 & 0.3 \\
  200\_75\_1 & 442894 & 442894 & 0.000 &  & 10 & 1.0 & 441907.10 & 442672 &  0.050 & 0.0 &  & 10 & 441907.10 & 442672 &  0.050 & 0.0 &  & 10 & 442719.90 & 442894 &  0.000 & 0.3 &  &  3 & 442752.70 & 442894 &  0.000 & 0.5 \\
  200\_75\_2 & 286643 & 286643 & 0.000 &  &  5 & 1.0 & 286572.50 & 286643 &  0.000 & 0.2 &  &  5 & 286572.50 & 286643 &  0.000 & 0.2 &  &  5 & 286630.60 & 286643 &  0.000 & 0.6 &  &  5 & 286630.90 & 286643 &  0.000 & 0.6 \\
  200\_75\_3 &  61924 &  61924 & 0.000 &  & 10 & 1.0 &  61924.00 &  61924 &  0.000 & 1.0 &  & 10 &  61924.00 &  61924 &  0.000 & 1.0 &  &  7 &  61924.00 &  61924 &  0.000 & 1.0 &  &  1 &  61924.00 &  61924 &  0.000 & 1.0 \\
  200\_75\_4 & 128351 & 128351 & 0.000 &  &  3 & 1.0 & 128351.00 & 128351 &  0.000 & 1.0 &  &  3 & 128351.00 & 128351 &  0.000 & 1.0 &  &  3 & 128351.00 & 128351 &  0.000 & 1.0 &  &  1 & 128351.00 & 128351 &  0.000 & 1.0 \\
  200\_75\_5 & 137885 & 137885 & 0.000 &  &  9 & 1.0 & 137842.90 & 137885 &  0.000 & 0.7 &  &  9 & 137842.90 & 137885 &  0.000 & 0.7 &  &  5 & 137885.00 & 137885 &  0.000 & 1.0 &  &  5 & 137885.00 & 137885 &  0.000 & 1.0 \\
  200\_75\_6 & 229631 & 229631 & 0.000 &  &  7 & 1.0 & 228703.10 & 229631 &  0.000 & 0.1 &  &  7 & 228703.10 & 229631 &  0.000 & 0.1 &  &  6 & 229159.00 & 229631 &  0.000 & 0.3 &  &  3 & 229592.90 & 229631 &  0.000 & 0.9 \\
  200\_75\_7 & 269887 & 269887 & 0.000 &  &  5 & 1.0 & 269846.60 & 269887 &  0.000 & 0.6 &  &  5 & 269846.60 & 269887 &  0.000 & 0.6 &  &  6 & 269863.10 & 269887 &  0.000 & 0.9 &  &  6 & 269863.10 & 269887 &  0.000 & 0.9 \\
  200\_75\_8 & 600858 & 600858 & 0.000 &  &  5 & 0.7 & 592861.00 & 600819 & 6.5E-3 & 0.0 &  &  5 & 600507.50 & 600858 &  0.000 & 0.4 &  &  7 & 600155.90 & 600858 &  0.000 & 0.3 &  &  5 & 600858.00 & 600858 &  0.000 & 1.0 \\
  200\_75\_9 & 516771 & 516771 & 0.000 &  & 10 & 1.0 & 516262.00 & 516619 &  0.029 & 0.0 &  & 10 & 516262.00 & 516619 &  0.029 & 0.0 &  &  8 & 516714.90 & 516771 &  0.000 & 0.5 &  &  8 & 516714.90 & 516771 &  0.000 & 0.5 \\
 200\_75\_10 & 142694 & 142694 & 0.000 &  &  6 & 1.0 & 142683.60 & 142694 &  0.000 & 0.9 &  &  6 & 142683.60 & 142694 &  0.000 & 0.9 &  &  4 & 142694.00 & 142694 &  0.000 & 1.0 &  &  1 & 142694.00 & 142694 &  0.000 & 1.0 \\
 200\_100\_1 & 937149 & 937149 & 0.000 &  & 10 & 1.0 & 932240.80 & 936716 &  0.046 & 0.0 &  &  4 & 936329.60 & 936742 &  0.043 & 0.0 &  & 10 & 937143.80 & 937149 &  0.000 & 0.8 &  &  9 & 937146.40 & 937149 &  0.000 & 0.9 \\
 200\_100\_2 & 303058 & 303058 & 0.000 &  & 10 & 1.0 & 302835.40 & 303058 &  0.000 & 0.2 &  & 10 & 302835.40 & 303058 &  0.000 & 0.2 &  & 10 & 303019.90 & 303058 &  0.000 & 0.6 &  & 10 & 303019.90 & 303058 &  0.000 & 0.6 \\
 200\_100\_3 &  29367 &  29367 & 0.000 &  &  4 & 1.0 &  29367.00 &  29367 &  0.000 & 1.0 &  &  4 &  29367.00 &  29367 &  0.000 & 1.0 &  &  4 &  29367.00 &  29367 &  0.000 & 1.0 &  &  4 &  29367.00 &  29367 &  0.000 & 1.0 \\
 200\_100\_4 & 100838 & 100838 & 0.000 &  &  4 & 1.0 & 100838.00 & 100838 &  0.000 & 1.0 &  &  4 & 100838.00 & 100838 &  0.000 & 1.0 &  &  3 & 100838.00 & 100838 &  0.000 & 1.0 &  &  1 & 100838.00 & 100838 &  0.000 & 1.0 \\
 200\_100\_5 & 786635 & 786635 & 0.000 &  &  6 & 1.0 & 786527.90 & 786635 &  0.000 & 0.1 &  &  3 & 786516.60 & 786635 &  0.000 & 0.5 &  &  5 & 786604.80 & 786635 &  0.000 & 0.8 &  &  6 & 786605.20 & 786635 &  0.000 & 0.8 \\
 200\_100\_6 &  41171 &  41171 & 0.000 &  &  3 & 1.0 &  41171.00 &  41171 &  0.000 & 1.0 &  &  3 &  41171.00 &  41171 &  0.000 & 1.0 &  &  3 &  41171.00 &  41171 &  0.000 & 1.0 &  &  1 &  41171.00 &  41171 &  0.000 & 1.0 \\
 200\_100\_7 & 701094 & 701094 & 0.000 &  &  9 & 1.0 & 700500.00 & 701094 &  0.000 & 0.1 &  &  9 & 700500.00 & 701094 &  0.000 & 0.1 &  &  7 & 701088.50 & 701094 &  0.000 & 0.9 &  &  9 & 701088.50 & 701094 &  0.000 & 0.9 \\
 200\_100\_8 & 782443 & 782443 & 0.000 &  & 10 & 1.0 & 781342.10 & 782201 &  0.031 & 0.0 &  & 10 & 781342.10 & 782201 &  0.031 & 0.0 &  &  9 & 782403.30 & 782443 &  0.000 & 0.3 &  &  3 & 782356.30 & 782443 &  0.000 & 0.4 \\
 200\_100\_9 & 628992 & 628992 & 0.000 &  &  6 & 1.0 & 627248.30 & 628948 & 7.0E-3 & 0.0 &  &  6 & 627248.30 & 628948 & 7.0E-3 & 0.0 &  &  4 & 628932.40 & 628992 &  0.000 & 0.5 &  &  4 & 628974.90 & 628992 &  0.000 & 0.5 \\
200\_100\_10 & 378442 & 378442 & 0.000 &  &  8 & 1.0 & 378240.50 & 378442 &  0.000 & 0.1 &  &  8 & 378240.50 & 378442 &  0.000 & 0.1 &  &  7 & 378442.00 & 378442 &  0.000 & 1.0 &  &  7 & 378442.00 & 378442 &  0.000 & 1.0 \\
  300\_25\_1 &  29140 &  29140 & 0.000 &  &  9 & 1.0 &  29129.20 &  29140 &  0.000 & 0.5 &  &  9 &  29129.20 &  29140 &  0.000 & 0.5 &  &  9 &  29140.00 &  29140 &  0.000 & 1.0 &  &  2 &  29140.00 &  29140 &  0.000 & 1.0 \\
  300\_25\_2 & 281990 & 281990 & 0.000 &  & 10 & 1.0 & 281559.20 & 281990 &  0.000 & 0.1 &  & 10 & 281559.20 & 281990 &  0.000 & 0.1 &  &  8 & 281832.50 & 281990 &  0.000 & 0.2 &  &  6 & 281960.30 & 281990 &  0.000 & 0.3 \\
  300\_25\_3 & 231075 & 231075 & 0.000 &  &  9 & 1.0 & 230378.60 & 230873 &  0.087 & 0.0 &  &  4 & 231075.00 & 231075 &  0.000 & 1.0 &  &  8 & 230863.20 & 231075 &  0.000 & 0.3 &  &  1 & 231075.00 & 231075 &  0.000 & 1.0 \\
  300\_25\_4 & 444759 & 444759 & 0.000 &  &  9 & 1.0 & 430978.60 & 444247 &  0.115 & 0.0 &  &  2 & 444401.60 & 444725 & 7.6E-3 & 0.0 &  &  5 & 444430.00 & 444759 &  0.000 & 0.1 &  &  5 & 444709.30 & 444759 &  0.000 & 0.3 \\
  300\_25\_5 &  14988 &  14988 & 0.000 &  &  9 & 1.0 &  14988.00 &  14988 &  0.000 & 1.0 &  &  4 &  14988.00 &  14988 &  0.000 & 1.0 &  &  9 &  14988.00 &  14988 &  0.000 & 1.0 &  &  4 &  14988.00 &  14988 &  0.000 & 1.0 \\
  300\_25\_6 & 269782 & 269782 & 0.000 &  & 10 & 1.0 & 268277.20 & 269601 &  0.067 & 0.0 &  &  2 & 269536.10 & 269782 &  0.000 & 0.9 &  &  8 & 269272.67 & 269782 &  0.000 & 0.1 &  &  2 & 269770.90 & 269782 &  0.000 & 0.9 \\
  300\_25\_7 & 485263 & 485263 & 0.000 &  &  8 & 0.9 & 479607.22 & 484736 &  0.109 & 0.0 &  &  6 & 484404.90 & 485004 &  0.053 & 0.0 &  &  4 & 485143.38 & 485263 &  0.000 & 0.1 &  &  4 & 485125.60 & 485263 &  0.000 & 0.1 \\
  300\_25\_8 &   9343 &   9343 & 0.000 &  & 10 & 1.0 &   9343.00 &   9343 &  0.000 & 1.0 &  & 10 &   9343.00 &   9343 &  0.000 & 1.0 &  &  9 &   9343.00 &   9343 &  0.000 & 1.0 &  &  1 &   9343.00 &   9343 &  0.000 & 1.0 \\
  300\_25\_9 & 250761 & 250761 & 0.000 &  &  9 & 1.0 & 250555.80 & 250751 & 4.0E-3 & 0.0 &  &  9 & 250555.80 & 250751 & 4.0E-3 & 0.0 &  &  9 & 250695.10 & 250761 &  0.000 & 0.2 &  &  7 & 250755.00 & 250761 &  0.000 & 0.4 \\
 300\_25\_10 & 383377 & 383377 & 0.000 &  &  8 & 1.0 & 382810.20 & 383377 &  0.000 & 0.1 &  &  4 & 383377.00 & 383377 &  0.000 & 1.0 &  &  7 & 383263.50 & 383377 &  0.000 & 0.3 &  &  1 & 383377.00 & 383377 &  0.000 & 1.0 \\
  300\_50\_1 & 513379 & 513379 & 0.000 &  & 10 & 1.0 & 513018.60 & 513379 &  0.000 & 0.2 &  & 10 & 513018.60 & 513379 &  0.000 & 0.2 &  & 10 & 513353.00 & 513379 &  0.000 & 0.7 &  &  5 & 513372.30 & 513379 &  0.000 & 0.9 \\
  300\_50\_2 & 105543 & 105543 & 0.000 &  &  6 & 1.0 & 105523.80 & 105543 &  0.000 & 0.9 &  &  4 & 105543.00 & 105543 &  0.000 & 1.0 &  &  6 & 105543.00 & 105543 &  0.000 & 1.0 &  &  3 & 105543.00 & 105543 &  0.000 & 1.0 \\
  300\_50\_3 & 875788 & 875788 & 0.000 &  & 10 & 1.0 & 873068.30 & 874564 &  0.140 & 0.0 &  &  4 & 873757.50 & 875788 &  0.000 & 0.2 &  &  6 & 874591.60 & 875769 & 2.2E-3 & 0.0 &  &  5 & 875684.50 & 875788 &  0.000 & 0.3 \\
  300\_50\_4 & 307124 & 307124 & 0.000 &  & 10 & 1.0 & 306912.10 & 307124 &  0.000 & 0.1 &  &  4 & 305157.30 & 307124 &  0.000 & 0.2 &  &  9 & 307079.70 & 307124 &  0.000 & 0.5 &  &  7 & 307123.30 & 307124 &  0.000 & 0.9 \\
  300\_50\_5 & 727820 & 727820 & 0.000 &  &  9 & 1.0 & 727446.60 & 727820 &  0.000 & 0.1 &  &  5 & 727441.60 & 727820 &  0.000 & 0.2 &  & 10 & 727779.20 & 727820 &  0.000 & 0.7 &  & 10 & 727779.20 & 727820 &  0.000 & 0.7 \\
  300\_50\_6 & 734053 & 734053 & 0.000 &  &  8 & 1.0 & 730756.60 & 734053 &  0.000 & 0.1 &  &  5 & 733876.80 & 734053 &  0.000 & 0.3 &  &  8 & 732651.30 & 734053 &  0.000 & 0.2 &  &  5 & 734048.20 & 734053 &  0.000 & 0.8 \\
  300\_50\_7 &  43595 &  43595 & 0.000 &  &  8 & 1.0 &  43558.60 &  43595 &  0.000 & 0.2 &  &  8 &  43558.60 &  43595 &  0.000 & 0.2 &  &  4 &  43570.90 &  43595 &  0.000 & 0.6 &  &  4 &  43571.80 &  43595 &  0.000 & 0.6 \\
  300\_50\_8 & 767977 & 767977 & 0.000 &  &  8 & 1.0 & 767291.40 & 767937 & 5.2E-3 & 0.0 &  &  4 & 767869.30 & 767948 & 3.8E-3 & 0.0 &  & 10 & 767834.40 & 767977 &  0.000 & 0.1 &  &  3 & 767952.80 & 767977 &  0.000 & 0.4 \\
  300\_50\_9 & 761351 & 761351 & 0.000 &  &  7 & 0.6 & 732640.67 & 761007 &  0.045 & 0.0 &  &  4 & 761153.00 & 761351 &  0.000 & 0.2 &  & 10 & 760476.40 & 761351 &  0.000 & 0.1 &  &  2 & 761351.00 & 761351 &  0.000 & 1.0 \\
 300\_50\_10 & 996070 & 996070 & 0.000 &  &  7 & 1.0 & 977326.60 & 993922 &  0.216 & 0.0 &  &  5 & 994350.00 & 996069 & 1.0E-4 & 0.0 &  &  7 & 995135.60 & 996070 &  0.000 & 0.4 &  &  1 & 995605.10 & 996070 &  0.000 & 0.8 \\
 \hline
 \end{tabular}
}
\end{table}
\end{landscape}

\begin{table*}[t]
  \caption{Full Results of Ising Machine on Large Instances (Greedy, AE, and AE-I).}
  \label{tab:full_ae_large_1}
  \centering
    \resizebox*{!}{0.96\textheight}{
      \begin{tabular}{@{\hspace{3pt}}c@{\hspace{9pt}}r@{\hspace{9pt}}|
  r@{\hspace{9pt}}r
  @{\hspace{9pt}}r@{\hspace{9pt}} 
  r@{\hspace{9pt}}r
  @{\hspace{9pt}}r@{\hspace{9pt}} 
  r@{\hspace{9pt}}r@{\hspace{9pt}}r@{\hspace{9pt}}r@{\hspace{9pt}}r@{\hspace{9pt}}r 
  @{\hspace{9pt}}r@{\hspace{9pt}} 
  r@{\hspace{9pt}}r@{\hspace{9pt}}r@{\hspace{9pt}}r@{\hspace{9pt}}r@{\hspace{9pt}}r 
  @{\hspace{3pt}}}
    \hline
    \multicolumn{2}{c|}{Instance} & 
    \multicolumn{2}{c}{Greedy} &&
    \multicolumn{6}{c}{AE} && 
    \multicolumn{5}{c}{AE-I} \\
    \cline{6-11}
    \cline{13-18}
    $n$\_$d$\_id & Best-known & Score & Gap &  & $a$ & FS & Mean & Best & Gap & SR &  & $a$ & FS & Mean & Best & Gap & SR \\
    \hline \hline
  1000\_25\_1 &   6172407 &   6165313 &  0.115 &  &  8 & 1.0 &  6046565.20 &   6143481 &  0.469 & 0.0 &  &  7 & 1.0 &  6168157.90 &   6172407 &  0.000 & 0.1 \\
  1000\_25\_2 &    229941 &    229941 &  0.000 &  &  9 & 0.1 &   227137.00 &    227137 &  1.219 & 0.0 &  &  9 & 0.1 &   229057.00 &    229057 &  0.384 & 0.0 \\
  1000\_25\_3 &    172418 &    172362 &  0.032 &  &  8 & 0.1 &   118485.00 &    118485 & 31.280 & 0.0 &  &  8 & 0.1 &   137516.00 &    137516 & 20.243 & 0.0 \\
  1000\_25\_4 &    367426 &    367426 &  0.000 &  & - & 0 & - & - & - & - &  & - & 0 & - & - & - & - \\
  1000\_25\_5 &   4885611 &   4884016 &  0.033 &  & 18 & 1.0 &  4794876.00 &   4865871 &  0.404 & 0.0 &  &  7 & 0.9 &  4882635.56 &   4885573 & 7.8E-4 & 0.0 \\
  1000\_25\_6 &     15689 &     15689 &  0.000 &  &  9 & 1.0 &    15689.00 &     15689 &  0.000 & 1.0 &  &  9 & 1.0 &    15689.00 &     15689 &  0.000 & 1.0 \\
  1000\_25\_7 &   4945810 &   4943898 &  0.039 &  &  5 & 0.9 &  4877049.44 &   4908591 &  0.753 & 0.0 &  &  4 & 0.9 &  4944271.78 &   4945810 &  0.000 & 0.1 \\
  1000\_25\_8 &   1710198 &   1709986 &  0.012 &  & - & 0 & - & - & - & - &  & - & 0 & - & - & - & - \\
  1000\_25\_9 &    496315 &    496315 &  0.000 &  & - & 0 & - & - & - & - &  & - & 0 & - & - & - & - \\
 1000\_25\_10 &   1173792 &   1173627 &  0.014 &  & - & 0 & - & - & - & - &  & - & 0 & - & - & - & - \\
  1000\_50\_1 &   5663590 &   5663518 & 1.3E-3 &  & - & 0 & - & - & - & - &  & - & 0 & - & - & - & - \\
  1000\_50\_2 &    180831 &    180725 &  0.059 &  &  6 & 0.1 &   178358.00 &    178358 &  1.368 & 0.0 &  &  6 & 0.1 &   180220.00 &    180220 &  0.338 & 0.0 \\
  1000\_50\_3 &  11384283 &  11384182 & 8.9E-4 &  &  8 & 1.0 & 10942873.80 &  11304391 &  0.702 & 0.0 &  &  3 & 1.0 & 11384044.80 &  11384283 &  0.000 & 0.1 \\
  1000\_50\_4 &    322226 &    322170 &  0.017 &  &  9 & 0.1 &   314771.00 &    314771 &  2.314 & 0.0 &  &  9 & 0.1 &   320179.00 &    320179 &  0.635 & 0.0 \\
  1000\_50\_5 &   9984247 &   9983024 &  0.012 &  & 13 & 1.0 &  9756809.50 &   9932876 &  0.515 & 0.0 &  &  4 & 0.8 &  9981206.50 &   9984002 & 2.5E-3 & 0.0 \\
  1000\_50\_6 &   4106261 &   4105256 &  0.024 &  & - & 0 & - & - & - & - &  & - & 0 & - & - & - & - \\
  1000\_50\_7 &  10498370 &  10497911 & 4.4E-3 &  &  5 & 1.0 & 10293633.60 &  10426694 &  0.683 & 0.0 &  &  2 & 0.9 & 10497958.56 &  10498370 &  0.000 & 0.1 \\
  1000\_50\_8 &   4981146 &   4978776 &  0.048 &  & 17 & 0.1 &  4858893.00 &   4858893 &  2.454 & 0.0 &  & 13 & 0.1 &  4978216.00 &   4978216 &  0.059 & 0.0 \\
  1000\_50\_9 &   1727861 &   1727861 &  0.000 &  & 10 & 0.1 &  1654507.00 &   1654507 &  4.245 & 0.0 &  & 10 & 0.1 &  1682796.00 &   1682796 &  2.608 & 0.0 \\
 1000\_50\_10 &   2340724 &   2340115 &  0.026 &  & 18 & 0.1 &  2332302.00 &   2332302 &  0.360 & 0.0 &  & 18 & 0.1 &  2336272.00 &   2336272 &  0.190 & 0.0 \\
  1000\_75\_1 &  11570056 &  11568107 &  0.017 &  & 14 & 0.7 & 11434541.71 &  11489586 &  0.696 & 0.0 &  & 17 & 0.8 & 11521261.50 &  11569868 & 1.6E-3 & 0.0 \\
  1000\_75\_2 &   1901389 &   1901083 &  0.016 &  & - & 0 & - & - & - & - &  & - & 0 & - & - & - & - \\
  1000\_75\_3 &   2096485 &   2090674 &  0.277 &  & 19 & 0.3 &  1991055.00 &   2086915 &  0.456 & 0.0 &  & 20 & 0.3 &  2084592.33 &   2091367 &  0.244 & 0.0 \\
  1000\_75\_4 &   7305321 &   7305320 & 1.4E-5 &  & - & 0 & - & - & - & - &  & - & 0 & - & - & - & - \\
  1000\_75\_5 &  13970842 &  13967977 &  0.021 &  & 16 & 1.0 & 13730507.30 &  13900932 &  0.500 & 0.0 &  &  8 & 1.0 & 13968979.30 &  13969760 & 7.7E-3 & 0.0 \\
  1000\_75\_6 &  12288738 &  12287677 & 8.6E-3 &  & 10 & 1.0 & 12106698.20 &  12171993 &  0.950 & 0.0 &  & 10 & 1.0 & 12285291.60 &  12288736 & 1.6E-5 & 0.0 \\
  1000\_75\_7 &   1095837 &   1093066 &  0.253 &  & 10 & 0.1 &   528554.00 &    528554 & 51.767 & 0.0 &  & 10 & 0.1 &   579407.00 &    579407 & 47.127 & 0.0 \\
  1000\_75\_8 &   5575813 &   5571863 &  0.071 &  & - & 0 & - & - & - & - &  & - & 0 & - & - & - & - \\
  1000\_75\_9 &    695774 &    695060 &  0.103 &  & 10 & 0.1 &   692459.00 &    692459 &  0.476 & 0.0 &  & 10 & 0.1 &   694689.00 &    694689 &  0.156 & 0.0 \\
 1000\_75\_10 &   2507677 &   2507415 &  0.010 &  & - & 0 & - & - & - & - &  & - & 0 & - & - & - & - \\
 1000\_100\_1 &   6243494 &   6240386 &  0.050 &  & 20 & 0.1 &  6237741.00 &   6237741 &  0.092 & 0.0 &  & 20 & 0.1 &  6241605.00 &   6241605 &  0.030 & 0.0 \\
 1000\_100\_2 &   4854086 &   4851219 &  0.059 &  & 19 & 0.2 &  4026639.50 &   4740591 &  2.338 & 0.0 &  & 19 & 0.2 &  4808609.50 &   4851243 &  0.059 & 0.0 \\
 1000\_100\_3 &   3172022 &   3169717 &  0.073 &  & - & 0 & - & - & - & - &  & - & 0 & - & - & - & - \\
 1000\_100\_4 &    754727 &    754041 &  0.091 &  & 20 & 1.0 &   753275.30 &    754048 &  0.090 & 0.0 &  & 20 & 1.0 &   754459.30 &    754663 & 8.5E-3 & 0.0 \\
 1000\_100\_5 &  18646620 &  18644356 &  0.012 &  & 19 & 1.0 & 18411742.00 &  18553784 &  0.498 & 0.0 &  & 18 & 1.0 & 18612884.10 &  18646307 & 1.7E-3 & 0.0 \\
 1000\_100\_6 &  16020232 &  16019071 & 7.2E-3 &  & 10 & 0.5 & 15832101.20 &  15900553 &  0.747 & 0.0 &  &  8 & 0.6 & 16004784.50 &  16019644 & 3.7E-3 & 0.0 \\
 1000\_100\_7 &  12936205 &  12935892 & 2.4E-3 &  & - & 0 & - & - & - & - &  & - & 0 & - & - & - & - \\
 1000\_100\_8 &   6927738 &   6927088 & 9.4E-3 &  & - & 0 & - & - & - & - &  & - & 0 & - & - & - & - \\
 1000\_100\_9 &   3874959 &   3874666 & 7.6E-3 &  & - & 0 & - & - & - & - &  & - & 0 & - & - & - & - \\
1000\_100\_10 &   1334494 &   1333599 &  0.067 &  & 20 & 0.8 &  1331945.25 &   1332813 &  0.126 & 0.0 &  & 20 & 0.8 &  1333752.88 &   1334390 & 7.8E-3 & 0.0 \\
  2000\_25\_1 &   5268188 &   5268172 & 3.0E-4 &  & 19 & 0.1 &  4264680.00 &   4264680 & 19.048 & 0.0 &  & 19 & 0.1 &  5264179.00 &   5264179 &  0.076 & 0.0 \\
  2000\_25\_2 &  13294030 &  13292220 &  0.014 &  & 18 & 0.6 & 13197215.00 &  13238963 &  0.414 & 0.0 &  & 12 & 0.1 & 13293975.00 &  13293975 & 4.1E-4 & 0.0 \\
  2000\_25\_3 &   5500433 &   5499695 &  0.013 &  & 19 & 0.1 &  3862911.00 &   3862911 & 29.771 & 0.0 &  & 19 & 0.1 &  5492238.00 &   5492238 &  0.149 & 0.0 \\
  2000\_25\_4 &  14625118 &  14624957 & 1.1E-3 &  & 20 & 0.9 & 14438872.67 &  14558124 &  0.458 & 0.0 &  & 10 & 0.7 & 14621883.71 &  14625118 &  0.000 & 0.1 \\
  2000\_25\_5 &   5975751 &   5974429 &  0.022 &  & - & 0 & - & - & - & - &  & - & 0 & - & - & - & - \\
  2000\_25\_6 &   4491691 &   4491649 & 9.4E-4 &  & - & 0 & - & - & - & - &  & - & 0 & - & - & - & - \\
  2000\_25\_7 &   6388756 &   6388705 & 8.0E-4 &  & 20 & 0.1 &  5854140.00 &   5854140 &  8.368 & 0.0 &  & 20 & 0.1 &  6387505.00 &   6387505 &  0.020 & 0.0 \\
  2000\_25\_8 &  11769873 &  11767061 &  0.024 &  & 19 & 0.3 & 11534726.00 &  11722386 &  0.403 & 0.0 &  & 18 & 0.2 & 11768276.50 &  11768330 &  0.013 & 0.0 \\
  2000\_25\_9 &  10960328 &  10960313 & 1.4E-4 &  & 19 & 0.2 & 10577342.50 &  10900905 &  0.542 & 0.0 &  &  8 & 0.1 & 10960113.00 &  10960113 & 2.0E-3 & 0.0 \\
 2000\_25\_10 &    139236 &    139236 &  0.000 &  & 16 & 0.1 &     3945.00 &      3945 & 97.167 & 0.0 &  & 16 & 0.1 &    55528.00 &     55528 & 60.120 & 0.0 \\
  2000\_50\_1 &   7070736 &   7064882 &  0.083 &  & - & 0 & - & - & - & - &  & - & 0 & - & - & - & - \\
  2000\_50\_2 &  12587545 &  12587266 & 2.2E-3 &  & 11 & 0.1 & 12276905.00 &  12276905 &  2.468 & 0.0 &  & 12 & 0.1 & 12585105.00 &  12585105 &  0.019 & 0.0 \\
  2000\_50\_3 &  27268336 &  27268336 &  0.000 &  & 20 & 0.7 & 26632622.14 &  27148430 &  0.440 & 0.0 &  & 11 & 0.2 & 27267716.50 &  27268037 & 1.1E-3 & 0.0 \\
  2000\_50\_4 &  17754434 &  17752803 & 9.2E-3 &  & 15 & 0.3 & 15946463.67 &  17665973 &  0.498 & 0.0 &  & 15 & 0.3 & 17750698.00 &  17752976 & 8.2E-3 & 0.0 \\
  2000\_50\_5 &  16806059 &  16803639 &  0.014 &  & 15 & 0.1 & 16161265.00 &  16161265 &  3.837 & 0.0 &  & 12 & 0.1 & 16802163.00 &  16802163 &  0.023 & 0.0 \\
  2000\_50\_6 &  23076155 &  23074597 & 6.8E-3 &  & 17 & 0.3 & 22841543.00 &  22867634 &  0.904 & 0.0 &  & 17 & 0.3 & 23062541.67 &  23074825 & 5.8E-3 & 0.0 \\
  2000\_50\_7 &  28759759 &  28756239 &  0.012 &  &  9 & 0.2 & 28600532.50 &  28600679 &  0.553 & 0.0 &  &  8 & 0.2 & 28756905.00 &  28757496 & 7.9E-3 & 0.0 \\
  2000\_50\_8 &   1580242 &   1580242 &  0.000 &  & 16 & 0.1 &  1206100.00 &   1206100 & 23.676 & 0.0 &  & 16 & 0.1 &  1282894.00 &   1282894 & 18.817 & 0.0 \\
  2000\_50\_9 &  26523791 &  26523221 & 2.1E-3 &  &  9 & 0.4 & 25837132.25 &  26348755 &  0.660 & 0.0 &  & 11 & 0.4 & 26521898.00 &  26523328 & 1.7E-3 & 0.0 \\
 2000\_50\_10 &  24747047 &  24747047 &  0.000 &  & 20 & 0.2 & 24559204.00 &  24630282 &  0.472 & 0.0 &  & 10 & 0.1 & 24746954.00 &  24746954 & 3.8E-4 & 0.0 \\
  2000\_75\_1 &  25121998 &  25119968 & 8.1E-3 &  & 17 & 0.1 & 24941261.00 &  24941261 &  0.719 & 0.0 &  & 13 & 0.2 & 25094179.50 &  25119908 & 8.3E-3 & 0.0 \\
  2000\_75\_2 &  12664670 &  12664008 & 5.2E-3 &  & - & 0 & - & - & - & - &  & - & 0 & - & - & - & - \\
  2000\_75\_3 &  43943994 &  43941916 & 4.7E-3 &  & 19 & 0.5 & 43299897.20 &  43678829 &  0.603 & 0.0 &  &  8 & 0.2 & 43943590.00 &  43943703 & 6.6E-4 & 0.0 \\
  2000\_75\_4 &  37496613 &  37496099 & 1.4E-3 &  & 20 & 0.4 & 33300362.50 &  37331383 &  0.441 & 0.0 &  &  9 & 0.1 & 37496271.00 &  37496271 & 9.1E-4 & 0.0 \\
  2000\_75\_5 &  24835349 &  24833545 & 7.3E-3 &  & 17 & 0.2 & 20223954.50 &  22728621 &  8.483 & 0.0 &  & 13 & 0.1 & 24832245.00 &  24832245 &  0.012 & 0.0 \\
  2000\_75\_6 &  45137758 &  45137758 &  0.000 &  & 10 & 0.4 & 44662139.25 &  44868910 &  0.596 & 0.0 &  &  5 & 0.2 & 45137345.00 &  45137758 &  0.000 & 0.1 \\
  2000\_75\_7 &  25502608 &  25502409 & 7.8E-4 &  & 15 & 0.1 & 24935633.00 &  24935633 &  2.223 & 0.0 &  & 18 & 0.1 & 25478168.00 &  25478168 &  0.096 & 0.0 \\
  2000\_75\_8 &  10067892 &  10067546 & 3.4E-3 &  & - & 0 & - & - & - & - &  & - & 0 & - & - & - & - \\
  2000\_75\_9 &  14177079 &  14169391 &  0.054 &  & - & 0 & - & - & - & - &  & - & 0 & - & - & - & - \\
 2000\_75\_10 &   7815755 &   7813832 &  0.025 &  & - & 0 & - & - & - & - &  & - & 0 & - & - & - & - \\
 2000\_100\_1 &  37929909 &  37929771 & 3.6E-4 &  & 17 & 0.2 & 35696999.50 &  37675871 &  0.670 & 0.0 &  &  8 & 0.1 & 37927936.00 &  37927936 & 5.2E-3 & 0.0 \\
 2000\_100\_2 &  33665281 &  33639083 &  0.078 &  & 12 & 0.2 & 28996903.50 &  32350573 &  3.905 & 0.0 &  &  8 & 0.2 & 33640008.00 &  33642902 &  0.066 & 0.0 \\
 2000\_100\_3 &  29952019 &  29949832 & 7.3E-3 &  & 11 & 0.1 & 29346552.00 &  29346552 &  2.021 & 0.0 &  & 12 & 0.1 & 29948550.00 &  29948550 &  0.012 & 0.0 \\
 2000\_100\_4 &  26949268 &  26947203 & 7.7E-3 &  & 15 & 0.1 & 26626386.00 &  26626386 &  1.198 & 0.0 &  & 14 & 0.2 & 26946246.00 &  26947244 & 7.5E-3 & 0.0 \\
 2000\_100\_5 &  22041715 &  22038689 &  0.014 &  & 19 & 0.1 & 20086249.00 &  20086249 &  8.872 & 0.0 &  & 19 & 0.1 & 22035826.00 &  22035826 &  0.027 & 0.0 \\
 2000\_100\_6 &  18868887 &  18867393 & 7.9E-3 &  & - & 0 & - & - & - & - &  & - & 0 & - & - & - & - \\
 2000\_100\_7 &  15850597 &  15848026 &  0.016 &  & - & 0 & - & - & - & - &  & - & 0 & - & - & - & - \\
 2000\_100\_8 &  13628967 &  13628029 & 6.9E-3 &  & - & 0 & - & - & - & - &  & - & 0 & - & - & - & - \\
 2000\_100\_9 &   8394562 &   8388596 &  0.071 &  & 19 & 0.1 &  7097454.00 &   7097454 & 15.452 & 0.0 &  & 19 & 0.1 &  7249975.00 &   7249975 & 13.635 & 0.0 \\
2000\_100\_10 &   4923559 &   4921159 &  0.049 &  & 18 & 0.1 &  3111936.00 &   3111936 & 36.795 & 0.0 &  & 15 & 0.1 &  3291205.00 &   3291205 & 33.154 & 0.0 \\
\hline
\end{tabular}
}
\end{table*}

\begin{table*}[t]
  \caption{Full Results of Ising Machine on Large Instances (Gurobi, AE-R, and AE-RI).}
  \label{tab:full_ae_large_2}
  \centering
    \resizebox*{!}{0.96\textheight}{
      \begin{tabular}{@{\hspace{3pt}}c@{\hspace{9pt}}r@{\hspace{9pt}}|
  r@{\hspace{9pt}}r
  @{\hspace{9pt}}r@{\hspace{9pt}} 
  r@{\hspace{9pt}}r
  @{\hspace{9pt}}r@{\hspace{9pt}} 
  r@{\hspace{9pt}}r@{\hspace{9pt}}r@{\hspace{9pt}}r@{\hspace{9pt}}r
  @{\hspace{9pt}}r@{\hspace{9pt}} 
  r@{\hspace{9pt}}r@{\hspace{9pt}}r@{\hspace{9pt}}r@{\hspace{9pt}}r
  @{\hspace{3pt}}}
    \hline
    \multicolumn{2}{c|}{Instance} & 
    \multicolumn{2}{c}{Gurobi} &&
    \multicolumn{5}{c}{AE-R} && 
    \multicolumn{5}{c}{AE-RI} \\
    \cline{6-10}
    \cline{12-16}
    $n$\_$d$\_id & Best-known & Score & Gap &  & $a$ & Mean & Best & Gap & SR &   & $a$ & Mean & Best & Gap & SR \\
    \hline \hline
  1000\_25\_1 &  6172407 &  6172407 & 0.000 &  &  4 &  6129734.5 &  6164169 &   0.133 & 0.0 &  &  7 &  6168834.5 &  6172407 &   0.000 & 0.2 \\
  1000\_25\_2 &   229941 &   229941 & 0.000 &  &  1 &   228355.2 &   229902 &   0.017 & 0.0 &  &  1 &   229933.5 &   229941 &   0.000 & 0.9 \\
  1000\_25\_3 &   172418 &   172418 & 0.000 &  &  6 &   172418.0 &   172418 &   0.000 & 1.0 &  &  6 &   172418.0 &   172418 &   0.000 & 1.0 \\
  1000\_25\_4 &   367426 &   367426 & 0.000 &  &  7 &   365784.3 &   367014 &   0.112 & 0.0 &  &  1 &   367426.0 &   367426 &   0.000 & 1.0 \\
  1000\_25\_5 &  4885611 &  4885611 & 0.000 &  &  6 &  4884277.2 &  4885538 &  1.5E-3 & 0.0 &  &  4 &  4885138.5 &  4885611 &   0.000 & 0.1 \\
  1000\_25\_6 &    15689 &    15689 & 0.000 &  &  5 &    15689.0 &    15689 &   0.000 & 1.0 &  &  1 &    15689.0 &    15689 &   0.000 & 1.0 \\
  1000\_25\_7 &  4945810 &  4945810 & 0.000 &  &  4 &  4943097.0 &  4945810 &   0.000 & 0.1 &  &  4 &  4945608.4 &  4945810 &   0.000 & 0.3 \\
  1000\_25\_8 &  1710198 &  1710132 & 3.9E-3 &  &  7 &  1709246.1 &  1710017 &   0.011 & 0.0 &  & 10 &  1710007.7 &  1710198 &   0.000 & 0.3 \\
  1000\_25\_9 &   496315 &   496315 & 0.000 &  &  6 &   496304.9 &   496315 &   0.000 & 0.9 &  &  5 &   496315.0 &   496315 &   0.000 & 1.0 \\
 1000\_25\_10 &  1173792 &  1173789 & 2.6E-4 &  &  2 &  1172790.6 &  1173694 &  8.3E-3 & 0.0 &  & 10 &  1173639.3 &  1173792 &   0.000 & 0.1 \\
  1000\_50\_1 &  5663590 &  5663590 & 0.000 &  &  6 &  5663158.6 &  5663590 &   0.000 & 0.1 &  &  6 &  5663574.2 &  5663590 &   0.000 & 0.7 \\
  1000\_50\_2 &   180831 &   180831 & 0.000 &  & 10 &   179702.8 &   180831 &   0.000 & 0.2 &  &  5 &   180831.0 &   180831 &   0.000 & 1.0 \\
  1000\_50\_3 & 11384283 & 11384283 & 0.000 &  &  4 & 11354296.1 & 11384247 &  3.2E-4 & 0.0 &  &  5 & 11384256.1 & 11384283 &   0.000 & 0.7 \\
  1000\_50\_4 &   322226 &   322226 & 0.000 &  & 10 &   320732.0 &   322222 &  1.2E-3 & 0.0 &  &  9 &   322220.4 &   322226 &   0.000 & 0.9 \\
  1000\_50\_5 &  9984247 &  9984155 & 9.2E-4 &  &  5 &  9979380.1 &  9983164 &   0.011 & 0.0 &  &  5 &  9983901.4 &  9984155 &  9.2E-4 & 0.0 \\
  1000\_50\_6 &  4106261 &  4106261 & 0.000 &  &  6 &  4104720.9 &  4106084 &  4.3E-3 & 0.0 &  &  4 &  4106142.3 &  4106261 &   0.000 & 0.8 \\
  1000\_50\_7 & 10498370 & 10498370 & 0.000 &  &  6 & 10485633.4 & 10498272 &  9.3E-4 & 0.0 &  &  6 & 10498331.4 & 10498370 &   0.000 & 0.6 \\
  1000\_50\_8 &  4981146 &  4981143 & 6.0E-5 &  & 12 &  4978197.9 &  4979892 &   0.025 & 0.0 &  & 16 &  4979493.7 &  4980274 &   0.018 & 0.0 \\
  1000\_50\_9 &  1727861 &  1727861 & 0.000 &  &  6 &  1724044.1 &  1727861 &   0.000 & 0.6 &  &  6 &  1727861.0 &  1727861 &   0.000 & 1.0 \\
 1000\_50\_10 &  2340724 &  2340724 & 0.000 &  & 16 &  2339432.1 &  2340724 &   0.000 & 0.1 &  & 16 &  2340023.7 &  2340724 &   0.000 & 0.3 \\
  1000\_75\_1 & 11570056 & 11570056 & 0.000 &  & 11 & 11544807.4 & 11568936 &  9.7E-3 & 0.0 &  &  6 & 11569740.8 & 11570055 &  8.6E-6 & 0.0 \\
  1000\_75\_2 &  1901389 &  1901389 & 0.000 &  &  5 &  1894235.8 &  1901231 &  8.3E-3 & 0.0 &  & 10 &  1899475.3 &  1901389 &   0.000 & 0.5 \\
  1000\_75\_3 &  2096485 &  2096485 & 0.000 &  & 16 &  2093204.3 &  2096485 &   0.000 & 0.1 &  & 16 &  2096471.0 &  2096485 &   0.000 & 0.8 \\
  1000\_75\_4 &  7305321 &  7305315 & 8.2E-5 &  &  5 &  7292449.8 &  7304592 & 0.010 & 0.0 &  &  6 &  7305152.0 &  7305321 &   0.000 & 0.5 \\
  1000\_75\_5 & 13970842 & 13970842 & 0.000 &  & 13 & 13941909.1 & 13968310 &   0.018 & 0.0 &  & 10 & 13969208.3 & 13969984 &  6.1E-3 & 0.0 \\
  1000\_75\_6 & 12288738 & 12288738 & 0.000 &  &  6 & 12287280.1 & 12288115 &  5.1E-3 & 0.0 &  &  1 & 12288438.5 & 12288738 &   0.000 & 0.4 \\
  1000\_75\_7 &  1095837 &  1095837 & 0.000 &  &  7 &  1092604.3 &  1093131 &   0.247 & 0.0 &  &  4 &  1093343.1 &  1095837 &   0.000 & 0.1 \\
  1000\_75\_8 &  5575813 &  5575813 & 0.000 &  &  6 &  5574312.1 &  5575811 &  3.6E-5 & 0.0 &  & 15 &  5575505.5 &  5575813 &   0.000 & 0.4 \\
  1000\_75\_9 &   695774 &   695767 & 1.0E-3 &  &  5 &   694039.0 &   695064 &   0.102 & 0.0 &  &  5 &   695496.2 &   695774 &   0.000 & 0.1 \\
 1000\_75\_10 &  2507677 &  2507677 & 0.000 &  &  6 &  2506795.2 &  2507677 &   0.000 & 0.1 &  &  6 &  2507567.5 &  2507677 &   0.000 & 0.5 \\
 1000\_100\_1 &  6243494 &  6243494 & 0.000 &  & 11 &  6238463.1 &  6242478 &   0.016 & 0.0 &  & 11 &  6243143.8 &  6243494 &   0.000 & 0.1 \\
 1000\_100\_2 &  4854086 &  4854086 & 0.000 &  & 14 &  4851616.1 &  4852477 &   0.033 & 0.0 &  & 15 &  4853764.5 &  4854043 &  8.9E-4 & 0.0 \\
 1000\_100\_3 &  3172022 &  3172022 & 0.000 &  &  6 &  3169655.3 &  3170376 &   0.052 & 0.0 &  &  7 &  3171416.8 &  3172022 &   0.000 & 0.5 \\
 1000\_100\_4 &   754727 &   754539 & 0.025 &  & 12 &   754120.8 &   754645 &   0.011 & 0.0 &  & 12 &   754591.7 &   754727 &   0.000 & 0.2 \\
 1000\_100\_5 & 18646620 & 18646607 & 7.0E-5 &  & 11 & 18601616.7 & 18643031 &   0.019 & 0.0 &  &  9 & 18645693.5 & 18646535 &  4.6E-4 & 0.0 \\
 1000\_100\_6 & 16020232 & 16020232 & 0.000 &  &  5 & 16012775.7 & 16016672 &   0.022 & 0.0 &  &  5 & 16019834.7 & 16020232 &   0.000 & 0.3 \\
 1000\_100\_7 & 12936205 & 12936205 & 0.000 &  &  6 & 12935073.5 & 12936083 &  9.4E-4 & 0.0 &  &  5 & 12928785.2 & 12936205 &   0.000 & 0.4 \\
 1000\_100\_8 &  6927738 &  6927671 & 9.7E-4 &  &  6 &  6925298.8 &  6925940 &   0.026 & 0.0 &  &  6 &  6927168.3 &  6927606 &  1.9E-3 & 0.0 \\
 1000\_100\_9 &  3874959 &  3874959 & 0.000 &  &  4 &  3874196.9 &  3874959 &   0.000 & 0.1 &  &  6 &  3874917.8 &  3874959 &   0.000 & 0.8 \\
1000\_100\_10 &  1334494 &  1334494 & 0.000 &  & 13 &  1333903.0 &  1334457 &  2.8E-3 & 0.0 &  & 11 &  1334474.0 &  1334494 &   0.000 & 0.5 \\
  2000\_25\_1 &  5268188 &  5268188 & 0.000 &  &  8 &  5267547.6 &  5268188 &   0.000 & 0.2 &  &  8 &  5268171.6 &  5268188 &   0.000 & 0.4 \\
  2000\_25\_2 & 13294030 & 13293975 & 4.1E-4 &  &  5 & 13281482.0 & 13293956 &  5.6E-4 & 0.0 &  & 12 & 13293730.1 & 13293975 &  4.1E-4 & 0.0 \\
  2000\_25\_3 &  5500433 &  5500411 & 4.0E-4 &  & 13 &  5497914.6 &  5499944 &  8.9E-3 & 0.0 &  &  7 &  5499740.7 &  5500403 &  5.5E-4 & 0.0 \\
  2000\_25\_4 & 14625118 & 14625031 & 5.9E-4 &  & 13 & 14620474.6 & 14624980 &  9.4E-4 & 0.0 &  &  8 & 14625069.6 & 14625118 &   0.000 & 0.3 \\
  2000\_25\_5 &  5975751 &  5975751 & 0.000 &  & 13 &  5973488.1 &  5975675 &  1.3E-3 & 0.0 &  & 13 &  5975276.2 &  5975751 &   0.000 & 0.3 \\
  2000\_25\_6 &  4491691 &  4491635 & 1.2E-3 &  & 12 &  4485267.1 &  4491630 &  1.4E-3 & 0.0 &  & 10 &  4491691.0 &  4491691 &   0.000 & 1.0 \\
  2000\_25\_7 &  6388756 &  6388756 & 0.000 &  & 13 &  6388445.2 &  6388756 &   0.000 & 0.2 &  & 13 &  6388744.1 &  6388756 &   0.000 & 0.7 \\
  2000\_25\_8 & 11769873 & 11769873 & 0.000 &  &  2 & 11757436.7 & 11769873 &   0.000 & 0.1 &  & 11 & 11769820.5 & 11769873 &   0.000 & 0.7 \\
  2000\_25\_9 & 10960328 & 10960263 & 5.9E-4 &  &  3 & 10952474.0 & 10960207 &  1.1E-3 & 0.0 &  & 11 & 10960066.9 & 10960328 &   0.000 & 0.1 \\
 2000\_25\_10 &   139236 &   139236 & 0.000 &  & 16 &   138459.9 &   139236 &   0.000 & 0.2 &  &  1 &   139236.0 &   139236 &   0.000 & 1.0 \\
  2000\_50\_1 &  7070736 &  7070736 & 0.000 &  & 12 &  7062808.2 &  7070341 &  5.6E-3 & 0.0 &  & 12 &  7068396.4 &  7070736 &   0.000 & 0.2 \\
  2000\_50\_2 & 12587545 & 12587545 & 0.000 &  & 12 & 12586350.3 & 12587482 &  5.0E-4 & 0.0 &  & 12 & 12587507.0 & 12587545 &   0.000 & 0.7 \\
  2000\_50\_3 & 27268336 & 27268336 & 0.000 &  &  3 & 27268335.2 & 27268336 &   0.000 & 0.9 &  &  3 & 27268336.0 & 27268336 &   0.000 & 1.0 \\
  2000\_50\_4 & 17754434 & 17754388 & 2.6E-4 &  & 12 & 17752283.1 & 17753673 &  4.3E-3 & 0.0 &  & 12 & 17754087.3 & 17754434 &   0.000 & 0.1 \\
  2000\_50\_5 & 16806059 & 16805435 & 3.7E-3 &  & 11 & 16801525.3 & 16804057 &   0.012 & 0.0 &  & 12 & 16804680.5 & 16805490 &  3.4E-3 & 0.0 \\
  2000\_50\_6 & 23076155 & 23076097 & 2.5E-4 &  & 12 & 23074405.9 & 23075875 &  1.2E-3 & 0.0 &  & 11 & 23075852.7 & 23076155 &   0.000 & 0.2 \\
  2000\_50\_7 & 28759759 & 28759759 & 0.000 &  & 11 & 27410959.4 & 28756657 &   0.011 & 0.0 &  & 10 & 28757135.3 & 28757834 &  6.7E-3 & 0.0 \\
  2000\_50\_8 &  1580242 &  1580242 & 0.000 &  & 11 &  1577163.4 &  1580242 &   0.000 & 0.3 &  &  1 &  1580242.0 &  1580242 &   0.000 & 1.0 \\
  2000\_50\_9 & 26523791 & 26523791 & 0.000 &  &  8 & 26519303.3 & 26523462 &  1.2E-3 & 0.0 &  &  7 & 26523472.8 & 26523791 &   0.000 & 0.2 \\
 2000\_50\_10 & 24747047 & 24747047 & 0.000 &  &  8 & 24743522.7 & 24746936 &  4.5E-4 & 0.0 &  &  1 & 24747047.0 & 24747047 &   0.000 & 1.0 \\
  2000\_75\_1 & 25121998 & 25121998 & 0.000 &  & 11 & 25120313.9 & 25121457 &  2.2E-3 & 0.0 &  & 11 & 25121744.4 & 25121998 &   0.000 & 0.4 \\
  2000\_75\_2 & 12664670 & 12664670 & 0.000 &  & 11 & 12656539.5 & 12664244 &  3.4E-3 & 0.0 &  & 10 & 12662072.5 & 12664670 &   0.000 & 0.4 \\
  2000\_75\_3 & 43943994 & 43943994 & 0.000 &  & 16 & 43800921.7 & 43943366 &  1.4E-3 & 0.0 &  &  8 & 43943724.9 & 43943994 &   0.000 & 0.6 \\
  2000\_75\_4 & 37496613 & 37496613 & 0.000 &  &  3 & 37493330.7 & 37496308 &  8.1E-4 & 0.0 &  &  2 & 37496367.3 & 37496613 &   0.000 & 0.2 \\
  2000\_75\_5 & 24835349 & 24835349 & 0.000 &  & 11 & 24828628.7 & 24831592 &   0.015 & 0.0 &  & 11 & 24834632.9 & 24834948 &  1.6E-3 & 0.0 \\
  2000\_75\_6 & 45137758 & 45137758 & 0.000 &  &  3 & 45132276.3 & 45137758 &   0.000 & 0.9 &  &  8 & 45137758.0 & 45137758 &   0.000 & 1.0 \\
  2000\_75\_7 & 25502608 & 25502608 & 0.000 &  & 12 & 25467707.7 & 25502608 &   0.000 & 0.1 &  &  4 & 25495828.7 & 25502608 &   0.000 & 0.3 \\
  2000\_75\_8 & 10067892 & 10067892 & 0.000 &  & 11 & 10064140.4 & 10067750 &  1.4E-3 & 0.0 &  & 10 & 10063814.3 & 10067892 &   0.000 & 0.3 \\
  2000\_75\_9 & 14177079 & 14177079 & 0.000 &  & 11 & 14163838.4 & 14168633 &   0.060 & 0.0 &  & 14 & 14171401.8 & 14171994 &   0.036 & 0.0 \\
 2000\_75\_10 &  7815755 &  7815334 & 5.4E-3 &  & 12 &  7811373.0 &  7812703 &   0.039 & 0.0 &  & 20 &  7814965.3 &  7815611 &  1.8E-3 & 0.0 \\
 2000\_100\_1 & 37929909 & 37929909 & 0.000 &  &  4 & 37926908.4 & 37929909 &   0.000 & 0.1 &  & 12 & 37929518.3 & 37929909 &   0.000 & 0.6 \\
 2000\_100\_2 & 33665281 & 33665281 & 0.000 &  & 12 & 32835326.4 & 33637892 &   0.081 & 0.0 &  & 12 & 33644437.0 & 33646541 &   0.056 & 0.0 \\
 2000\_100\_3 & 29952019 & 29951413 & 2.0E-3 &  & 12 & 29553629.4 & 29948391 &   0.012 & 0.0 &  & 10 & 29951436.0 & 29952019 &   0.000 & 0.1 \\
 2000\_100\_4 & 26949268 & 26948616 & 2.4E-3 &  & 13 & 26943483.1 & 26947024 &  8.3E-3 & 0.0 &  & 11 & 26948996.6 & 26949268 &   0.000 & 0.2 \\
 2000\_100\_5 & 22041715 & 22041314 & 1.8E-3 &  & 11 & 22032284.4 & 22034438 &   0.033 & 0.0 &  & 17 & 22038647.2 & 22041221 &  2.2E-3 & 0.0 \\
 2000\_100\_6 & 18868887 & 18868887 & 0.000 &  & 10 & 18834013.3 & 18868626 &  1.4E-3 & 0.0 &  & 10 & 18866428.5 & 18868887 &   0.000 & 0.2 \\
 2000\_100\_7 & 15850597 & 15850597 & 0.000 &  & 10 & 15847566.7 & 15848960 &   0.010 & 0.0 &  & 11 & 15850259.2 & 15850594 &  1.9E-5 & 0.0 \\
 2000\_100\_8 & 13628967 & 13628967 & 0.000 &  & 12 & 13622163.4 & 13626547 &   0.018 & 0.0 &  & 12 & 13628850.7 & 13628967 &   0.000 & 0.4 \\
 2000\_100\_9 &  8394562 &  8394101 & 5.5E-3 &  & 11 &  8388749.0 &  8390723 &   0.046 & 0.0 &  & 20 &  8394196.6 &  8394562 &   0.000 & 0.4 \\
2000\_100\_10 &  4923559 &  4923387 & 3.5E-3 &  & 13 &  4918213.2 &  4919752 &   0.077 & 0.0 &  & 14 &  4922601.9 &  4923470 &  1.8E-3 & 0.0 \\
\hline
\end{tabular}
}
\end{table*}

\begin{figure*}[t]
     \centering
     \subfloat[$n=1000, d=25$]{
         \includegraphics[width=0.36\textwidth]{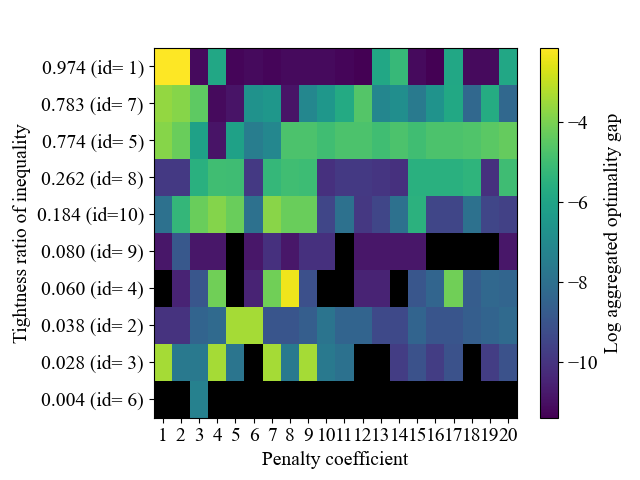}
     }
     \hfil
     \subfloat[$n=1000, d=50$]{
         \includegraphics[width=0.36\textwidth]{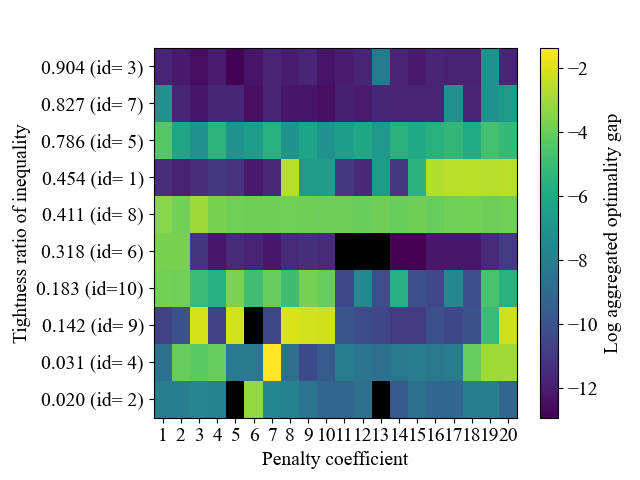}
     }
     \hfil
     \subfloat[$n=1000, d=75$]{
         \includegraphics[width=0.36\textwidth]{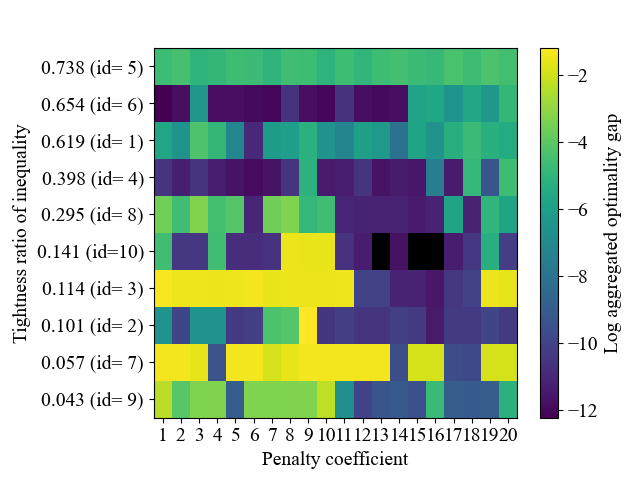}
     }
     \hfil
     \subfloat[$n=1000, d=100$]{
         \includegraphics[width=0.36\textwidth]{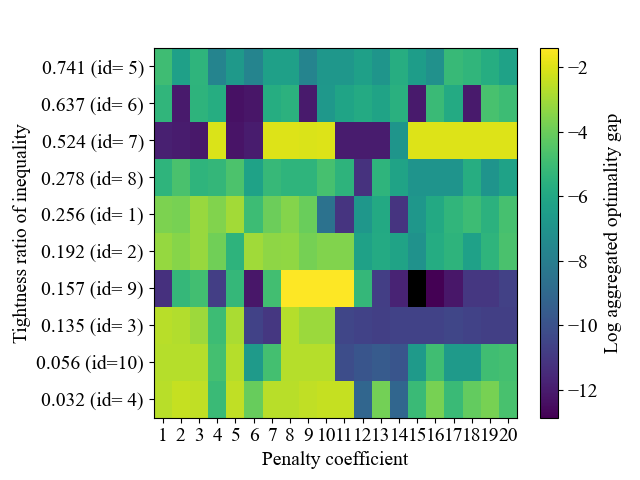}
     }
     \hfil
     \subfloat[$n=2000, d=25$]{
         \includegraphics[width=0.36\textwidth]{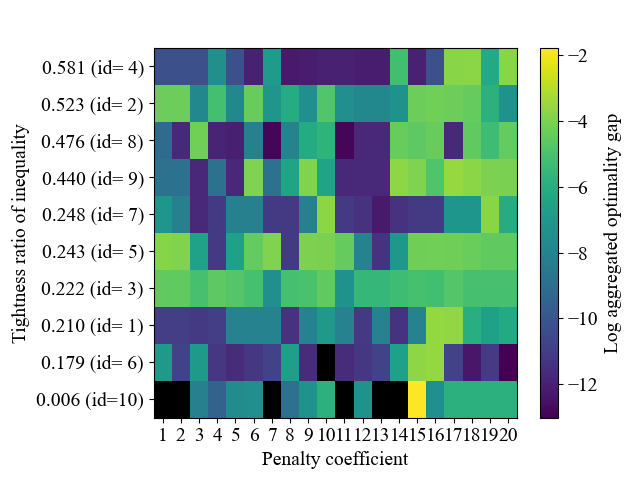}
     }
     \hfil
     \subfloat[$n=2000, d=50$]{
         \includegraphics[width=0.36\textwidth]{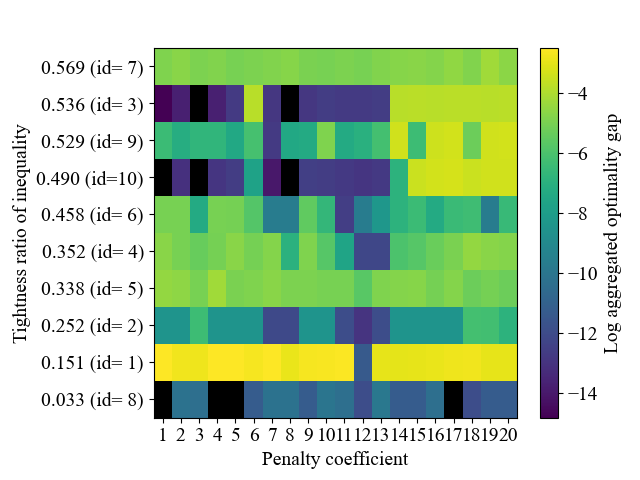}
     }
     \hfil
     \subfloat[$n=2000, d=75$]{
         \includegraphics[width=0.36\textwidth]{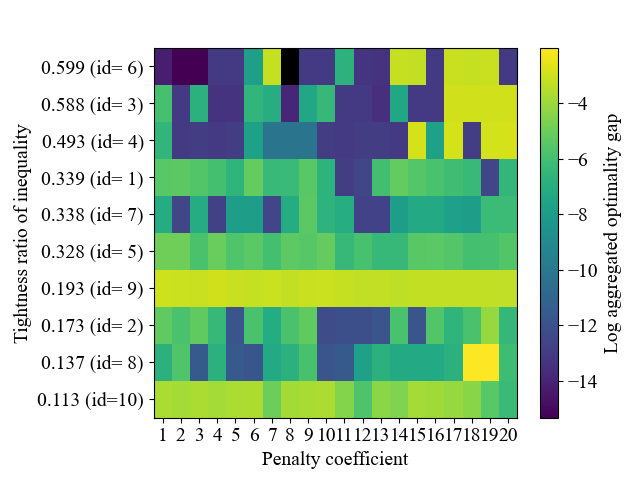}
     }
     \hfil
     \subfloat[$n=2000, d=100$]{
         \includegraphics[width=0.36\textwidth]{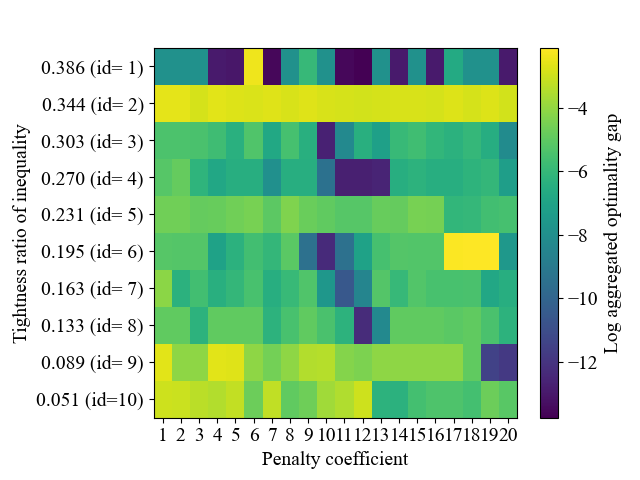}
     }
    \caption{
        Aggregated optimality gap for AE-RI on all 80 large QKP instances.
        Color bar for gap is shown in log scale. Zero gap (i.e. 100\% success rate) is shown in black. Instances are sorted with tightness ratio $\alpha = C/\sum_i w_i$ for each combination of problem size $n$ and density $d$.
        Note that penalty coefficient is rescaled and x-axis actually denotes $a$ in Eq.~(\ref{eq:penalty_scale}).
        Due to this rescaling, we see less trends in distribution of good penalty coefficient than Fig.~\ref{fig:heatmap_small_sa}.
    }
    \label{fig:heatmap_large_ae}
\end{figure*}

\subsection{Full Results on Ising Machine}\label{app:ae_full_results}

Full results of the benchmark of AE conducted in Section~\ref{sec:experiment} are shown in Table~\ref{tab:full_ae_1}, \ref{tab:full_ae_large_1} and \ref{tab:full_ae_large_2}.
Legends for columns are the same as those in Table~\ref{tab:full_sa_1} except for the optimal penalty coefficient $\lambda$.
As we rescaled $\lambda$ as in Eq.~(\ref{eq:penalty_scale}), the value of $\lambda$ is defined according to the value of $a$ in Eq.~(\ref{eq:penalty_scale}).
Therefore, we report the value of $a$ giving the optimal $\lambda$.
Since there are several large instances on which AE cannot obtain feasible a solution even with $a=20$, the results on those instances are not reported.

We also plot the aggregated optimality gap for AE-RI on each instance in Fig.~\ref{fig:heatmap_large_ae}.
We observe that the tested range of penalty coefficients seems to cover optimal coefficients on most instances.
Note that the x-axis corresponds to values of $a$ in Eq.~(\ref{eq:penalty_scale}), not $\lambda$.
Due to the rescaling of $\lambda$, we see less visual trends in Figure~\ref{fig:heatmap_large_ae} compared to Fig.~\ref{fig:heatmap_small_sa}.
This indicates that the rescaling based on SA analysis also works well when using the Ising machine for large-scale instances.

\EOD
\end{document}